\documentclass[conference]{IEEEtran}
\IEEEoverridecommandlockouts

\usepackage{amsfonts}
\usepackage{amsthm}
\usepackage{arydshln} 
\usepackage{booktabs}
\usepackage[hypcap = false]{caption}
\usepackage{cite}
\usepackage{color}
\usepackage{cuted}
\usepackage{enumerate}
\let\labelindent\relax 
\usepackage{enumitem}
\usepackage{float}
\usepackage{graphicx}
\usepackage[backref = page]{hyperref}
\hypersetup{
	colorlinks,
	linkcolor=blue, 
	citecolor    = red
}
\usepackage{makecell}
\usepackage{mathtools}
\usepackage{multirow}
\usepackage{multicol}
\usepackage{nicematrix}
\usepackage{pgf}
\usepackage{pgfplots}
\usepackage{pgfplotstable}
\usepackage{pifont}
\usepackage{ragged2e}
\usepackage{stmaryrd}
\usepackage{stfloats}
\usepackage{subcaption}
\usepackage{tabu}
\usepackage{tabularx}
\usepackage{threeparttable}
\usepackage{tikz}
\usepackage{url}
\usepackage{verbatim}
\usepackage{wasysym}
\usepackage{wrapfig}
\usepackage{xcolor}
\usepackage{xspace}
\usepackage[english]{babel}
\usepackage [autostyle, english = american]{csquotes}
\MakeOuterQuote{"}
\usepackage[framemethod=tikz]{mdframed} 
\usepackage{algpseudocode}
\usepackage{nccmath}

\definecolor{UniBlau}{cmyk}{1,0.7,0,0}
\definecolor{UniGruen}{cmyk}{0.6,0,1,0}
\definecolor{UniOrange}{cmyk}{0,0.3,1,0}
\definecolor{UniRot}{cmyk}{0.4,1,0,0}
\definecolor{darkred}{rgb}{.6,0,0}
\definecolor{darkgreen}{rgb}{0,.4,0}
\definecolor{darkblue}{rgb}{0,0,.6}

\definecolor{alabaster}{rgb}{0.93,0.92,0.88}
\definecolor{antiquewhite}{rgb}{0.98,0.92,0.84}
\definecolor{beige}{rgb}{0.96,0.96,0.86}
\definecolor{burlywood}{rgb}{0.87,0.72,0.53}
\definecolor{camel}{rgb}{0.76,0.6,0.42}
\definecolor{darkkhaki}{rgb}{0.74,0.72,0.42}

\newcommand{\negl}{\ensuremath{\mathsf{negl}}}
\newcommand{\csec}{\kappa}

\newcommand{\abort}{\ensuremath{\mathtt{abort}}}

\newcommand{\continue}{\ensuremath{\mathtt{continue}}}
\newcommand{\flag}{\ensuremath{\mathsf{flag}}}

\newcommand{\Partyset}{\ensuremath{\mathcal{P}}}

\newcommand{\bitb}{\ensuremath{\mathsf{b}}} 
\newcommand{\isTr}{\ensuremath{\mathsf{isTr}}}

\newcommand{\Hash}{\ensuremath{\mathsf{H}}}

\newcommand{\SELECT}{\ensuremath{\mathsf{select}}}
\newcommand{\INPUT}{\ensuremath{\mathsf{Input}}}
\newcommand{\OUTPUT}{\ensuremath{\mathsf{Output}}}

\newcommand{\Adv}{\ensuremath{\mathcal{A}}}
\newcommand{\Sim}{\ensuremath{\mathcal{S}}}

\newtheorem{theorem}{Theorem}[section]

\newtheorem{lemma}[theorem]{Lemma}

\newtheorem{notation}[theorem]{Notation}

\newcommand{\xor}{\oplus}

\newcommand{\band}{\odot}

\newcommand{\cmark}{\ding{51}}
\newcommand{\xmark}{\ding{55}}

\newcommand{\iseq}{\ensuremath{\stackrel{?}{=} }}

\newcommand{\bigominus}{\mathop{\raisebox{-.05em}{\large\boldmath$\ominus$}}}
\newcommand{\MatMul}{\bigodot}

\newcommand{\bitset}{\{0,1\}}
\newcommand{\Prob}{\ensuremath{\mathsf{Pr}}}
\newcommand{\msb}{\ensuremath{\mathsf{msb}}}
\newcommand{\lsb}{\ensuremath{\mathsf{lsb}}}

\newcommand{\Z}[1]{\ensuremath{\mathbb{Z}}_{2^{#1}}}
\newcommand{\F}{\ensuremath{\mathbb{F}}}

\newcommand{\dotp}{\ensuremath{\mathsf{dotp}}}
\newcommand{\bitA}{\ensuremath{\mathsf{bit2A}}}
\newcommand{\bitinj}{\ensuremath{\mathsf{bitInj}}}
\newcommand{\piecewise}{\ensuremath{\mathsf{piecewise}}}
\newcommand{\bitext}{\ensuremath{\mathsf{bitext}}}
\newcommand{\relu}{\ensuremath{\mathsf{ReLU}}}
\newcommand{\sig}{\ensuremath{\mathsf{Sig}}}
\newcommand{\sftmx}{\ensuremath{\mathsf{softmax}}}

\newcommand{\prot}[1]{\ensuremath{\Pi_{#1}}}
\newcommand{\protB}[1]{\ensuremath{\Pi_{#1}^{\bf B}}}

\setlist[description]{style=unboxed,leftmargin=0cm}

\setlist[enumerate]{itemsep=0mm}

\newenvironment{myitemize}{
	\begin{list}{{$\bullet$}}{
			\setlength\partopsep{0pt}
			\setlength\parskip{0pt}
			\setlength\parsep{0pt}
			\setlength\topsep{2pt}
			\setlength\itemsep{2pt}
			\setlength{\itemindent}{8pt}
			\setlength{\leftmargin}{1pt}
		}
	}{
	\end{list}
}

\newenvironment{itemndss}{
	\begin{list}{{$\bullet$}}{
			\setlength\partopsep{0pt}
			\setlength\parskip{0pt}
			\setlength\parsep{2pt}
			\setlength\topsep{4pt}
			\setlength\itemsep{4pt}
			\setlength{\itemindent}{8pt}
			\setlength{\leftmargin}{3pt}
		}
	}{
	\end{list}
}

\newcommand{\tabref}[1]{Table~\protect\ref{tab:#1}}
\newcommand{\secref}[1]{~\S\protect\ref{sec:#1}}

\newcommand{\figref}[1]{Figure~\ref{fig:#1}}
\newcommand{\figlab}[1]{\label{fig:#1}}

\newenvironment{boxfig*}[2]{
	\begin{figure*}[h!]		
		\fontsize{5}{5}\selectfont
		\newcommand{\FigCaption}{#1}
		\newcommand{\FigLabel}{#2}
		\vspace{-.05cm}
		\begin{center}
			\begin{small}			 
				\begin{adjustbox}{max width=\textwidth}
					\begin{tabular}{@{}|@{~~}l@{~~}|@{}}
						\hline
						\rule[-1ex]{0pt}{1ex}\begin{minipage}[b]{.95\linewidth}
							\vspace{1ex}	
						}{%
						\end{minipage}\\
						\hline
					\end{tabular}	
				\end{adjustbox}		
			\end{small}
			\vspace{-0.1cm}
			\caption{\FigCaption}
			\figlab{\FigLabel}
		\end{center}
		\vspace{-.38cm}
	\end{figure*}
}

\newenvironment{myboxfig*}[2]{
	\begin{figure*}[!htb]		
		\fontsize{5}{5}\selectfont
		\newcommand{\FigCaption}{#1}
		\newcommand{\FigLabel}{#2}
		\vspace{-.10cm}
		\begin{center}
			\caption{\FigCaption}
			\begin{small}			 
				\begin{adjustbox}{max width=\textwidth}
					\begin{tabular}{@{}|@{~~}l@{~~}|@{}}
						\hline
						\rule[-1ex]{0pt}{1ex}\begin{minipage}[b]{.95\linewidth}
							\vspace{1ex}	
						}{%
						\end{minipage}\\
						\hline
					\end{tabular}	
				\end{adjustbox}		
			\end{small}
			\vspace{-0.25cm}
			\figlab{\FigLabel}
		\end{center}
		\vspace{-.38cm}
	\end{figure*}
}

\newcommand{\boxref}[1]{Fig.~\ref{#1}}

\RequirePackage{mdframed}

\newenvironment{titlebox}[5]
{\mdfsetup{
		style=#2,
		innertopmargin=1.1\baselineskip,
		skipabove={\dimexpr0.2\baselineskip+\topskip\relax},
		skipbelow={1em},needspace=3\baselineskip,
		singleextra={\node[#3,right=10pt,overlay] at (P-|O){~{\sffamily\bfseries #1 }};},%
		firstextra={\node[#3,right=10pt,overlay] at (P-|O) {~{\sffamily\bfseries #1 }};},
		frametitleaboveskip=9em,
		innerrightmargin=5pt
	}
	\newcommand{\TitleCaption}{#4}
	\newcommand{\TitleLabel}{#5}
	\begin{mdframed}[font=\small]
		\setlist[itemize]{leftmargin=13pt}\setlist[enumerate]{leftmargin=13pt}\raggedright%
	}
	{\end{mdframed}
	\vspace{-1.5em}
	{\captionof{figure}{\small \TitleCaption}\label{\TitleLabel}}
	\medskip
}

\tikzstyle{normal} = [thick, fill=white, text=black, draw, rounded corners, rectangle, minimum height=.7cm, inner sep=3pt]
\tikzstyle{gray} = [thick, fill=gray!90, text=white, rounded corners, rectangle, minimum height=.7cm, inner sep=3pt]

\mdfdefinestyle{commonbox}{%
	align=center, middlelinewidth=1.1pt,userdefinedwidth=\linewidth
	innerrightmargin=0pt,innerleftmargin=2pt,innertopmargin=5pt,
	splittopskip=7pt,splitbottomskip=3pt
}

\mdfdefinestyle{roundbox}{style=commonbox,roundcorner=5pt,userdefinedwidth=\linewidth}

\newenvironment{systembox}[3]
{\vspace{\baselineskip}\begin{titlebox}{Functionality \normalfont #1}{roundbox}{normal}{#2}{#3}}
	{\end{titlebox}}

\newenvironment{protocolbox}[3]
{\begin{titlebox}{Protocol \normalfont #1}{commonbox}{normal}{#2}{#3}}
	{\end{titlebox}}

\newenvironment{simulatorbox}[3]
{\begin{titlebox}{Simulator \normalfont #1}{commonbox}{normal}{#2}{#3}}
	{\end{titlebox}}

\newenvironment{splittitlebox}[5]
{\mdfsetup{
		style=#2,
		innertopmargin=1.1\baselineskip,
		skipabove={\dimexpr0.2\baselineskip+\topskip\relax},
		skipbelow={1em},needspace=3\baselineskip,
		singleextra={\node[#3,right=10pt,overlay] at (P-|O){~{\sffamily\bfseries #1 }};},%
		firstextra={\node[#3,right=10pt,overlay] at (P-|O) {~{\sffamily\bfseries #1 }};},
		frametitleaboveskip=9em,
		innerrightmargin=5pt
	}
	\newcommand{\TitleCaption}{#4}
	\newcommand{\TitleLabel}{#5}
	\begin{mdframed}[font=\small]
		\setlist[itemize]{leftmargin=13pt}\setlist[enumerate]{leftmargin=13pt}\raggedright%
	}
	{\end{mdframed}
	\vspace{-0.5em}
	{\captionof{figure}{\small \TitleCaption}\label{\TitleLabel}}
	\medskip
}

\mdfdefinestyle{commonsplitbox}{%
	align=center, middlelinewidth=1.1pt,userdefinedwidth=\linewidth
	innerrightmargin=0pt,innerleftmargin=2pt,innertopmargin=5pt,
	splittopskip=7pt,splitbottomskip=3pt
}

\newenvironment{systembox*}[3]
{\begin{strip}
		\vspace{\baselineskip}\begin{titlebox}{Functionality \normalfont #1}{roundbox}{normal}{#2}{#3}}
		{\end{titlebox}
\end{strip}}

\newenvironment{gsystembox*}[3]
{\begin{strip}
		\vspace{\baselineskip}\begin{titlebox}{Global Functionality \normalfont #1}{roundbox}{normal}{#2}{#3}}
		{\end{titlebox}
\end{strip}}

\newenvironment{protocolbox*}[3]
{\begin{strip}
		\begin{titlebox}{Protocol \normalfont #1}{commonbox}{normal}{#2}{#3}}
		{\end{titlebox}
\end{strip}}

\newenvironment{algobox*}[3]
{\begin{strip}
		\begin{titlebox}{Algorithm \normalfont #1}{commonbox}{normal}{#2}{#3}}
		{\end{titlebox}
\end{strip}}

\newenvironment{reductionbox*}[3]
{\begin{strip}
		\begin{titlebox}{Reduction \normalfont #1}{commonbox}{normal}{#2}{#3}}
		{\end{titlebox}
\end{strip}}

\newenvironment{gamebox*}[3]
{\begin{strip}
		\begin{titlebox}{Game \normalfont #1}{commonbox}{gray}{#2}{#3}}
		{\end{titlebox}
\end{strip}}

\newenvironment{simulatorbox*}[3]
{\begin{strip}
		\begin{titlebox}{Simulator \normalfont #1}{commonbox}{normal}{#2}{#3}}
		{\end{titlebox}
\end{strip}}

\newenvironment{protocolsplitbox*}[3]
{\begin{strip}
		\begin{splittitlebox}{Protocol \normalfont #1}{commonbox}{normal}{#2}{#3}}
		{\end{splittitlebox}
\end{strip}}

\mdfdefinestyle{mycommonbox}{%
	align=center, middlelinewidth=1.1pt,userdefinedwidth=\linewidth
	innerrightmargin=0pt,innerleftmargin=2pt,innertopmargin=3pt,
	splittopskip=7pt,splitbottomskip=3pt
}
\mdfdefinestyle{myroundbox}{style=mycommonbox,roundcorner=5pt,userdefinedwidth=\linewidth}

\newenvironment{titlebox*}[5]
{\mdfsetup{
		style=#2,
		innertopmargin=0.3\baselineskip,
		skipabove={0.4em},
		skipbelow={1em},needspace=3\baselineskip,
		frametitleaboveskip=5em,
		innerrightmargin=5pt
	}
	\newcommand{\TitleCaption}{#4}
	\newcommand{\TitleLabel}{#5}
	\begin{mdframed}[font=\small]
		\setlist[itemize]{leftmargin=13pt}\setlist[enumerate]{leftmargin=13pt}\raggedright%
	}
	{\end{mdframed}
	\vspace{-2em}
	{\captionof{figure}{\normalfont \TitleCaption}\label{\TitleLabel}}
	\medskip
}

\newenvironment{mysystembox*}[3]
{\begin{strip}
		\vspace{\baselineskip}\begin{titlebox*}{Functionality \normalfont #1}{myroundbox}{normal}{#2}{#3}}
		{\end{titlebox*}
\end{strip}}

\newenvironment{mygsystembox*}[3]
{\begin{strip}
		\vspace{\baselineskip}\begin{titlebox*}{Global Functionality \normalfont #1}{myroundbox}{normal}{#2}{#3}}
		{\end{titlebox*}
\end{strip}}

\newenvironment{myprotocolbox*}[3]
{\begin{strip}
		\begin{titlebox*}{Protocol \normalfont #1}{mycommonbox}{normal}{#2}{#3}}
		{\end{titlebox*}
\end{strip}}

\newenvironment{myalgobox*}[3]
{\begin{strip}
		\begin{titlebox*}{Algorithm \normalfont #1}{mycommonbox}{normal}{#2}{#3}}
		{\end{titlebox*}
\end{strip}}

\newenvironment{myreductionbox*}[3]
{\begin{strip}
		\begin{titlebox*}{Reduction \normalfont #1}{mycommonbox}{normal}{#2}{#3}}
		{\end{titlebox*}
\end{strip}}

\newenvironment{mygamebox*}[3]
{\begin{strip}
		\begin{titlebox*}{Game \normalfont #1}{mycommonbox}{gray}{#2}{#3}}
		{\end{titlebox*}
\end{strip}}

\newenvironment{mysimulatorbox*}[3]
{\begin{strip}
		\begin{titlebox*}{Simulator \normalfont #1}{mycommonbox}{normal}{#2}{#3}}
		{\end{titlebox*}
\end{strip}}

\newcommand{\algoHead}[1]{\vspace{0.2em} \underline{\textbf{#1}} \vspace{0.3em}}

\makeatletter
\algnewcommand{\ExtendedState}[1]{\State
	\parbox[t]{\dimexpr\linewidth-\ALG@thistlm}{\hangindent=\algorithmicindent\strut\hangafter=3#1\strut}}
\makeatother

\algnewcommand\algorithmicinput{\textbf{Input:}}
\algnewcommand\Input{\item[\algorithmicinput]}

\algrenewcommand{\algorithmiccomment}[1]{{\color{gray}// #1}}

\newcommand{\xmath}[1]{\ensuremath{#1}\xspace}

\newcommand{\Func}[1][\relax]{\xmath{\mathcal{F}_{\textsc{#1}}}}

\DeclarePairedDelimiterX{\dbrackets}[1]{\lbrack}{\rbrack}{
	\nhphantom{\lbrack}{#1} \delimsize\lbrack \mathopen{} #1 \mathclose{} \delimsize\rbrack \nhphantom{\rbrack}{#1}
}
\DeclarePairedDelimiterX{\dbraces}[1]{\lbrace}{\rbrace}{
	\nhphantom{\lbrace}{#1} \delimsize\lbrace \mathopen{} #1 \mathclose{} \delimsize\rbrace \nhphantom{\rbrace}{#1}
}
\DeclarePairedDelimiterX{\dparens}[1]{\lparen}{\rparen}{
	\nhphantom{\lparen}{#1} \delimsize\lparen \mathopen{} #1 \mathclose{} \delimsize\rparen \nhphantom{\rparen}{#1}
}

\newcommand{\nhphantom}[2]{\sbox0{$\left#1\vphantom{#2}\right.$}\hspace{-0.58\wd0}}

\newcommand{\this}{\textsf{Tetrad}}

\newcommand{\thisA}{\textsf{Tetrad-R}}

\newcommand{\thisT}{\textsf{Tetrad}\textsubscript{\sf T}}
\newcommand{\thisC}{\textsf{Tetrad}\textsubscript{\sf C}}

\newcommand{\vl}[1]{\ensuremath{\mathsf{#1}}}

\newcommand{\pd}[1]{\ensuremath{\mathsf{\lambda}_{#1}}}
\newcommand{\pad}[2]{\ensuremath{\mathsf{\lambda}_{#1}^{#2}}} 
\newcommand{\mk}[1]{\ensuremath{\mathsf{m}_{#1}}}

\newcommand{\gm}[2]{\ensuremath{\gamma_{#1}^{#2}}}

\newcommand{\vct}[1]{\ensuremath{\vec{\mathbf{#1}}}}
\newcommand{\Mat}[1]{\ensuremath{\mathbf{#1}}}

\newcommand{\sqr}[1]{\ensuremath{\left[#1\right]}}
\newcommand{\spr}[1]{\ensuremath{\dparens*{#1}}}
\newcommand{\sgr}[1]{\ensuremath{\langle #1 \rangle}}
\newcommand{\shr}[1]{\ensuremath{\llbracket #1 \rrbracket}}

\newcommand{\shrB}[1]{\ensuremath{{\llbracket #1 \rrbracket}^{\bf B}}}
\newcommand{\shrG}[1]{\ensuremath{{\llbracket #1 \rrbracket}^{\bf G}}}

\newcommand{\TTP}{\ensuremath{\mathsf{P}_{\mathsf{TP}}}}
\newcommand{\ttp}{\mathsf{ttp}}

\newcommand{\jsend}{\ensuremath{\mathsf{jsnd}}}
\newcommand{\Sh}{\ensuremath{\mathsf{Sh}}}
\newcommand{\JSh}{\ensuremath{\mathsf{JSh}}}
\newcommand{\Rec}{\ensuremath{\mathsf{Rec}}}
\newcommand{\Mult}{\ensuremath{\mathsf{Mult}}}
\newcommand{\MultT}{\ensuremath{\mathsf{Mult3}}}

\newcommand{\Multsgr}{\ensuremath{\mathsf{MulR}}}

\newcommand{\VrfyP}{\ensuremath{\mathsf{VrfyP0}}}

\newcommand{\pijsend}{\ensuremath{\Pi_{\jsend}}}
\newcommand{\piSh}{\ensuremath{\Pi_{\Sh}}}
\newcommand{\piJSh}{\ensuremath{\Pi_{\JSh}}}
\newcommand{\piRec}{\ensuremath{\Pi_{\Rec}}}
\newcommand{\piMult}{\ensuremath{\Pi_{\Mult}}}
\newcommand{\piMultT}{\ensuremath{\Pi_{\MultT}}}
\newcommand{\piMultsgr}{\ensuremath{\Pi_{\Multsgr}}}

\newcommand{\piMultO}{\ensuremath{\Pi_{\Mult}^{\mathsf{NoPre}}}}

\newcommand{\piVrfyP}{\ensuremath{\Pi_{\VrfyP}}}

\newcommand{\FSETUP}{\ensuremath{\mathcal{F}_{\mathsf{setup}}}} 
\newcommand{\FZero}{\ensuremath{\mathcal{F}_{\mathsf{zero}}}} 

\newcommand{\SetD}{\ensuremath{\mathcal{D}}}
\newcommand{\SetE}{\ensuremath{\mathcal{E}}}

\newcommand{\GS}{\ensuremath{\mathcal{G}}}
\newcommand{\Gb}{\ensuremath{\mathsf{Gb}}}
\newcommand{\En}{\ensuremath{\mathsf{En}}}
\newcommand{\Ev}{\ensuremath{\mathsf{Ev}}}
\newcommand{\De}{\ensuremath{\mathsf{De}}}

\newcommand{\rtt}{\ensuremath{\mathsf{rtt}}}
\newcommand{\TP}{\ensuremath{\mathsf{TP}}}

\newcommand{\PlSet}[1]{\ensuremath{\Phi_{#1}}}
\newcommand{\shrC}[1]{\ensuremath{{\llbracket #1 \rrbracket}^{\bf C}}}
\newcommand{\av}[1]{\ensuremath{\mathsf{\alpha}_{\mathsf{#1}}}}
\newcommand{\key}[2]{\ensuremath{\mathsf{K}_{#1}^{#2}}}
\newcommand{\pigsh}{\ensuremath{\mathrm{\Pi}_{\mathsf{Sh}}^{\bf G}}}
\newcommand{\GC}{\ensuremath{\mathsf{GC}}}
\newcommand{\h}{\ensuremath{\mathcal{H}}}
\newcommand{\pigrec}{\ensuremath{\mathrm{\Pi}_{\mathsf{Rec}}^{\bf G}}}

\newcommand{\piab}{\ensuremath{\mathrm{\Pi}_{\mathsf{A2B}}}}
\newcommand{\arval}[1]{\ensuremath{#1^{\sf R}}} 
\newcommand{\piba}{\ensuremath{\mathrm{\Pi}_{\mathsf{B2A}}}}
\newcommand{\Y}{\ensuremath{\mathbf{Y}}}
\newcommand{\X}{\ensuremath{\mathbf{X}}}
\newcommand{\Ckt}{\ensuremath{\mathsf{Ckt}}}
\newcommand{\onesec}{\ensuremath{1^\kappa}}
\newcommand{\poly}{\ensuremath{\mathsf{poly}}}
\newcommand{\ppt}{\textsf{PPT}}
\newcommand{\priv}{\ensuremath{\mathsf{\bf priv}}}
\newcommand{\Real}{\ensuremath{\textsc{real}}}
\newcommand{\Ideal}{\ensuremath{\textsc{ideal}}}
\newcommand{\Grb}[1]{\ensuremath{\mathsf{G}^{\sf #1}}}
\newcommand{\GrbD}[1]{\ensuremath{\hat{\mathsf{G}}^{\sf #1}}}
\newcommand{\Size}[1]{\ensuremath{|#1|}}

\newcommand{\piargmin}{\ensuremath{\mathrm{\Pi}_{\mathsf{argmin}}}}
\newcommand{\piargmax}{\ensuremath{\mathrm{\Pi}_{\mathsf{argmax}}}}
\newcommand{\pibitextf}{\ensuremath{\mathrm{\Pi}_{\mathsf{bitext}}}}
\newcommand{\piobv}{\ensuremath{\mathrm{\Pi}_{\mathsf{obv}}}}

\newcommand{\Key}[1]{\ensuremath{k_{#1}}}
\newcommand{\MS}{\ensuremath{\mathsf{M}}} 
\newcommand{\msg}{\mathsf{msg}}

\newcommand{\Hyb}{\ensuremath{{\textsc{hyb}}}}
\newcommand{\prv}{\ensuremath{\mathsf{prv}}}
\newcommand{\calF}{\ensuremath{\mathcal{F}}}

\newif\iffullversion
\fullversionfalse

\iffullversion
    \newcommand{\detail}[1]{\textcolor{darkgreen} {#1}}
\else
    \newcommand{\detail}[1]{}
\fi

\newif\ifsubmission
\submissionfalse

\ifsubmission
	\newcommand{\EXTRALINES}[1] {}

	\newcommand{\OLD}[1]{}
	
\else
	\newcommand{\EXTRALINES}[1] {\textcolor{darkblue} {#1}}

		\newcommand{\OLD}[1]{{\leavevmode\color{UniBlau}{(OLD CONTENT): #1}}}
	
\fi

\begin{document}
\title{Tetrad: Actively Secure 4PC for Secure Training and Inference (Full Version)\thanks{This article is the full and extended version of an article to appear in the Network and Distributed System Security Symposium (NDSS) 2022.}}

\author{
	\IEEEauthorblockN{Nishat Koti\IEEEauthorrefmark{1}, Arpita Patra\IEEEauthorrefmark{1},
		Rahul Rachuri\IEEEauthorrefmark{2}, 
		Ajith Suresh\IEEEauthorrefmark{1}\textsuperscript{\textsection}}
	\IEEEauthorblockA{\IEEEauthorrefmark{1}Indian Institute of Science, Bangalore, Email: \{kotis, arpita, ajith\}@iisc.ac.in}
	\IEEEauthorblockA{\IEEEauthorrefmark{2}Aarhus University, Denmark, Email: rachuri@cs.au.dk}
}
\maketitle
\begingroup\renewcommand\thefootnote{\textsection}
\footnotetext{Corresponding Author. Currently postdoctoral researcher at TU Darmstadt (\url{suresh@encrypto.cs.tu-darmstadt.de})}
\endgroup

%
%
%
\begin{abstract}
Mixing arithmetic and boolean circuits to perform privacy-preserving machine learning has become increasingly popular. Towards this, we propose a framework for the case of four parties with at most one active corruption called Tetrad.

Tetrad works over rings and supports two levels of security, fairness and robustness. The fair multiplication protocol costs 5 ring elements, improving over the state-of-the-art Trident~(Chaudhari et al. NDSS'20). A key feature of Tetrad is that robustness comes for free over fair protocols. Other highlights across the two variants include (a) probabilistic truncation without overhead, (b) multi-input multiplication protocols, and (c) conversion protocols to switch between the computational domains, along with a tailor-made garbled circuit approach. 

Benchmarking of Tetrad for both training and inference is conducted over deep neural networks such as LeNet and VGG16. We found that Tetrad is up to 4 times faster in ML training and up to 5 times faster in ML inference. Tetrad is also lightweight in terms of deployment cost, costing up to 6 times less than Trident.
\end{abstract}

%
\section{Introduction}
\label{sec:4pcIntro}
Increased concerns about privacy coupled with policies such as European Union General Data Protection Regulation (GDPR) make it harder for multiple parties to collaborate on machine learning computations. The emerging field of privacy-preserving machine learning (PPML) addresses this issue by offering tools to let parties perform computations without sacrificing privacy of the underlying data. PPML can be deployed across various domains such as healthcare, recommendation systems, text translation, etc., with works like \cite{CORR:AVBCG20} demonstrating practicality.

One of the main ways in which PPML is realised is through the paradigm of secure outsourced computation (SOC). Clients can outsource the training/prediction computation to powerful servers available on a `pay-per-use' basis from cloud service providers. Of late, secure multiparty computation (MPC) based techniques~\cite{PoPETS:BCPS20, CCSW:CCPS19, NDSS:ChaRacSur20,RSA:MRSV19, CCS:MohRin18, SP:MohZha17, NDSS:PatSur20, PoPETS:WagGupCha19, ASIACCS:RWTSSK18} have been gaining interest, where a server enacts the role of a party in the MPC protocol. MPC~\cite{FOCS:Yao82b,STOC:GolMicWig87}
allows mutually distrusting parties to compute a function in a distributed fashion while guaranteeing {\em privacy} of the parties' inputs and {\em correctness} of their outputs against any coalition of $t$ parties.

The goal of PPML is practical deployment, making {\em efficiency} a primary consideration. Functions such as comparison, activation functions (e.g., ReLU), are heavily used in machine learning. Instantiating these functions via MPC naively turns out to be prohibitively inefficient due to their non-linearity. Hence, there is motivation to design specialised protocols that can compute these functions efficiently. We work towards this goal in the 4-party (4PC) setting, assuming honest majority~\cite{AC:GorRanWan18,NDSS:ChaRacSur20,PoPETS:BCPS20,USENIX:KPPS21}. 4PC is interesting because it buys us the following over 3PC (which is threshold optimal): (1) {\em independence from broadcast:} broadcast can be achieved by a simple protocol in which the sender sends to everyone and residual parties exchange and apply a majority rule (2) {\em efficient dot-product:} 4PC offers a more efficient dot-product protocol (which is an important building block for several ML algorithms) with communication complexity independent of feature size (3)  {\em simplicity and efficiency:} protocols are vastly more efficient and simpler in terms of design (as shown in this and prior works).
To enhance practical efficiency, many recent works~\cite{C:DamOrlSim18,EC:KelPasRot18,NDSS:ChaRacSur20,NDSS:PatSur20} resort to the preprocessing paradigm, which splits the computation into two phases; a preprocessing phase where input-independent (but function-dependent) computationally heavy tasks can be computed, followed by a fast online phase. Since the same functions in ML are evaluated several times, this paradigm naturally fits the case of PPML, where the ML algorithm is known beforehand. Further, recent works~\cite{C:DamOrlSim18,NDSS:DemSchZoh15,SP:DEFKSV19} propose MPC protocols over $32$ or $64$ bit rings to leverage CPU optimizations.

MPC protocols can be categorized as high-throughput~\cite{CCS:AFLNO16,EC:FLNW17,SP:ABFLLN17,CCS:MohRin18,CCSW:CCPS19,EPRINT:ADEN19,NDSS:ChaRacSur20,NDSS:PatSur20,USENIX:KPPS21,USENIX:PSSY21} and low-latency~\cite{CCS:BJPR18,CCS:BHPS19}, where the former, based on secret-sharing, requires less communication compared to the latter (garbled circuits). High-throughput protocols typically work over the boolean ring $\Z{}$ or an arithmetic ring $\Z{\ell}$ and aim to minimize communication overhead (bandwidth) at the expense of non-constant rounds. While high-throughput protocols enable efficient computation of functions such as addition, multiplication and dot-product, other functions such as division are best performed using garbled circuits. Activation functions such as ReLU used in neural networks~(NN) alternate between multiplication and comparison, wherein multiplication is better suited to the arithmetic world and comparison to the boolean world. Hence, MPC protocols working over different representations (arithmetic/boolean/garbled circuit based) can be mixed to achieve better efficiency. This provided motivation for mixed protocols where each subprotocol is executed in a world where it performs best. Mixed-protocol frameworks~\cite{NDSS:DemSchZoh15,SP:MohZha17,CCS:MohRin18,ASIACCS:RWTSSK18,INDOCRYPT:RotWoo19,NDSS:ChaRacSur20,C:EGKRS20,USENIX:PSSY21} have support for efficient ways to switch between the worlds, thereby getting the best from each of them. This work proposes a mixed-protocol PPML framework via MPC with four parties and honest majority with active security.

\begin{table*}[t]
	\centering
	\resizebox{0.99\textwidth}{!}{
		\begin{NiceTabular}{c r c r|r r|r r|ccc}[notes/para]
			\toprule
			\Block{2-1}{\# Parties} & \Block{2-1}{Reference\tabularnote{Amortized costs are reported for 1 million operations}} 
			& \Block[c]{2-1}{\#Active\\Parties\tabularnote{parties that carry out most of the computation during online phase}}
			& \Block[c]{2-1}{Security} & \Block[c]{1-2}{Dot Product\tabularnote{$\ell$ - size of ring in bits, $\kappa$ - security parameter, $\vl{d}$ - length of the vectors.}} & 
			& \Block[c]{1-2}{Dot Product with Truncation} &
			& \Block[c]{1-3}{Conversions\tabularnote{A, B, G indicate support for arithmetic, boolean, and garbled worlds respectively}} & & \\ \cmidrule{5-11}
		 &  &  &  
		 & Comm\textsubscript{pre}\tabularnote{`Comm' - communication, `pre' - preprocessing, `on' - online} 
		 & Comm\textsubscript{on} & Comm\textsubscript{pre} & Comm\textsubscript{on} & A & B & G\\ 
			\midrule
			\Block{3-1}{3}
			& ABY3~\cite{CCS:MohRin18} & 3 & Abort & $12\vl{d}\ell$ & $9\vl{d}\ell$ & $12\vl{d}\ell + 84\ell$ & $9\vl{d}\ell + 3\ell$ & \cmark & \cmark & \cmark \\
			& BLAZE~\cite{NDSS:PatSur20} & 2 & Fair & $3\ell$ & $3\ell$ & $15\ell$ & $3\ell$ & \cmark & \cmark & \xmark\\	
			& SWIFT~ (3PC)~\cite{USENIX:KPPS21} & 2 & Robust & $3\ell$ & $3\ell$ & $15\ell$ & $3\ell$ & \cmark & \cmark & \xmark \\
			\midrule
			\Block{7-1}{4}
			& Mazloom et al.~\cite{USENIX:MLRG20} & 4 & Abort & $2\ell$ & $4\ell$ & $2\ell$ & $4\ell$ & \cmark & \cmark & \xmark \\
			& Trident~\cite{NDSS:ChaRacSur20} & 3 & Fair & $3\ell$ & $3\ell$ & $6\ell$ & $3\ell$ & \cmark & \cmark & \cmark \\		
			& \textbf{$\this$} & 2 & Fair & $2\ell$ & $3\ell$ & $2\ell$ & $3\ell$ & \cmark & \cmark & \cmark \\
			\cmidrule{2-11}
			& SWIFT~(4PC)~\cite{USENIX:KPPS21} & 2 & Robust & $3\ell$ & $3\ell$ & $4\ell$ & $3\ell$ & \cmark & \cmark & \xmark \\
			& Fantastic Four~\cite{USENIX:DalEscKel20} (Best)\tabularnote{cf. \S\ref{pa:fantasticfour} for details} 
			& 4 & Robust & - & $6\ell$ &  $\ell$  & $9\ell$  & \cmark & \cmark & \xmark \\
			& Fantastic Four~\cite{USENIX:DalEscKel20} (Worst) & 3 & Robust & - & $6(\ell + \kappa$) &  $\approx80\ell+ 76\kappa$  & $9\ell + 6\kappa$  & \cmark & \cmark & \xmark \\
			& \textbf{$\thisA$} & 2 & Robust & $2\ell$ & $3\ell$ & $2\ell$ & $3\ell$ & \cmark & \cmark & \cmark \\
			\bottomrule
		\end{NiceTabular}
	}
	\caption{\small Comparison of actively-secure MPC frameworks~(3PC and 4PC) for PPML.}\label{tab:mpcCost}
	\vspace{-3mm}
\end{table*}

Works such as~\cite{CCS:MohRin18,PoPETS:WagGupCha19,USENIX:MLRG20} typically go for active security with abort, where the adversary can act maliciously to obtain the output and make honest parties abort. The stronger notion of fairness guarantees that either all or none of the parties obtain the output. This incentivizes the adversary to behave honestly in resource-expensive tasks such as PPML, as causing an abort will waste its resources. Trident~\cite{NDSS:ChaRacSur20} showed that  fairness can be achieved at the cost of security with abort.
In cases where the risk of failure of the system is too high, for instance, when deploying PPML for healthcare applications, participants might want to avoid the case when none of them receive the output. The way to tackle this issue is to modify protocols to guarantee that the correct output is always delivered to the participants irrespective of an adversary's misbehaviour. This is provided by guaranteed output delivery (GOD) or robustness. A robust protocol prevents the adversary from repeatedly causing the computations to rerun, thereby upholding the trust in the system.
We propose two variants of the framework -- one with fairness and the other with robustness. We detail the related work in \S\ref{app:Prelims} and continue with our contributions.

\subsection{Our Contributions}
We make several contributions towards designing a practically efficient 4PC mixed-protocol framework, tolerating at most one active corruption. It operates over the ring $\Z{\ell}$ and provides {\em end-to-end} conversions to switch between arithmetic, boolean and garbled worlds. We assume a one-time key setup phase and work in the (function-dependent) preprocessing model which paves the way for a fast online phase.

Depending on the sensitivity of the application and the underlying data, one might want different levels of security. For this, we propose two variants of the framework, covering fairness~($\this$) and robustness~($\thisA$) guarantees.
The fair variant improves upon the state-of-the-art {\em fair} framework of Trident~\cite{NDSS:ChaRacSur20}. $\thisA$ improves communication over the best robust protocols~\cite{USENIX:KPPS21,USENIX:DalEscKel20}, while offering support for secure training of neural networks, which was not supported in previous works.

\subsubsection{Improved Arithmetic/Boolean 4PC}
In $\this$, the multiplication protocol has a communication cost of only 5 ring elements as opposed to 6 in the state-of-the-art framework of Trident~\cite{NDSS:ChaRacSur20}. Security is elevated to robustness via $\thisA$, which has a minimal overhead over the fair one, in the preprocessing. Concretely, for a 64-bit ring with 40-bit statistical security, the overhead per multiplication is 0.027 bits for a circuit containing $2^{20}$ multiplications. This means robustness essentially comes free in the case of large circuits.

A notable contribution is the design of the multiplication protocol. It gives the following benefits -- i) support for on-demand applications, ii) probabilistic truncation without overhead and iii) multi-input multiplication gates.

\noindent {\em On-demand applications:}
The design of the multiplication protocol allows $\this$ to support on-demand applications where a preprocessing phase is not available. This variant of the protocols~(cf.~\S\ref{app:4pc}) has a round complexity that is the same as that of the online phases of the protocols in the preprocessing model and retains the same overall communication. It takes advantage of parallelization, which is often not possible in the {\em function-dependent} preprocessing model where the preprocessing and the online phases must be executed sequentially.

\noindent {\em Probabilistic truncation without any overhead:}
Multiplication (and dot product) with truncation forms an essential component while working with fixed-point values. Techniques for probabilistic truncation were proposed by~\cite{SP:MohZha17,CCS:MohRin18}. Recently, \cite{USENIX:MLRG20} gave an efficient instantiation of truncation for 4PC with abort, based on the technique of ABY3. Using that as a baseline, we demonstrate for the {\em first time} in the fair and robust settings, how multiplication (and dot-product) with truncation can be performed without any additional cost over a multiplication.

\noindent {\em Multi-input multiplication:}
Inspired by~\cite{USENIX:PSSY21, FC:OhaNui20}, we propose new protocols for 3 and 4-input multiplication, allowing multiplication of 3 and 4 inputs in one online round. Naively, performing a 4-input multiplication follows a tree-based approach, and the required communication is that of three 2-input multiplications and 2 online rounds.

Our contribution lies in keeping the communication and the round of the online phase the same as that of 2-input multiplication (i.e. invariant of the number of inputs). To achieve this,  we trade off the preprocessing cost. Looking ahead, multi-input multiplication, when coupled with the optimized parallel prefix adder circuit from~\cite{USENIX:PSSY21}, brings in a $2 \times$ improvement in online rounds. It also cuts down the online communication of secure comparison, impacting PPML applications.

\subsubsection{4PC Mixed-Protocol Framework}
In addition to relying on the improved arithmetic/boolean world, we observe that a large portion of the computation in most MPC-based PPML frameworks is done over the arithmetic and boolean worlds. The garbled world is used only to perform the non-linear operations~(e.g. softmax) that are expensive in the arithmetic/boolean world and switch back immediately after.  Leveraging this observation we propose tailor-made GC-based protocols with {\em end-to-end} conversion techniques.

The tailor-made GC for the fair protocols, has the following advantages over Trident -- i) no use of commitments for the inputs, and ii) no requirement of an explicit input sharing and output reconstruction phase, as explained later. The overall communication cost remains the same as Trident with 1 GC and 2 online rounds. In addition, for time-constrained applications we offer a variant that trades off 1 GC at the expense of 1 lesser online round. When it comes to robustness, the state-of-the-art for GC protocols are~\cite{C:IKKP15}, costing 12 GC and 2 rounds, and~\cite{CCS:BJPR18}, costing 2 GC and 4 rounds. We propose robust GC conversions for the first time, and they cost 2 GC and have an amortized round complexity of 1.

As mentioned earlier, the framework operates over three domains - arithmetic, boolean, and garbled~(\secref{4pcMixFrame}). For an operation that required computing over the garbled domain, the standard approach is to first switch from {\em Arithmetic to Garbled} and evaluate the garbled circuit to obtain a garbled-shared output. These shares are brought back to the arithmetic domain using a {\em Garbled to Arithmetic} conversion. Our approach instead is to  modify the garbled circuit such that the output is in the arithmetic domain. This eliminates the need for an explicit {\em Garbled to Arithmetic} conversion, saving in both communication and rounds in the online phase. More generally, end-to-end conversions are of the form ``$\sf{x}$-Garbled-$\sf{x}$'' where $\sf{x}$ can be either arithmetic or boolean, and need a single round for the garbled world~(cf.~\S\ref{sec:4pcMixFrame}).

Comparison of Tetrad with actively secure PPML frameworks in 3PC and 4PC is presented in \tabref{mpcCost}. The dot product is chosen as a parameter as it is one of the most crucial building blocks in PPML applications.

\subsubsection{Benchmarking and PPML}
We demonstrate the practicality of the framework, which combines the arithmetic, boolean, garbled worlds via benchmarking. The training and inference phases of deep neural networks such as LeNet~\cite{lenet} and VGG16~\cite{vgg16} and the inference phase of Support Vector Machines are benchmarked.

The implementation section is presented through the lens of deployment scenarios with two different goals. Participants in the first scenario are interested in the shortest online runtime for the computation, whereas participants in the second one want to minimize the deployment cost. Correspondingly, there are variants of our framework that cater to both scenarios.

Considering online runtime as the metric, $\thisT$ is the time-optimized~(${\sf T}$) variant with the fastest online phase.  $\thisC$ is the cost-optimized~(${\sf C}$) variant, minimizing deployment cost. This is measured via {\em monetary cost}~\cite{C:PRTY19}, which helps to capture the effect of the total runtime of the parties, and communication together.
Both variants are compared against Trident~\cite{NDSS:ChaRacSur20}, and their relative performance is indicated in \tabref{Comp}. The comparison is with respect to run time, communication, monetary cost, and throughput~(\tabref{notations}).

	\begin{table}[htb!]
		\centering
		\resizebox{0.48\textwidth}{!}{
			\begin{NiceTabular}{r c c c c c}
				\toprule
				\Block{2-1}{Protocol} 
				& \Block[c]{1-3}{Training \& Inference\tabularnote{`Com' - Communication, `Time' - Runtime, `CT' - Cumulative Runtime, `Cost' - Monetary Cost, `$\TP_{\sf on}$' - Online Throughput, ${\sf on}$ - online, ${\sf tot}$ - total}} & &
				& Training    & Inference \\ \cmidrule{2-6}
			    & $\text{Time}_{\sf on}$\tabularnote{$\Circle$ - good, $\LEFTcircle$ - better, $\CIRCLE$ - best,~(w.r.t parameter considered)} 
			    & $\text{Com}_{\sf tot}$ & $\text{CT}_{\sf tot}$ & Cost & $\TP_{\sf on}$\\
				\midrule
				\thisT
				& $\CIRCLE$      & $\LEFTcircle$  & $\CIRCLE$    & $\LEFTcircle$  & $\CIRCLE$  \\
				\thisC
				& $\LEFTcircle$  & $\CIRCLE$     & $\LEFTcircle$ & $\CIRCLE$      & $\LEFTcircle$ \\
				Trident
				& $\Circle$      & $\Circle$     & $\Circle$     & $\Circle$      & $\Circle$ \\
				\bottomrule
			\end{NiceTabular}
		}
		\vspace{-1mm}
		\caption{\small Comparison of Trident~\cite{NDSS:ChaRacSur20} with the versions of \this~for deep neural networks~(cf. NN-4 in \S\ref{sec:4pcImplementation}). }\label{tab:Comp}
		\vspace{-2mm}
	\end{table}

Trident requires $3$ parties to be active for most of the online phase, the 4th party coming in only towards the end of the computation. In Tetrad, it is brought down to $2$, having a significant impact on the monetary cost.

\tabref{Comp} shows that Tetrad is better when compared to Trident across all the parameters considered. Within Tetrad, \thisT~fares better when it comes to online run time for both training and inference, while \thisC~does better in terms of communication.
When it comes to inference, throughput is more relevant than the cost, and here, the time-optimized variant fares the best. Robust variants follow the same trends, and the reasons behind them are elaborated in \secref{4pcImplementation}.

\section{Preliminaries and Definitions}
\label{sec:4pcPrelim}
We consider $4$ parties denoted by $\Partyset = \{ P_0, P_1, P_2, P_3 \}$ that are connected by pair-wise private and authentic channels in a synchronous network, and a static, active adversary that can corrupt at most 1 party. In the secure outsourced computation (SOC) setting, the 4 servers hired to carry out the computation enact the role of the 4 parties mentioned above. In this setting inputs, intermediate values, and outputs exist in a secret-shared form. For ML training, data owners secret-share their data to the servers, which train the model using MPC. The trained model can then be reconstructed towards the data owners. Our framework is secure even if the corrupt server colludes with an arbitrary number of data owners. For ML inference, the model owner secret-shares a pre-trained model among the servers. A client secret-shares its query amongst the servers, who carry out the inference via MPC. The output is reconstructed towards the client. Security is guaranteed against a corrupt server that colludes either with the model owner or with the client. We do not guarantee the privacy of the training data against attacks such as attribute inference, membership inference, or model inversion~\cite{CCS:FreJhaRis15,USENIX:TZJRR16,SP:SSSS17}. This is an orthogonal problem, and we consider it as out-of-scope of this work.

In \this, parties rely on a one-time shared key setup (cf. \S\ref{app:Prelims} for the ideal functionality)~\cite{CCS:MohRin18,CCSW:CCPS19,NDSS:PatSur20,NDSS:ChaRacSur20,PoPETS:BCPS20} to facilitate generation of correlated randomness non-interactively. Our protocols work over the arithmetic ring $\Z{\ell}$ or boolean ring $\Z{1}$. We use fixed-point arithmetic (FPA)~\cite{CCS:MohRin18,CCSW:CCPS19,NDSS:PatSur20,NDSS:ChaRacSur20,PoPETS:BCPS20} representation to deal with floating-point values where a decimal value is represented as an $\ell$-bit integer in signed 2's complement representation. The most significant bit (MSB) represents the sign bit and $x$ least significant bits are reserved for the fractional part. The $\ell$-bit integer is then treated as an element of $\Z{\ell}$ and operations are performed modulo $2^{\ell}$. We set $\ell = 64$, $x = 13$, with $\ell - x - 1$ bits for the integral part.

\begin{notation}\label{vectors}
	For a vector $\vct{a}$, ${\vl{a}}_i$ denotes the $i^{th}$ element in the vector. For two vectors $\vct{a}$ and $\vct{b}$ of length $\vl{d}$, the dot product is given by, $\vct{a} \band \vct{b} = \sum_{i = 1}^{\vl{d}} {\vl{a}}_i {\vl{b}}_i$. 
	Given two matrices $\Mat{A}, \Mat{B}$, the operation $\Mat{A} \circ \Mat{B}$ denotes the matrix multiplication.
\end{notation}

\medskip
\begin{notation}\label{arval}
	For a bit $\bitb \in \bitset$, $\arval{\bitb}$ denotes the representation of the bit value $\bitb$ over the arithmetic ring $\Z{\ell}$. In detail, all the bits of $\arval{\bitb}$ will be zero except for the least significant bit, which is set to $\bitb$.
\end{notation}

{\em Primitives: } For our constructs we use two standard primitives~(cf. \S\ref{app:Prelims}) (a) a {\em collision-resistant} hash function, denoted as $\Hash(\cdot)$; (b) a garbling scheme $\GS = (\Gb, \En, \allowbreak \Ev, \De)$.

\paragraph{Sharing Semantics}
To enforce security, we perform computation on secret-shared data.  For the arithmetic and boolean sharing, we follow  a $(4, 1)$ replicated secret sharing~(RSS)~\cite{NDSS:ChaRacSur20}, denoted by $\shr{\cdot}$. To leverage the benefits of the preprocessing paradigm, we associate meaning to the shares and demarcate the parties in terms of their roles.  Three of the shares of a $(4, 1)$ RSS for a value $\vl{v}$ can be generated in the preprocessing phase independent of the value to be shared, and their sum can be interpreted as a mask. The fourth share, dependent on $\vl{v}$,  can be computed in the online phase and can be treated as the masked value. We denote the three preprocessed shares as $\pad{\vl{v}}{1}, \pad{\vl{v}}{2}, \pad{\vl{v}}{3}$ and the mask as $\pad{\vl{v}}{} =  \pad{\vl{v}}{1} + \pad{\vl{v}}{2} + \pad{\vl{v}}{3}$. The masked value is denoted as $\mk{\vl{v}}$, and $\mk{\vl{v}}= \vl{v} +\pad{\vl{v}}{}$.  

\begin{table}[htb!]
	\centering
	\resizebox{.49\textwidth}{!}{
		\begin{NiceTabular}{r l l l l}[notes/para]
			\toprule
			Type  & $P_0$ & $P_1$ & $P_2$ & $P_3$\\
			\midrule
			$\sqr{\cdot}$-sharing\tabularnote{$\vl{v} = \vl{v}_1 + \vl{v}_2~(+ \vl{v}_3)$} 
			& $-$      & ${\vl{v}}^1$     & ${\vl{v}}^2$    & $-$\\
			$\spr{\cdot}$-sharing 
			& $-$ 
			& ${\vl{v}}^1$     
			& ${\vl{v}}^2$    & ${\vl{v}}^3$\\
			$\sgr{\cdot}$-sharing 
			& $-$    &  $({\vl{v}}^1, {\vl{v}}^3)$ 
			& $({\vl{v}}^2, {\vl{v}}^3)$     & $({\vl{v}}^1, {\vl{v}}^2)$\\
			$\shr{\cdot}$-sharing\tabularnote{$\mk{\vl{v}} = \vl{v} + \pad{\vl{v}}{}$}  
			& $(\pad{\vl{v}}{1}, \pad{\vl{v}}{2}, \pad{\vl{v}}{3})$ 
			& $(\mk{\vl{v}}, \pad{\vl{v}}{1}, \pad{\vl{v}}{3})$  
			& $(\mk{\vl{v}}, \pad{\vl{v}}{2}, \pad{\vl{v}}{3})$ 
			& $(\mk{\vl{v}}, \pad{\vl{v}}{1}, \pad{\vl{v}}{2})$\\
			\bottomrule
		\end{NiceTabular}
	}
	\vspace{-1mm}
	\caption{\small Sharing semantics for a value $\vl{v} \in \Z{\ell}$ in \this. All the shares are $\ell$-bit ring elements.}\label{tab:sharing}
	\vspace{-3mm}
\end{table}

Next, we distinguish the four parties into two sets; the {\em eval} set $\SetE = \{P_1,P_2\}$ which is assigned the task of carrying out the computation, and is active throughout the online phase. The {\em helper} set $\SetD = \{P_0, P_3\}$ is used to assist $\SetE$ in verification, so it is only active towards the end of the computation. Complying with the roles and the RSS format, the distribution is done as follows: $P_0: \{\pad{\vl{v}}{1}, \pad{\vl{v}}{2},  \pad{\vl{v}}{3}\}, P_1: \{\pad{\vl{v}}{1},  \pad{\vl{v}}{3},  \mk{\vl{v}}\}, P_2: \{\pad{\vl{v}}{2},  \pad{\vl{v}}{3},  \mk{\vl{v}}\}$, and $P_3: \{\pad{\vl{v}}{1}, \pad{\vl{v}}{2}, \mk{\vl{v}}\}$. 
The shares are distributed among $\SetD$ such that $P_3$ gets $\mk{\vl{v}}$ whereas $P_0$ gets all the shares of $\pad{\vl{v}}{}$. During preprocessing, $P_0$ computes a part of the data needed for verification~(cf. \boxref{fig:piMultiplication}) using its input independent shares, which is communicated to $P_3$. This enables a verification in the online without $P_0$, for the fair protocols. 

Exploiting the asymmetry of the roles allows for minimal online participation, giving a huge improvement in the cumulative runtime (sum of uptime of all the parties), thereby saving in monetary costs~(cf. \S\ref{sec:4pcImplementation}). The RSS sharing semantics are presented in \tabref{sharing}, denoted by $\shr{\cdot}$, in a modular way with the help of three intermediate sharing semantics $\sqr{\cdot}, \spr{\cdot}$ and $\sgr{\cdot}$. All the sharing schemes used are linear i.e. given shares of values $\vl{v}_1,\ldots, \vl{v}_m$ and public constants $c_1,\ldots,c_m$, sharing of $\sum_{i=1}^m c_i \vl{v}_i$ can be computed locally for an integer $m$.

\begin{notation} \label{notation:concise}
	(a) For the $\shr{\cdot}$-shares of $n$ values $\vl{a}_1,\ldots,\vl{a}_n$, $\gm{\vl{a}_1 \ldots \vl{a}_n}{} = \prod\limits_{i=1}^{n} \pad{\vl{a}_i}{}$ and $\mk{\vl{a}_1 \ldots \vl{a}_n}{} = \prod\limits_{i=1}^{n} \mk{\vl{a}_i}{}$ (b) We use superscripts ${\bf B}$, and ${\bf G}$ to denote sharing semantics in boolean, and garbled world, respectively-- $\shrB{\cdot}$,  $\shrG{\cdot}$. We omit the superscript for arithmetic world. 
\end{notation}

Sharing semantics for boolean sharing over $\Z{}$ is similar to arithmetic sharing except that addition is replaced with XOR. The semantics for garbled sharing are described in \S\ref{sec:4pcMixFrame} with the relevant context.

\section{4PC Protocol}
\label{sec:4pcFourPC}
This section covers the details of our 4PC protocol over an arithmetic ring $\Z{\ell}$. We begin by explaining the relevant primitives in \secref{4pcprimitives}. The multiplication protocol with abort is presented in \secref{new4pc}, followed by details on elevating the security to fairness in \secref{fair4pc}. Lastly, in \secref{robust4pc}, we show how to improve the security to robustness\footnote{The classical notion of robustness is achieved}. Formal details along with the cost analysis for the protocols is deferred to \S\ref{app:4pc}.

\subsection{Primitives}
\label{sec:4pcprimitives}

\paragraph{Joint-Send~($\jsend$)}
\label{p:jsend}
The Joint-Send~($\jsend$) primitive allows two parties $P_i, P_j$ to relay a message $\vl{v}$ to a third party $P_k$ ensuring either the delivery of the message or $\abort$ in case of inconsistency. Towards this, $P_i$ sends $\vl{v}$ to $P_k$, while $P_j$ sends a hash of the same, $\Hash(\vl{v})$, to $P_k$. Party $P_k$ accepts the message if the hash values are consistent and $\abort$s otherwise. Note that the communication of the hash can be clubbed together for several instances and be deferred to the end of the protocol, amortizing the cost.

\paragraph{Joint-Send~($\jsend$) for robust protocols}
To achieve robustness, we instantiate $\jsend$ using the joint-message passing~(jmp) primitive of~\cite{USENIX:KPPS21}. The $\jsend$ primitive~(\boxref{fig:4pcjsendrobust}) allows two senders $P_i, P_j$ to relay a common message, $\vl{v} \in \Z{\ell}$, to a recipient $P_k$, either by ensuring successful delivery of $\vl{v}$, or by establishing a conflicting pair of parties, one among which is guaranteed to be corrupt. This implies the residual two parties are honest, one of which is then entrusted to take the computation to completion by enacting the role of a trusted party ($\TTP$). The instantiation of $\jsend$ can be viewed as consisting of two phases ({\em send, verify}), where the {\em send} phase consists of $P_i$ sending $\vl{v}$ to $P_k$ and the rest of the protocol steps go to {\em verify} phase (which ensures correct {\em send} or $\TTP$ identification). This requires $1$ round of interaction and $\ell$ bits of communication. To leverage amortization, {\em verify} will be executed only once, at the end of the computation, and requires $2$ rounds.

The $\jsend$ primitive is instantiated depending  on the desired security guarantee. For simplicity, we give common constructions for fair and robust variants of the protocols, when they only differ in the instantiation of $\jsend$.

\begin{notation}\label{notation_jsend}
	Protocol $\pijsend$ denotes the instantiation of Joint-Send~($\jsend$) primitive. We say that $P_i, P_j$ $\jsend$ $\vl{v}$ to $P_k$ when they invoke $\pijsend(P_i, P_j, \vl{v}, P_k)$.
\end{notation}

\paragraph{Sharing}
\label{p:share}
Protocol $\prot{\Sh}$~(\boxref{fig:piSh}) enables $P_i$ to generate $\shr{\cdot}$-share of a value $\vl{v}$. During the preprocessing phase, $\pd{}$-shares are sampled non-interactively using the pre-shared keys~(cf. \S\ref{app:keysetup}) in a way that $P_i$ will get the entire mask $\pd{\vl{v}}$. During the online phase, $P_i$ computes $\mk{\vl{v}} = \vl{v} + \pd{\vl{v}}$ and sends to $P_1, P_2, P_3$, which exchange the hash values to check for consistency. Parties abort in the fair protocol in case of inconsistency, whereas for robust security, parties proceed with a default value.

\smallskip
\begin{protocolbox}{$\prot{\Sh}(P_i, \vl{v})$}{$\shr{\cdot}$-sharing of a value $\vl{v}$ by party $P_i$.}{fig:piSh}
	\justify
	\algoHead{Preprocessing:} 
	Sample the following: \vspace{-2mm}
	\begin{align*}
		P_i, P_0, P_1, P_3: \pad{\vl{v}}{1}~~\Big|~~
		P_i, P_0, P_2, P_3: \pad{\vl{v}}{2}~~\Big|~~
		P_i, P_0, P_1, P_2: \pad{\vl{v}}{3}
	\end{align*}
	\justify
	\vspace{-2mm}
	\algoHead{Online:} \vspace{-2mm}
	\begin{enumerate}[itemsep=0mm]
		\item $P_i$ computes $\mk{\vl{v}} = \vl{v} + \pd{\vl{v}}$ and sends to $P_1, P_2, P_3$.
		\item $P_1, P_2, P_3$ mutually exchange $\Hash(\mk{\vl{v}})$ and accept the sharing if there exists a majority. Else parties $\abort$ for the case of fairness and accept a default value for the case of robust security. 
	\end{enumerate}       
\end{protocolbox}
\vspace{-2mm}

\paragraph{Joint Sharing}
\label{p:jsh}
Protocol $\prot{\JSh}$ enables parties $P_i, P_j$ to generate $\shr{\cdot}$-share of a value $\vl{v}$. The protocol is similar to $\prot{\Sh}$ except that $P_j$ ensures the correctness of the sharing performed by $P_i$. During the preprocessing, $\pd{}$-shares are sampled such that both $P_i, P_j$ will get the entire mask $\pd{\vl{v}}$. During the online phase, $P_i, P_j$ compute and $\jsend$ $\mk{\vl{v}} = \vl{v} + \pd{\vl{v}}$ to parties $P_1, P_2, P_3$.

For joint-sharing a value $\vl{v}$ possessed by $P_0$ along with another party in the preprocessing, the communication can be optimized further. The protocol steps based on the $(P_i, P_j)$ pair are summarised below:

	\begin{itemndss}
		\item $(P_0, P_1): \Partyset \setminus \{P_2\}$ sample $\pad{\vl{v}}{1} \in_R \Z{\ell}$; Set $\pad{\vl{v}}{2} = \mk{\vl{v}} = 0$; $P_0, P_1$ $\jsend$ $\pad{\vl{v}}{3} = - \vl{v} - \pad{\vl{v}}{1}$ to $P_2$. 
		\item $(P_0, P_2): \Partyset \setminus \{P_3\}$ sample $\pad{\vl{v}}{3} \in_R \Z{\ell}$; Set $\pad{\vl{v}}{1} = \mk{\vl{v}} = 0$; $P_0, P_2$ $\jsend$ $\pad{\vl{v}}{2} = - \vl{v} - \pad{\vl{v}}{3}$ to $P_3$. 
		\item $(P_0, P_3): \Partyset \setminus \{P_1\}$ sample $\pad{\vl{v}}{2} \in_R \Z{\ell}$; Set $\pad{\vl{v}}{3} = \mk{\vl{v}} = 0$; $P_0, P_3$ $\jsend$ $\pad{\vl{v}}{1} = - \vl{v} - \pad{\vl{v}}{1}$ to $P_1$.
	\end{itemndss}

\paragraph{Reconstruction}
\label{p:rec}
Protocol $\prot{\Rec}(\Partyset, \vl{v})$ (\boxref{fig:piRec}) enables parties in $\Partyset$ to compute $\vl{v}$, given its $\shr{\cdot}$-share. Note that each party misses one share to reconstruct the output, and the other 3 parties hold this share. 2 out of the 3 parties will $\jsend$ the missing share to the party that lacks it. Reconstruction towards a single party can be viewed as a special case.

\paragraph{$\FZero$ - Generating additive shares of zero}
\label{pa:pizero}
In $\this$, we make use of a functionality $\FZero$ to enable parties $P_0, P_i$ obtain $Z_i$ for $i \in \{1,2,3\}$ such that $Z_1 + Z_2 + Z_3 = 0$. We observe that the functionality can be instantiated non-interactively using the pre-shared keys~(cf. \S\ref{app:keysetup}). For this, parties in $\Partyset \setminus \{P_j\}$ sample random value $\vl{r}_j$ for $j \in \{1,2,3\}$. The shares are then defined as $Z_1 = \vl{r}_3 - \vl{r}_2, Z_2 = \vl{r}_1 - \vl{r}_3$ and $Z_3 = \vl{r}_2 - \vl{r}_1$.

\paragraph{Multiplication of $\sgr{\vl{a}}, \sgr{\vl{b}}$, held in clear by $P_0$}
\label{pa:piMultsgr}
To multiply $\sgr{\vl{a}}, \sgr{\vl{b}}$, where $\vl{a}, \vl{b} \in \Z{\ell}$ are held in clear by $P_0$, and generate $\sgr{\vl{z}}$ such that $\vl{z} = \vl{a} \vl{b}$, $\piMultsgr$~(\boxref{fig:piMultsgr}) proceed as follows. Parties locally generate a $\spr{\cdot}$-sharing of $\vl{z}$, where $P_0$ knows all three $\spr{\cdot}$-shares. To complete the generation of $\sgr{\vl{z}}$, $P_0, P_i$ for $i \in \{1, 2, 3\}$, randomize their $\spr{\cdot}$-share of $\vl{z}$ using $\spr{\cdot}$-share of 0, and $\jsend$ $\spr{\vl{z}}^{i}$, to one other party. 

\begin{protocolbox}{$\piMultsgr(\sgr{\vl{a}}, \sgr{\vl{b}})$}{Multiplication of $\sgr{\cdot}$-shared values, held on clear by $P_0$.}{fig:piMultsgr}
	\medskip
	\justify 
	\begin{enumerate}
		\item  Invoke $\FZero$ to enable $P_0, P_j$ obtain $Z_j$ for $j \in \{1,2,3\}$ such that $Z_1 + Z_2 + Z_3 = 0$. 
		\begin{align*}
			P_0, P_1 &~\jsend~\spr{\vl{z}}^{1} = \vl{a}^{1} \vl{b}^{3} + \vl{a}^{3} \vl{b}^{1} + \vl{a}^{3} \vl{b}^{3} + Z_1~\text{to } P_2.\\
			P_0, P_2 &~\jsend~\spr{\vl{z}}^{2} = \vl{a}^{2} \vl{b}^{3} + \vl{a}^{3} \vl{b}^{2} + \vl{a}^{2} \vl{b}^{2} + Z_2~\text{to } P_3.\\
			P_0, P_3 &~\jsend~\spr{\vl{z}}^{3} = \vl{a}^{1} \vl{b}^{2} + \vl{a}^{2} \vl{b}^{1} + \vl{a}^{1} \vl{b}^{1} + Z_3~\text{to } P_1.
		\end{align*}
		\item  Set $\sgr{\vl{z}}^{}$ as $\vl{z}^{1} = \spr{\vl{z}}^{3},~~\vl{z}^{2} = \spr{\vl{z}}^{2},~~\vl{z}^{3} = \spr{\vl{z}}^{1}$.
	\end{enumerate}
\end{protocolbox}
\vspace{-3mm}

\subsection{Multiplication in $\this$}
\label{sec:new4pc}
Given the shares of $\vl{a}, \vl{b}$, the goal of the multiplication protocol is to generate shares of $\vl{z} = \vl{ab}$. The protocol is designed such that parties $P_1, P_2$ obtain a masked version of the output $\vl{z}$, say $\vl{z} - \vl{r}$ in the online phase, and $P_0, P_3$ obtain the mask $\vl{r}$ in the preprocessing phase. Parties then generate $\shr{\cdot}$-sharing of these values by executing $\piJSh$, and locally compute $\shr{\vl{z} - \vl{r}} + \shr{\vl{r}}$ to obtain the final output. 

\paragraph{Online} 
Note that,

\begin{align}\label{eq:z+r}
	\vl{z} - \vl{r} &= \vl{a}\vl{b} - \vl{r} = (\mk{\vl{a}} - \pd{\vl{a}})(\mk{\vl{b}} - \pd{\vl{b}}) - \vl{r} \nonumber\\ 
	&= \mk{\vl{ab}} - \mk{\vl{a}}\pd{\vl{b}} - \mk{\vl{b}}\pd{\vl{a}} + \gm{\vl{a}\vl{b}}{} - \vl{r}
	~~\text{\footnotesize{(cf. notation~\ref{notation:concise})}}
\end{align}

In Eq~\ref{eq:z+r}, $P_1, P_2$ can compute $\mk{\vl{ab}}$ locally, and hence we are interested in computing $\vl{y} = (\vl{z - r}) - \mk{\vl{ab}}$. Let us view $\vl{y}$ as $\vl{y} = \vl{y}_1 + \vl{y}_2 + \vl{y}_3$, where $\vl{y}_1$ and $\vl{y}_2$ can be computed respectively by $P_1$ and $P_2$, and $\vl{y}_3$ consists of terms that can be computed by both.
    \begin{align}\label{eq:y}
			P_1: \vl{y}_1 &= - \pad{\vl{a}}{1} \mk{\vl{b}} - \pad{\vl{b}}{1} \mk{\vl{a}} + \sqr{\gm{\vl{a}\vl{b}}{} - \vl{r}}_1 \nonumber\\
			P_2: \vl{y}_2 &= - \pad{\vl{a}}{2} \mk{\vl{b}} - \pad{\vl{b}}{2} \mk{\vl{a}} + \sqr{\gm{\vl{a}\vl{b}}{} - \vl{r}}_2 \nonumber\\
			P_1, P_2: \vl{y}_3 &= - \pad{\vl{a}}{3} \mk{\vl{b}} - \pad{\vl{b}}{3} \mk{\vl{a}}
    \end{align}

The preprocessing is set up such that $P_1, P_2$ receive additive shares~($\sqr{\cdot}$) of $\gm{\vl{a}\vl{b}}{} - \vl{r}$.  $P_1, P_2$ then mutually exchange the missing share to reconstruct $\vl{y}$ and subsequently $\vl{z - r}$. 

\paragraph{Verification}
To ensure correctness of the values exchanged in the online phase, we use the assistance of $P_3$. Concretely, $P_3$ obtains $\vl{y}_1 + \vl{y}_2 + \vl{s}$, where $\vl{s}$ is a random mask known to $P_0, P_1, P_2$. For this, $P_3$ needs $\gm{\vl{a}\vl{b}}{} + \vl{s}$, which it obtains from the preprocessing phase. The mask $\vl{s}$ is used to prevent the leakage from $\gm{\vl{a}\vl{b}}{}$ to $P_3$. $P_3$ computes a hash of $\vl{y}_1 + \vl{y}_2 + \vl{s}$ and sends it to $P_1, P_2$, which $\abort$ if it is inconsistent. 

\paragraph{Preprocessing} 
Parties should obtain the following values from the preprocessing phase:
\begin{equation*}
	{\sf i)}~~P_1, P_2: \sqr{\gm{\vl{ab}}{} - \vl{r}}~~\Big|~~
	{\sf ii)}~~P_0, P_3: \vl{r}~~\Big|~~
	{\sf iii)}~~P_3: \gm{\vl{a}\vl{b}}{} + \vl{s}
\end{equation*}

For ${\sf i)}$ and ${\sf ii)}$, let $\gm{\vl{ab}}{} = \gm{\vl{ab}}{1} + \gm{\vl{ab}}{2} + \gm{\vl{ab}}{3}$, where $P_0$ along with $P_i$ can compute $\gm{\vl{ab}}{i}$ for $i \in \{1, 2, 3\}$. For $P_1, P_2$, to form an additive sharing of $(\gm{\vl{ab}}{} - r)$, it suffices for them to define their share as $\gm{\vl{ab}}{i} + \sqr{\gm{\vl{ab}}{3} - r}$. 
Instead of sampling a fresh random value for $\vl{r}$, $P_0, P_3$, along with $P_i$, sample the share for $\gm{\vl{ab}}{3} - \vl{r}$ as $\vl{u}^i$ for $i \in \{1, 2\}$.  $P_0, P_3$ compute $\vl{r}$ as $\gm{\vl{ab}}{3} - \vl{u}^1 - \vl{u}^2$. Note that $\vl{r}$ computed this way is still uniformly random, as $\vl{u}^1, \vl{u}^2$ are sampled uniformly at random.

For ${\sf iii)}$, $P_3$ needs $\vl{w} = \gm{\vl{ab}}{1} + \gm{\vl{ab}}{2} + \vl{s}$. To tackle this, $P_0, P_1, P_2$ sample $\vl{s}_1, \vl{s}_2$, and set $\vl{s} = \vl{s}_1 + \vl{s}_2$. $P_0, P_i$, for $i \in \{1, 2\}$, $\jsend$ $\gm{\vl{ab}}{i} + \vl{s}_i$ to $P_3$. This requires a communication of 2 elements. 
As an optimization, $P_0$ sends $\vl{w}$ to $P_3$. If $P_0$ is malicious, it might send a wrong value to $P_3$. However, in this case, every party in the online phase would be honest. And since $P_1, P_2$ do not use $\vl{w}$ in their computation, any error in $\vl{w}$ is bound to get caught in the verification phase.

\begin{protocolbox}{$\piMult(\vl{a}, \vl{b}, \isTr)$}{Multiplication with / without truncation in $\this$.}{fig:piMultiplication}
	Let $\isTr$ be a bit that denotes whether truncation is required ($\isTr =1$) or not~($\isTr=0$). \\
	\detail{
		{\bf Input(s):} $\shr{\vl{a}}, \shr{\vl{b}}$.\\
		{\bf Output:} $\shr{\vl{o}}$ where $\vl{o} = \vl{z}^{\vl{t}}$ if $\isTr = 1$ and $\vl{o} = \vl{z}$ if $\isTr = 0$ and $\vl{z} = \vl{ab}$.
	}
	\justify 
	\vspace{-2mm}
	\algoHead{Preprocessing:} \vspace{-2mm}
	\begin{enumerate}[itemsep=0mm]
		\item Locally compute:
		\begin{align*}
			P_0, P_1: \gm{\vl{a}\vl{b}}{1} &= \pad{\vl{a}}{1} \pad{\vl{b}}{3} + \pad{\vl{a}}{3} \pad{\vl{b}}{1} + \pad{\vl{a}}{3} \pad{\vl{b}}{3} \\
			P_0, P_2: \gm{\vl{a}\vl{b}}{2} &= \pad{\vl{a}}{2} \pad{\vl{b}}{3} + \pad{\vl{a}}{3} \pad{\vl{b}}{2} + \pad{\vl{a}}{2} \pad{\vl{b}}{2}\\
			P_0, P_3: \gm{\vl{a}\vl{b}}{3} &= \pad{\vl{a}}{1} \pad{\vl{b}}{2} + \pad{\vl{a}}{2} \pad{\vl{b}}{1} + \pad{\vl{a}}{1} \pad{\vl{b}}{1}
		\end{align*}
		\item $P_0, P_3$ and $P_j$ sample random ${\vl{u}}^j \in_R \Z{\ell}$ for $j \in \{1,2\}$. Let ${\vl{u}^1} + \vl{u}^2 = \gm{\vl{a}\vl{b}}{3} - \vl{r}$ for a random $\vl{r} \in_R \Z{\ell}$.  
		\item $P_0, P_3$ compute $\vl{r} = \gm{\vl{a}\vl{b}}{3} - {\vl{u}^1} - \vl{u}^2$ and set $\vl{q} = \vl{r}^{\vl{t}}$  if $\isTr = 1$, else set $\vl{q} = \vl{r}$. $P_0, P_3$ execute $\piJSh(P_0, P_3, \vl{q})$ to generate $\shr{\vl{q}}$.
		\item  $P_0, P_1, P_2$ sample random ${\vl{s}}_1, {\vl{s}}_2 \in_R \Z{\ell}$ and set ${\vl{s}} = {\vl{s}}_1 + {\vl{s}} _2$\footnote{For the fair protocol, it is enough for $P_0, P_1, P_2$ to sample ${\vl{s}}$ directly.}. 
		$P_0$ sends $\vl{w} = \gm{\vl{a}\vl{b}}{1} + \gm{\vl{a}\vl{b}}{2} + {\vl{s}}$ to $P_3$.
	\end{enumerate}
	\justify
	\vspace{-2mm}
	\algoHead{Online:} Let $\vl{y} = (\vl{z} - \vl{r}) - \mk{\vl{a}} \mk{\vl{b}}$. \vspace{-2mm}
	\begin{enumerate}[itemsep=0mm]
		\item Locally compute:
		\begin{align*}
			P_1: \vl{y}_1 &= - \pad{\vl{a}}{1} \mk{\vl{b}} - \pad{\vl{b}}{1} \mk{\vl{a}} + \gm{\vl{a}\vl{b}}{1} + {\vl{u}}^1 \\
			P_2: \vl{y}_2 &= - \pad{\vl{a}}{2} \mk{\vl{b}} - \pad{\vl{b}}{2} \mk{\vl{a}} + \gm{\vl{a}\vl{b}}{2} + {\vl{u}}^2 \\
			P_1, P_2: \vl{y}_3 &= - \pad{\vl{a}}{3} \mk{\vl{b}} - \pad{\vl{b}}{3} \mk{\vl{a}}
		\end{align*}
		\item $P_1$ sends $\vl{y}_1$ to $P_2$, while $P_2$ sends $\vl{y}_2$ to $P_1$, and they locally compute $\vl{z} - \vl{r} = (\vl{y}_1 + \vl{y}_2 + \vl{y}_3) + \mk{\vl{a}} \mk{\vl{b}}$.
		\item If $\isTr = 1$, $P_1, P_2$ set $\vl{p} = (\vl{z} - \vl{r})^{\vl{t}}$, else $\vl{p} = \vl{z} - \vl{r}$. $P_1, P_2$ execute $\piJSh(P_1, P_2, \vl{p})$ to generate $\shr{\vl{p}}$. 
		\item Locally compute $\shr{\vl{o}} = \shr{\vl{p}} + \shr{\vl{q}}$. Here $\vl{o} = \vl{z}^{\vl{t}}$ if $\isTr = 1$ and $\vl{z}$ otherwise.
		\item {\em Verification:} $P_3$ computes $\vl{v} = - (\pad{\vl{a}}{1} + \pad{\vl{a}}{2}) \mk{\vl{b}} - (\pad{\vl{b}}{1} + \pad{\vl{b}}{2}) \mk{\vl{a}} + {\vl{u}^1} + \vl{u}^2 + \vl{w}$ and sends $\Hash(\vl{v})$ to $P_1$ and $P_2$. Parties $P_1, P_2$ $\abort$ iff $\Hash(\vl{v}) \ne \Hash(\vl{y}_1 + \vl{y}_2 + {\vl{s}})$.
	\end{enumerate}     
\end{protocolbox}
\vspace{-2mm}

\paragraph{Truncation}
For a value $\vl{v} = \vl{v}_1 + \vl{v}_2$, SecureML~\cite{SP:MohZha17} showed that the truncated value $\vl{v}/2^x$, denoted by $\vl{v}^\vl{t}$, can be computed as $\vl{v}_1^{\vl{t}} + \vl{v}_2^{\vl{t}}$. With high probability, a truncated value having at most one bit error in the least significant position is generated. It was shown in SecureML that accuracy drop for ML algorithms due to the one bit error is minimal. However, the method cannot be generalized to more than two parties. 
ABY3~\cite{CCS:MohRin18} demonstrated the extension to 3-party setting with a generic design that uses a truncation pair of the form $(\vl{r}, {\vl{r}}^{\vl{t}})$. Here, $\vl{r}$ is a random value and $\vl{r}^{\vl{t}}$ denotes its truncated version. Given this pair, $\vl{z}$ can be truncated by opening $\vl{z} - \vl{r}$ towards all, and computing $\vl{z}^{\vl{t}}$ as $\vl{z}^{\vl{t}} = (\vl{z-r})^{\vl{t}} + \vl{r}^{\vl{t}}$. Note that all operations are carried out on shares.

The design of our multiplication allows for truncation to be carried out this way without any additional overhead in communication. Towards this, $P_1, P_2$ locally truncate $(\vl{z - r})$ and generate $\shr{\cdot}$-shares of it in the online phase. Similarly, $P_0, P_3$ truncate $\vl{r}$ in the preprocessing phase and generate its $\shr{\cdot}$-shares. Then $\shr{\vl{z}^{\vl{t}}} = \shr{(\vl{z-r})^{\vl{t}}} + \shr{\vl{r}^{\vl{t}}}$

\paragraph{Multiplication by constant}
This operation in MPC is typically local: given constant $\alpha$ and $\shr{\vl{v}}$, the product can be written as $\alpha\vl{v} = \beta^1 + \beta^2$ where $\beta^1 = \alpha.(\mk{\vl{v}} - \pad{\vl{v}}{3})$ and $\beta^2 = \alpha.(- \pad{\vl{v}}{1} - \pad{\vl{v}}{2})$.
However, in FPA, we need to perform a truncation on the output. For this $P_1, P_2$ truncate $\beta^{1}$ and execute $\prot{\JSh}$, while $P_0, P_3$ do the same with $\beta^{2}$.

\subsection{Achieving Fairness}
\label{sec:fair4pc}

Here we show how to extend the security of $\this$ from abort to fairness using techniques from Trident~\cite{NDSS:ChaRacSur20}. Before proceeding with the output reconstruction, we need to ensure that all the honest parties are alive after the verification phase. For this, all the parties maintain an {\em aliveness} bit, say $\bitb$, which is initialized to $\continue$. If the verification phase is not successful for a party, it sets $\bitb = \abort$. In the first round of reconstruction, the parties mutually exchange their $\bitb$ bit and accept the value that forms the majority. Since we have only one corruption, it is guaranteed that all the honest parties will be in agreement on $\bitb$. If $\bitb = \continue$, then the parties exchange their missing shares and accept the majority. As per the sharing semantics, every missing share is possessed by three parties, out of which there can be at most one corruption. As an optimization, for instances where many values are reconstructed, two out of the three parties can send the share while the third can send a hash of the same.

\subsection{Achieving Robustness}
\label{sec:robust4pc}
Here we show how to extend the security of $\this$ to provide robustness while retaining the same amortized communication complexity. The robust variant, denoted by $\thisA$, additionally requires a verification check in the preprocessing phase of multiplication as compared to $\this$. Moreover, the reconstruction protocol is similar to the fair counterpart, except that aliveness check is not required since a cheating would result in identifying an honest party~($\TTP$). 

The multiplication protocol $\piMult$~(\boxref{fig:piMultiplication}) is modified as follows. First, the robust variant of $\piJSh$  is used instead of the fair one. This ensures correctness of messages to be communicated or identifies a conflicting pair of parties, one among which is guaranteed to be corrupt.  
Next, to ensure the correctness of $\vl{w}$ sent by $P_0$ alone in the preprocessing phase, we introduce $\piVrfyP$~(\boxref{fig:piVrfyP0}). If $\piVrfyP$ fails, parties identify a $\TTP$ in the preprocessing phase itself. 
Finally, in case of an $\abort$ in the online phase (which proceeds similar to the that of $\piMult$), $P_0$ is assigned as the $\TTP$. Since $P_0$ does not participate in the online phase of multiplication, and its communication in the preprocessing has been verified via $\piVrfyP$, this assignment is safe.

 {\em Verifying the communication by $P_0$:}
In $\piMult$~(\boxref{fig:piMultiplication}), $P_0$ computes and sends $\vl{w} = \gm{\vl{a}\vl{b}}{1} + \gm{\vl{a}\vl{b}}{2} + {\vl{s}}_1 + {\vl{s}}_2$ to $P_3$, where $P_0, P_1, P_2$ know ${\vl{s}}_1, {\vl{s}}_2$ in clear. Note that $\vl{w} = \vl{w}^1 + \vl{w}^2$ for $\vl{w}^1 = \gm{\vl{a}\vl{b}}{1} + {\vl{s}}_1$ and $\vl{w}^2 = \gm{\vl{a}\vl{b}}{2} + {\vl{s}}_2$. Also, $P_0$ along with $P_1, P_2$ and $P_3$ possess the values $\vl{w}^1, \vl{w}^2$ and $\vl{w}$ respectively. Checking the correctness of $\vl{w}$ thus reduces to verifying if $\vl{w} = \vl{w}^1 + \vl{w}^2$.

To verify this relation for all $M$ multiplication gates in the circuit, i.e. $\{\vl{w}_j \iseq \vl{w}_j^1 + \vl{w}_j^2\}_{j \in [M]}$, one approach is to compute a random linear combination and verify the relation on the sum. 
While working over a field $\F_p$, this solution has an error probability $1/|\F_p|$, where $|\F_p|$ denotes the size of $\F_p$. However, this solution does not work naively over rings since not every element in the ring has an inverse, as opposed to fields. Concretely, the check can still pass with a probability of at most $1/2$~\cite{TCC:ACDEY19,CCS:BGIN19}. To reduce the cheating probability, the check is repeated $\csec$ times, thereby bounding the cheating probability by $1/2^{\csec}$. As an optimization, it is sufficient to choose the random combiners from $\{0,1\}$. Thus, for one check, parties need to sample only a binary string of $M$ bits using the shared-key. The formal verification protocol appears in \boxref{fig:piVrfyP0}.

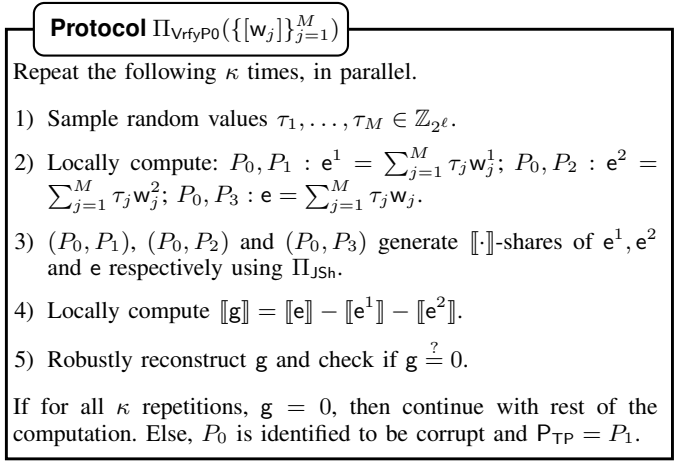
\begin{protocolbox}{$\piVrfyP(\{\sqr{{\vl{w}}_j}\}_{j=1}^{M})$}{Verification of $P_0$'s communication in the multiplication protocol of $\thisA$}{fig:piVrfyP0}
	\detail{
		{\bf Input(s):} $P_0, P_1: {\vl{w}}_j^1~~\Big|~~P_0, P_2: {\vl{w}}_j^2~~\Big|~~P_0, P_3: {\vl{w}}_j~~\Big|~$, for $j = 1, \ldots, M$.\\
		{\bf Output:} Whether ${\vl{w}}_j = {\vl{w}}_j^1 + {\vl{w}}_j^2$ or not, for $j = 1, \ldots, M$.
	}
	\justify 
	Repeat the following $\csec$ times, in parallel. 
	\begin{enumerate}[itemsep=0mm]
		\item Sample random values $\tau_1,\ldots, \tau_M \in \Z{\ell}$.
		\item Locally compute: $P_0, P_1: \vl{e}^1 = \sum_{j=1}^{M} \tau_j \vl{w}_j^1$; $P_0, P_2: \vl{e}^2 = \sum_{j=1}^{M} \tau_j \vl{w}_j^2$; $P_0, P_3: \vl{e} = \sum_{j=1}^{M} \tau_j \vl{w}_j$.
		\item $(P_0,P_1)$, $(P_0,P_2)$ and $(P_0,P_3)$ generate $\shr{\cdot}$-shares of $ \vl{e}^1,  \vl{e}^2$ and  $\vl{e}$ respectively using $\piJSh$. 
		\item Locally compute $\shr{\vl{g}} = \shr{\vl{e}} - \shr{\vl{e}^1} - \shr{\vl{e}^2}$. 
		\item Robustly reconstruct $\vl{g}$ and check if $\vl{g} \iseq 0$. 
	\end{enumerate}     
	If for all $\csec$ repetitions, $\vl{g} = 0$, then continue with rest of the computation. Else, $P_0$ is identified to be corrupt and $\TTP = P_1$.
\end{protocolbox}
\vspace{-3mm}

The robust protocol can be optimized further if cheating is detected ($\abort$ signal is generated) in the preprocessing phase. Concretely, this can be identified in the preprocessing phase either from the verification of $\jsend$ instances or output of $\piVrfyP$. When such a cheating is detected, the corrupt party is identified as follows. Parties first broadcast their shared keys established in the key-setup phase (cf. \S\ref{app:keysetup}). They recompute all the preprocessing data and verify against the data that was communicated to identify the corrupt party. Note that disclosing the shared keys does not violate input privacy because the preprocessing data is input independent. 
On identifying the corrupt party, it is eliminated from the computation, and a semi-honest 3-party computation is performed from this point onwards.

\subsection{The complete 4PC}
\label{sec:complete4pc}
The above primitives can be compiled to compute an arithmetic circuit over $\Z{\ell}$ as follows. 

Parties first invoke the key-setup functionality $\FSETUP$ (\boxref{fig:FSETUP}) for key distribution, and preprocessing of input sharing ($\piSh$) and multiplication ($\piMult$), as per the given circuit. This generates the masks ($\pad{}{}$) for all the wires in the circuit as per the sharing semantics. The preprocessing for linear gates can be performed non-interactively.
The verification of all the protocols is executed before moving on to the online phase. 

During the online phase,  $P_i \in \Partyset$ shares its input $\vl{x_i}$ by executing online steps of $\piSh$ (\boxref{fig:piSh}). Parties then evaluate the gates in the circuit in the topological order, with linear gates being computed locally, and multiplication gates being  computed via online phase of $\piMult$ (\boxref{fig:piMultiplication}). Finally, $\piRec$ (\boxref{fig:piRec}) is executed for the output wires to reconstruct the function output.

\subsection{Supporting on-demand computations}
\label{p:nopre}
For on-demand applications where the underlying function to be computed is not known in advance, the preprocessing model is not desirable. We observe that the $\this$ protocol can be modified by executing the preprocessing phase in the online phase itself, keeping the same overall communication cost. The formal protocol appears in \boxref{fig:piMultNoPre}.

\section{Mixed Protocol Framework}
\label{sec:4pcMixFrame}

In the applications we consider, the garbled circuit is used as an intermediary to evaluate certain functions where the input to the function as well as the output are in $\shr{\cdot}$-shared (or $\shrB{\cdot}$-shared) form. For this, we design end-to-end conversions which are of the form ``$\sf{x}$-Garbled-$\sf{x}$'' where $\sf{x}$ can be either arithmetic or boolean.

Similar to Trident~\cite{NDSS:ChaRacSur20}, we design a fair GC world, using techniques from~\cite{CCS:MohRosZha15}, that requires communicating 1 GC and 2 rounds for end-to-end conversions. We further extend it to provide robustness without inflating the cost. Due to its close resemblance to Trident, the details are deferred to \S\ref{subsubsec:gcworld1}.
We observe that the online rounds for end-to-end conversions can be further reduced to 1 at the expense of communicating one more GC in a parallel execution. Note that a similar approach of using 2 parallel executions in Trident does not lead to obtaining a 1-round conversion due to their protocol design and reliance on piece-wise conversions. A high-level comparison is provided in \tabref{4PCConvMall}, and more details are deferred to \S\ref{app:mixed}.

When compared to the standalone protocol of~\cite{CCS:MohRosZha15}, the customized fair GC protocol for mixed framework eliminates the need for commitments to ensure input consistency and explicit input sharing and output reconstruction phases. For robustness, the standalone GC protocols of ~\cite{C:IKKP15} requires communicating 12 GCs in 2 rounds while~\cite{CCS:BJPR18} communicates 2 GCs in 4 rounds. On the other hand, the robust variant in this work requires communicating 2 GC in 1 round. Moreover, these protocols leverage the benefit of amortization which comes from using $\jsend$.

\begin{table}[htb!]
	\centering
	\resizebox{0.48\textwidth}{!}{
		\begin{NiceTabular}{rrrrr}
			\toprule 
			Protocol\tabularnote{Notations: $\ell$ - size of ring in bits, $\kappa$ - computational security parameter.}
			& Reference 
			& \Block[r]{}{Communication\tabularnote{Cost of GC is omitted, see Tables \ref{tab:4PCConvM2}, \ref{tab:4PCConvM1} for more details.} \\ (Preprocessing)}
			& \Block[r]{}{Rounds \\ (Online)}
			& \Block[r]{}{Communication \\ (Online)} \\
			\midrule
			\Block[r]{2-1}{2 GC \\variant}
			& Trident          & \Block[r]{2-1}{$6\ell \kappa + \ell$}
			& $2$ & $4\ell \kappa + 2\ell$\\ 
			& \this & 
			& $1$ & $4\ell \kappa + \ell$\\ 
			\midrule
			\Block[r]{2-1}{1 GC \\variant}
			& Trident          & \Block[r]{2-1}{$3\ell \kappa + \ell$}
			& $2$ & $2\ell \kappa + 3\ell$\\ 
			& \this & 
			& $2$ & $2\ell \kappa + 2\ell$\\ 
			\bottomrule
		\end{NiceTabular}
	}
	\vspace{-1mm}
	\caption{\small End-to-end conversions in Trident~\cite{NDSS:ChaRacSur20} and $\this$.}\label{tab:4PCConvMall}
	\vspace{-4mm}
\end{table}

Leveraging an honest majority among the garblers and using $\jsend$, we only need semi-honest GC computation to get active security. Moreover, the state-of-the-art GC optimizations of free-XOR~\cite{ICALP:KolSch08,C:KolMohRos14}, half gates~\cite{EC:ZahRosEva15,JC:GLNP18}, and fixed AES-key~\cite{SP:BHKR13} are deployed in our protocol.

\subsection{GC for mixed protocol framework}
\label{p:gcworld2}
The 2 GC variant has two parallel executions, each comprising of 3 garblers and 1 evaluator. $P_1, P_2$ act as evaluators in two independent executions and the parties in $\PlSet{1} = \{P_0, P_2, P_3\}$, $\PlSet{2} = \{P_0, P_1 ,P_3\}$ act as garblers, respectively. Note that it suffices for only $P_0, P_3$ to generate and $\jsend$ the GC to the evaluator.

Garbled evaluation proceeds in three phases-- i) Input phase,  ii) Evaluation, and iii) Output phase. The input phase involves transferring the keys to the evaluators for every input to the GC. Note here that the function (to be evaluated via the GC) input is already $\shrB{\cdot}$-shared.  
Since each share of the function input is available with two garblers in each garbling instance, the correct key transfer is ensured via $\jsend$. The evaluation consists of GC transfer followed by GC evaluation.  Lastly, in the output phase, evaluators obtain the encoded output. Preliminary details about the garbling scheme and additional details of the GC protocol are given in \S\ref{app:garbled}. 

\paragraph{Input Phase}
\label{p:gbip}
Given that the function  input $\vl{x}$ is already available as $\shrB{\vl{x}}$, the boolean values $\mk{\vl{x}}, \av{\vl{x}}, \pad{\vl{x}}{3}$, where $\av{\vl{x}} = \pad{\vl{x}}{1} \xor \pad{\vl{x}}{2}$ and $\vl{x} = \mk{\vl{x}} \xor \av{\vl{x}} \xor \pad{\vl{x}}{3}$, act as the {\em new} inputs for the garbled computation, and garbled sharing ($\shrG{\cdot}$) is generated for each of these values. The semantics of $\shrB{\cdot}$-sharing ensures that each of these shares ($\mk{\vl{x}}, \av{\vl{x}}, \pad{\vl{x}}{3}$) is available with two garblers in each garbling instance. The keys for the shares can either be sent (using $\jsend$) correctly to the evaluators or the inconsistency is detected. This key delivery essentially generates  $\shrG{\cdot}$-sharing for each of these three values  which enables GC evaluation. Thus, the goal of our input phase is to create the compound sharing, $\shrC{\vl{x}} = (\shrG{\mk{\vl{x}}}, \shrG{\av{\vl{x}}}, \shrG{\pad{\vl{x}}{3}})$ for every input $\vl{x}$ to the function to be evaluated via the GC. We first discuss the semantics for $\shrG{\cdot}$-sharing followed by steps for generating $\shrC{\cdot}$-sharing.

\paragraph{Garbled sharing semantics}\label{gcsemantics}
A value $\vl{v} \in \Z{}$  is $\shrG{\cdot}$-shared (garbled shared) amongst  $\Partyset$ if $P_i \in \{P_0, P_3\}$ holds $\shrG{\vl{v}}_{i}= (\key{{\vl{v}}}{0,1}, \key{{\vl{v}}}{0,2})$, $P_1$ holds $\shrG{\vl{v}}_{1} = (\key{{\vl{v}}}{\vl{v},1}, \key{{\vl{v}}}{0,2})$ and $P_2$ holds $\shrG{\vl{v}}_{2} = (\key{{\vl{v}}}{0,1}, \key{{\vl{v}}}{\vl{v},2})$. Here, $\key{{\vl{v}}}{\vl{v}, j} = \key{{\vl{v}}}{0, j} \xor \vl{v} \Delta^{j}$ for $j \in \{1, 2\}$, and $\Delta^{j}$, which is known only to the garblers in $\PlSet{j}$, denotes the global offset with its least significant bit set to $1$ and is same for every wire in the circuit. 
A value $\vl{x} \in \Z{}$ is said to be $\shrC{\cdot}$-shared (compound shared) if each value  from $(\mk{\vl{x}}, \av{\vl{x}}, \pad{\vl{x}}{3})$, which are as defined above, is $\shrG{\cdot}$-shared. We write $\shrC{\vl{x}} = (\shrG{\mk{\vl{x}}},\shrG{\av{\vl{x}}},\shrG{\pad{\vl{x}}{3}})$. 

\paragraph{Generation of $\shrG{\vl{v}}$ and $\shrC{\vl{x}}$} 
Protocol $\pigsh(\Partyset, \vl{v})$~(\boxref{fig:pigsh}) enables generation of $\shrG{\vl{v}}$ where two garblers in each garbling instance hold $\vl{v}$, and proceeds as follows. Consider the first garbling instance with evaluator $P_1$ where garblers $P_k, P_l$ hold $\vl{v}$. Garblers in $\PlSet{1}$ generate $\{\key{{{\vl{v}}}}{\bitb, 1}\}_{\bitb \in \{0, 1\}}$ which denotes the key for value $\bitb$ on wire $\vl{v}$, following the free-XOR technique~\cite{ICALP:KolSch08,C:KolMohRos14}. 
$P_k, P_l$ $\jsend$ $\key{{{\vl{v}}}}{\vl{v}, 1}$ to evaluator $P_1$. Similar steps carried out with respect to the second garbling instance, at the end of which, garblers in $\PlSet{2}$ possess $\{\key{\vl{v}}{\bitb, 2}\}_{\bitb \in \{0,1\}}$ while the evaluator $P_2$ holds $\key{\vl{v}}{\vl{v}, 2}$. Following this, the shares $\shrG{\vl{v}}_s$ held by $P_s \in \Partyset$ are defined as $\shrG{\vl{v}}_0 = \shrG{\vl{v}}_3 = (\key{\vl{v}}{0, 1}, \key{\vl{v}}{0, 2})$, $\shrG{\vl{v}}_1 = (\key{\vl{v}}{\vl{v}, 1}, \key{\vl{v}}{0, 2})$, $\shrG{\vl{v}}_2 = (\key{\vl{v}}{0, 1}, \key{\vl{v}}{\vl{v}, 2})$. 

To generate $\shrC{\vl{x}}$, we need a way to generate  $(\shrG{\mk{\vl{x}}}, \shrG{\av{\vl{x}}}, \shrG{\pad{\vl{x}}{3}})$, given $\shrB{\vl{x}}$. For this, $\pigsh$ is invoked for each of $\mk{\vl{x}}, \av{\vl{x}}, \pad{\vl{x}}{3}$.  

\begin{protocolbox}{$\pigsh(\Partyset, \vl{v})$}{Generation of $\shrG{\vl{v}}$}{fig:pigsh}
	\justify
	\begin{enumerate}[itemsep=0mm]
		\item Garblers in $\PlSet{j}$ for $j \in \{1, 2\}$ generate keys $\key{{\vl{v}}}{0, j}, \key{{\vl{v}}}{1, j}$ for wire $\vl{v}$, using free-XOR technique.
		\item Let $P_k^j, P_l^j$ denote the garblers in the $j^{\text{th}}$ instance, for $j \in \{1, 2\}$, who hold $\vl{v} \in \Z{}$. $P_k^j, P_l^j$ $\jsend$ $\key{\vl{v}}{\vl{v}, j}$ to evaluator $P_j$. 
		\item $P_i \in \{P_0, P_3\}$ sets $\shrG{\vl{v}}_i = (\key{{\vl{v}}}{0,1}, \key{{\vl{v}}}{0,2})$, $P_1$ sets $\shrG{\vl{v}}_{1} = (\key{{\vl{v}}}{\vl{v},1}, \key{{\vl{v}}}{0,2})$ and $P_2$ sets $\shrG{\vl{v}}_{2} = (\key{{\vl{v}}}{0,1}, \key{{\vl{v}}}{\vl{v},2})$.
	\end{enumerate}
\end{protocolbox}
\vspace{-3mm}

\subsection{Conversions involving Garbled World} 
\label{p:conv2gc}
Assume the GC is required to compute a function $f$ on inputs $\vl{x}, \vl{y} \in \Z{\ell}$ and let the output be $f(\vl{x}, \vl{y})$. All the conversions described are for the 2 GC variant. Conversions for the 1 GC variant are straightforward, hence we omit the details. The conversions are generic for fair and robust variants, where the security follows from that of the underlying primitives. 

\medskip
{\em Case I: Boolean-Garbled-Boolean.} Since the inputs to the GC are available in boolean form, say $\shrB{\vl{x}}, \shrB{\vl{y}}$, parties generate $\shrC{\vl{x}}, \shrC{\vl{y}}$ by invoking the garbled sharing protocol $\pigsh$.
Additionally, parties $P_0, P_3$ sample $\vl{R} \in \Z{\ell}$ to mask the function output, $f(\vl{x}, \vl{y})$, and generate $\shrB{\vl{R}}$ (using the joint sharing protocol) and $\shrG{\vl{R}}$. Garblers $P_g \in \{P_0, P_2, P_3\}$ garble the circuit which computes $\vl{z} = f(\vl{x}, \vl{y}) \xor \vl{R}$, and send the GC along with the decoding information to evaluator $P_1$. Analogous steps are performed for evaluator $P_2$. Upon GC evaluation and output decoding, evaluators obtain $\vl{z} = f(\vl{x}, \vl{y}) \xor \vl{R}$, and jointly boolean share $\vl{z}$ to generate $\shrB{\vl{z}}$. Parties then compute $\shrB{f(\vl{x}, \vl{y})} = \shrB{\vl{z}} \xor \shrB{\vl{R}}$.  

\medskip
{\em Case II: Boolean-Garbled-Arithmetic.} This is similar to {\em Case I} except that the circuit which computes $\vl{z} = f(\vl{x}, \vl{y}) + \vl{R}$ is garbled instead. Boolean sharing of $\vl{z}$ is replaced with arithmetic, followed by computing $\shr{f(\vl{x}, \vl{y})} = \shr{\vl{z}} - \shr{\vl{R}}$.

\medskip
{\em Cases III \& IV: Input in Arithmetic Sharing.} 
The function to be computed $f(\vl{x}, \vl{y})$, is modified as $f^{\prime}(\mk{\vl{x}}, \av{\vl{x}}, \pad{\vl{x}}{3}, \mk{\vl{y}}, \av{\vl{y}}, \pad{\vl{y}}{3}) = f(\mk{\vl{x}}-\av{\vl{x}}-\pad{\vl{x}}{3}, \mk{\vl{y}}-\av{\vl{y}}-\pad{\vl{y}}{3})$ where inputs $\vl{x}, \vl{y}$ are replaced by the triples $\{\mk{\vl{x}}, \av{\vl{x}}, \pad{\vl{x}}{3}\}, \{\mk{\vl{y}}, \av{\vl{y}}, \pad{\vl{y}}{3}\}$ and $\av{\vl{x}} = \pad{\vl{x}}{1} + \pad{\vl{x}}{2}$ and $\av{\vl{y}} = \pad{\vl{y}}{1} + \pad{\vl{y}}{2}$. The circuit to be garbled thus, corresponds to the function $f^{\prime}$. Parties generate $\shrG{\mk{\vl{x}}}, \shrG{\av{\vl{x}}}, \shrG{\pad{\vl{x}}{3}}, \allowbreak \shrG{\mk{\vl{y}}}, \shrG{\av{\vl{y}}}, \shrG{\pad{\vl{y}}{3}}$ via $\pigsh$, following which, parties proceed with the rest of the computation whose steps are similar to {\em Case I}, and {\em II}, depending on the requirement on the output sharing.

\subsection{Other Conversions} 
\label{sec:otherconv}

\paragraph{Arithmetic to Boolean} To convert arithmetic sharing of $\vl{v} \in \Z{\ell}$ to boolean sharing, observe that $\vl{v} = \vl{v}_1 + \vl{v}_2$ where $\vl{v}_1 = \mk{\vl{v}} - \pad{\vl{v}}{3}$ is possessed by parties $P_1, P_2$, while $\vl{v}_2 = -(\pad{\vl{v}}{1} + \pad{\vl{v}}{2})$ is possessed by parties $P_0, P_3$. Thus, $\shrB{\vl{v}}$  can be computed as $\shrB{\vl{v}} = \shrB{\vl{v}_1} + \shrB{\vl{v}_2}$, where $\shrB{\vl{v}_2}$ can be generated in the preprocessing phase, and $\shrB{\vl{v}_1}$ can be generated in the online phase by the respective parties executing joint boolean sharing protocol. The protocol appears in \boxref{fig:piab}. Boolean addition, when instantiated using the adder of ABY2.0~\cite{USENIX:PSSY21}, requires $\log_4(\ell)$ rounds.

\paragraph{Boolean to Arithmetic} To convert a boolean sharing of $\vl{v}$ into an arithmetic sharing, we use techniques from~\cite{NDSS:ChaRacSur20,USENIX:KPPS21}. For a value $\vl{v} \in \Z{\ell}$, note that 

	\begin{align*}
		\vl{v} &= \sum_{i=0}^{\ell - 1} 2^{i} \vl{v}_i = \sum_{i=0}^{\ell - 1} 2^{i} ({\pd{\vl{v}}}_i \xor {\mk{\vl{v}}}_i) \\
		&= \sum_{i=0}^{\ell - 1} 2^{i} \left( \arval{{\mk{\vl{v}}}_i} +  \arval{{\pd{\vl{v}}}_i} (1 - 2\arval{{\mk{\vl{v}}}_i}) \right) 
	\end{align*}

where $\arval{{\pd{\vl{v}}}_i}, \arval{{\mk{\vl{v}}}_i}$ denote the arithmetic value of bits ${\pd{\vl{v}}}_i, {\mk{\vl{v}}}_i$ over the ring $\Z{\ell}$. 
For each bit $\vl{v}_i$ of $\vl{v}$, parties generate the arithmetic sharing of $\arval{{{\pd{\vl{v}}}}_i}$ in the preprocessing, using techniques from bit to arithmetic protocol~(cf.~\S\ref{sec:4pcTools}). During the online phase, additive shares for each bit $\vl{v}_i$ is locally computed similar to bit to arithmetic protocol. Parties then multiply the $i$th share with $2^i$ and locally add up to obtain an additive sharing of $\vl{v}$. The rest of the steps are similar to the bit to arithmetic protocol, and the formal protocol appears in \boxref{fig:piba}.

\section{Building Blocks}
\label{sec:4pcTools}
This section covers the primitives needed for realising privacy-preserving variants of the applications considered, and elaborate details appear in \S\ref{app:tools}. The building blocks can be combined to construct different layers in a neural network, as shown in \cite{ASIACCS:RWTSSK18} (Fig. 3).

\paragraph{Dot Product~(Scalar Product)}
\label{p:4pcDotp}
Given $\shr{\vct{a}}, \shr{\vct{b}}$ with $|\vct{a}| = |\vct{b}| = \vl{d}$, protocol $\prot{\dotp}$~(\boxref{fig:piDotP}) computes $\shr{\vl{z}}$ such that $\vl{z} = (\vct{a} \band \vct{b})^{\vl{t}}$ if truncation is enabled, else $\vl{z} = \vct{a} \band \vct{b}$. Following~\cite{NDSS:ChaRacSur20,USENIX:KPPS21}, we combine the partial products from the multiplication protocol across $\vl{d}$ multiplications and communicate them in a single shot. This makes the communication cost of the dot product independent of the vector size. The protocol for robust setting follows similarly. 

Matrix multiplication is an extension of the dot product protocol. We abuse notation and follow the $\shr{\cdot}$-sharing semantics~(ref. \secref{4pcPrelim}) for matrices as well. For $\Mat{X}^{u\times v}$, we have $\mk{\Mat{X}} = \Mat{X} \bigoplus \sqr{\pad{\Mat{X}}{1}} \bigoplus \sqr{\pad{\Mat{X}}{2}} \bigoplus \sqr{\pad{\Mat{X}}{3}}$. Here $\mk{\Mat{X}}$, $\sqr{\pad{\Mat{X}}{1}}$, $\sqr{\pad{\Mat{X}}{2}}$, and $\sqr{\pad{\Mat{X}}{3}}$ are matrices of dimension $u \times v$, and $\bigoplus$ denote the matrix addition operation. Looking ahead $\bigominus, \MatMul$ will  be used to denote matrix subtraction and multiplication operation, respectively. Multiplication of two matrices, $\Mat{X}^{u\times v}$, $\Mat{Y}^{v\times w}$ is a collection of $uw$ independent dot product operations over vectors of length $v$. 

\medskip
\begin{protocolbox}{$\prot{\dotp}(\vct{a}, \vct{b}, \isTr)$}{Dot Product with / without Truncation.}{fig:piDotP}
	Let $\isTr$ be a bit that denotes whether truncation is required ($\isTr =1$) or not~($\isTr=0$). \\
	\detail{
		{\bf Input(s):} $\shr{\vct{a}}, \shr{\vct{b}}$.\\
		{\bf Output:} $\shr{\vl{o}}$ where $\vl{o} = \vl{z}^{\vl{t}}$ if $\isTr = 1$ and $\vl{o} = \vl{z}$ if $\isTr = 0$ and $\vl{z} = \vct{a} \band \vct{b} = \sum_{i = 1}^{\vl{d}} {\vl{a}}_i {\vl{b}}_i$.
	}
	\justify 
	\vspace{-2mm}
	\algoHead{Preprocessing:} 
	\begin{enumerate}[itemsep=0mm]
		\item Locally compute:
		\begin{align*}
			P_0, P_1: \gm{\vct{a}\vct{b}}{1} &= \sum_{i=1}^{\vl{d}} (\pad{\vl{a}_i}{1} \pad{\vl{b}_i}{3} + \pad{\vl{a}_i}{3} \pad{\vl{b}_i}{1} + \pad{\vl{a}_i}{3} \pad{\vl{b}_i}{3})\\
			P_0, P_2: \gm{\vct{a}\vct{b}}{2} &= \sum_{i=1}^{\vl{d}} (\pad{\vl{a}_i}{2} \pad{\vl{b}_i}{3} + \pad{\vl{a}_i}{3} \pad{\vl{b}_i}{2} + \pad{\vl{a}_i}{2} \pad{\vl{b}_i}{2})\\
			P_0, P_3: \gm{\vct{a}\vct{b}}{3} &= \sum_{i=1}^{\vl{d}} (\pad{\vl{a}_i}{1} \pad{\vl{b}_i}{2} + \pad{\vl{a}_i}{2} \pad{\vl{b}_i}{1} + \pad{\vl{a}_i}{1} \pad{\vl{b}_i}{1})
		\end{align*}
		\item $P_0, P_3$ and $P_j$ sample random ${\vl{u}}^j \in_R \Z{\ell}$ for $j \in \{1,2\}$. Let ${\vl{u}^1} + \vl{u}^2 = \gm{\vct{a}\vct{b}}{3} - \vl{r}$ for a random $\vl{r} \in_R \Z{\ell}$.  
		\item $P_0, P_3$ compute $\vl{r} = \gm{\vct{a}\vct{b}}{3} - {\vl{u}^1} - \vl{u}^2$ and set $\vl{q} = \vl{r}^{\vl{t}}$  if $\isTr = 1$, else set $\vl{q} = \vl{r}$. $P_0, P_3$ execute $\piJSh(P_0, P_3, \vl{q})$ to generate $\shr{\vl{q}}$.
		\item  $P_0, P_1, P_2$ sample random ${\vl{s}}_1, {\vl{s}}_2 \in_R \Z{\ell}$ and set ${\vl{s}} = {\vl{s}}_1 + {\vl{s}} _2$\footnote{For the fair protocol, it is enough for $P_0, P_1, P_2$ to sample ${\vl{s}}$ directly.}. 
		$P_0$ sends $\vl{w} = \gm{\vct{a}\vct{b}}{1} + \gm{\vct{a}\vct{b}}{2} + {\vl{s}}$ to $P_3$.
	\end{enumerate}
	\justify
	\vspace{-2mm}
	\algoHead{Online:} Let $\vl{y} = (\vl{z} - \vl{r}) - \sum_{i=1}^{\vl{d}} \mk{\vl{a}_i} \mk{\vl{b}_i}$.
	\begin{enumerate}[itemsep=0mm]
		\item Locally compute:
		\begin{align*}
			P_1: \vl{y}_1 &= \sum_{i=1}^{\vl{d}} (- \pad{\vl{a}_i}{1} \mk{\vl{b}_i} - \pad{\vl{b}_i}{1} \mk{\vl{a}_i}) + \gm{\vct{a}\vct{b}}{1} + {\vl{u}}^1\\
			P_2: \vl{y}_2 &= \sum_{i=1}^{\vl{d}} (- \pad{\vl{a}_i}{2} \mk{\vl{b}_i} - \pad{\vl{b}_i}{2} \mk{\vl{a}_i}) + \gm{\vct{a}\vct{b}}{2} + {\vl{u}}^2\\
			P_1, P_2: \vl{y}_3 &= \sum_{i=1}^{\vl{d}} (- \pad{\vl{a}_i}{3} \mk{\vl{b}_i} - \pad{\vl{b}_i}{3} \mk{\vl{a}_i})
		\end{align*}
		\item $P_1$ sends $\vl{y}_1$ to $P_2$, while $P_2$ sends $\vl{y}_2$ to $P_1$, and they locally compute $\vl{z} - \vl{r} = (\vl{y}_1 + \vl{y}_2 + \vl{y}_3) + \sum_{i=1}^{\vl{d}} \mk{\vl{a}_i} \mk{\vl{b}_i}$.
		\item If $\isTr = 1$, $P_1, P_2$ set $\vl{p} = (\vl{z} - \vl{r})^{\vl{t}}$, else $\vl{p} = \vl{z} - \vl{r}$. $P_1, P_2$ execute $\prot{\JSh}(P_1, P_2, \vl{p})$ to generate $\shr{\vl{p}}$. 
		\item Parties locally compute $\shr{\vl{o}} = \shr{\vl{p}} + \shr{\vl{q}}$. Here $\vl{o} = \vl{z}^{\vl{t}}$ if $\isTr = 1$ and $\vl{z}$ otherwise.
		\item {\em Verification:} $P_3$ computes $\vl{v} = \sum_{i=1}^{\vl{d}} (- (\pad{\vl{a}_i}{1} + \pad{\vl{a}_i}{2}) \mk{\vl{b}_i} - (\pad{\vl{b}_i}{1} + \pad{\vl{b}_i}{2}) \mk{\vl{a}_i})+ {\vl{u}^1} + \vl{u}^2 + \vl{w}$ and sends $\Hash(\vl{v})$ to $P_1$ and $P_2$. Parties $P_1, P_2$ $\abort$ iff $\Hash(\vl{v}) \ne \Hash(\vl{y}_1 + \vl{y}_2 + {\vl{s}})$.
	\end{enumerate}     
\end{protocolbox}

In a convolutional neural network, a convolution operation can be reduced to matrix multiplications~\cite{ConvStanford,USENIX:KPPS21} as follows. Consider an $f \times f$ kernel over a $w \times h$ input with $p \times p$ padding using $s \times s$ stride having $i$ input channels and $o$ output channels. A convolution can be computed as a matrix multiplication on matrices of dimension $(w^{\prime} \cdot h^{\prime}) \times (i \cdot f \cdot f)$ and $(i \cdot f \cdot f) \times (o)$ where $w^{\prime} = \dfrac{w-f+2p}{s} + 1$ and $h^{\prime} = \dfrac{h-f+2p}{s} + 1$.

\paragraph{Multi-input Multiplication}
\label{p:4pcMultiinputfair}
Inspired from ABY2.0~\cite{USENIX:PSSY21}, we design 3-input and 4-input multiplication protocols for our setting. We remark that the multi-input multiplication, when coupled with the optimized PPA circuit from~\cite{USENIX:PSSY21}, improves the rounds as well as communication in the online phase.

The goal of 3-input multiplication is to generate $\shr{\cdot}$-sharing of $\vl{z} = \vl{a} \vl{b} \vl{c}$ given $\shr{\vl{a}}, \shr{\vl{b}}, \shr{\vl{c}}$, without the need for performing two sequential multiplications (i.e. first $\vl{a} \vl{b}$ then $\vl{a} \vl{b} \vl{c}$). For this parties proceed similar to the multiplication protocol (see \S\ref{sec:new4pc}), where they compute $\shr{\vl{z}} = \shr{\vl{z} - \vl{r}} + \shr{\vl{r}}$. Observe that 

\begin{align*} 
	\vl{z} - \vl{r}
	&= \vl{a} \vl{b} \vl{c} - \vl{r} = (\mk{\vl{a}} - \pd{\vl{a}})(\mk{\vl{b}} - \pd{\vl{b}})(\mk{\vl{c}} - \pd{\vl{c}}) - \vl{r} \\
	&= \mk{\vl{a}\vl{b}\vl{c}} - \mk{\vl{a}\vl{c}} \pd{\vl{b}} - \mk{\vl{b}\vl{c}} \pd{\vl{a}} - \mk{\vl{a}\vl{b}} \pd{\vl{c}} + \mk{\vl{a}} \gm{\vl{b} \vl{c}}{} +  \mk{\vl{b}} \gm{\vl{a} \vl{c}}{} \\
	&~~~~+ \mk{\vl{c}} \gm{\vl{a} \vl{b}}{} - \gm{\vl{a} \vl{b} \vl{c}}{} - \vl{r}
\end{align*}

Similar to the 2-input fair multiplication $\piMult$~(\boxref{fig:piMultiplication}), the goal of the preprocessing phase is to generate additive shares of $\gm{\vl{ab}}{}, \gm{\vl{ac}}{}, \gm{\vl{bc}}{},  \gm{\vl{abc}}{}$ among $P_1, P_2$.

Informally, the terms that $P_1, P_2$ cannot compute locally for the aforementioned $\gm{}{}$ values, can be computed by $P_0, P_3$, as evident from our sharing semantics. $P_0, P_3$ compute the missing terms and share them among $P_1, P_2$ in the preprocessing phase. $P_1, P_2$ proceed with online phase similar to $\piMult$, to compute $\vl{z} - \vl{r}$. Thus the online complexity is retained as that of $\piMult$ while the preprocessing communication is increased to 9 elements. The protocol appears in \boxref{fig:piMultT}.

For the 4-input case, the goal is to compute $\vl{z}= \vl{abcd}$ for which the additive shares of $\gm{\vl{ab}}{}$, $\gm{\vl{ac}}{}$, $\gm{\vl{ad}}{}$, $\gm{\vl{bc}}{}$, $\gm{\vl{bd}}{}$, $\gm{\vl{cd}}{}$, $\gm{\vl{abc}}{}$, $\gm{\vl{acd}}{}$, $\gm{\vl{bcd}}{}$, $\gm{\vl{abcd}}{}$ needs to be generated in the preprocessing. The protocol is very similar to the 3-input case, and the details are deferred to \S\ref{app:tools}. 

\paragraph{Secure Comparison}
\label{p:4pcBitExt}
To compute $\vl{a} > \vl{b}$ in the FPA representation, given its $\shr{\cdot}$-sharing, $\prot{\bitext}$ uses the technique of extracting the most significant bit~($\msb$) of the value $\vl{v} = \vl{a} - \vl{b}$~\cite{CCS:MohRin18, NDSS:PatSur20, USENIX:KPPS21}.  
To compute the $\msb$, we use two variants - i) the communication optimized parallel prefix adder~(PPA) circuit from ABY3~\cite{CCS:MohRin18} ($2 (\ell - 1)$ AND gates, $\log \ell$ depth), and ii) the round optimized bit extraction circuit from ABY2~\cite{USENIX:PSSY21}. The circuit of ABY2 uses multi-input AND gates and has a multiplicative depth of $\log_4(\ell)$. Both these circuits take two $\ell$-bit values in boolean sharing as the input and outputs the result in boolean sharing form. Note that $\vl{v} = (\mk{\vl{v}} - \pad{\vl{v}}{3}) + (- \pad{\vl{v}}{1} - \pad{\vl{v}}{2})$ as per the sharing semantics~(cf.~\tabref{sharing}). $P_0, P_3$ execute $\protB{\JSh}$ on $(- \pad{\vl{v}}{1} - \pad{\vl{v}}{2})$ during the preprocessing, while $P_0, P_3$ execute $\protB{\JSh}$ on $(\mk{\vl{v}} - \pad{\vl{v}}{3})$ during the online phase to generate the respective boolean sharing.

\paragraph{Bit to Arithmetic}
\label{p:4pcBit2A}
Protocol $\prot{\bitA}$~(\boxref{fig:piBitA}) enables computing $\shr{\bitb}$ of a bit $\bitb$ given its boolean sharing $\shrB{\bitb}$. Let $\arval{\bitb}$ denotes the value of $\bitb \in \bitset$ over the arithmetic ring $\Z{\ell}$. Then for $\bitb = \bitb_1 \xor \bitb_2$, note that $\arval{\bitb} = (\arval{\bitb_1} - \arval{\bitb_2})^2$. 

Let $\bitb_1 = \mk{\vl{\bitb}} \xor \pad{\vl{v}}{3}$ and $\bitb_2 = \pad{\vl{v}}{1} \xor \pad{\vl{v}}{2}$. To compute $\shr{\bitb}$, a pair of parties can generate the arithmetic sharing corresponding to $\arval{\bitb_1}$ and $\arval{\bitb_2}$ by executing $\prot{\JSh}$. $\shr{\bitb}$ can be computed by invoking $\piMult$ once with inputs $\vl{x} = \vl{y} = \arval{\bitb_1}- \arval{\bitb_2}$.

Using the techniques from~\cite{NDSS:ChaRacSur20, USENIX:KPPS21}, we obtain a communication-optimized variant by trading off computation in the preprocessing. For this, note that 

\begin{equation}
	\arval{\bitb} = \arval{(\mk{\bitb} \xor \pad{\bitb}{})} = \arval{\mk{\bitb}} + \arval{(\pad{\bitb}{})}(1-2\arval{\mk{\bitb}})
\end{equation} 

Let $\vl{v} = \arval{\mk{\bitb}}$ and $\vl{u} = \arval{(\pad{\bitb}{})}$. During the preprocessing, $P_0$ generates $\sgr{\cdot}$-sharing of $\vl{u}$ and a check is executed to verify its correctness. The online phase consists of each pair of parties $(P_1, P_3)$, $(P_2, P_3)$ and $(P_1, P_2)$ locally computing an additive sharing of $\arval{\bitb}$, generating the corresponding $\shr{\cdot}$-sharing using $\prot{\JSh}$, and locally adding the shares to obtain $\shr{\bitb}$.  

\paragraph{Bit Injection}
\label{p:4pcBitInj}

Protocol $\prot{\bitinj}$ enables computing $\shr{\bitb \vl{v}}$, given the boolean sharing $\shrB{\bitb}$ of a bit $\bitb$ and the arithmetic sharing $\shr{\vl{v}}$ of a value $\vl{v} \in \Z{\ell}$. Similar to $\prot{\bitA}$, 

\begin{align*}
	\arval{(\bitb\vl{v})} &= \arval{(\mk{\bitb} \xor \pad{\bitb}{})}(\mk{\vl{v}} - \pd{\vl{v}}{}) \\
	&= (\arval{\mk{\bitb}} + \arval{(\pad{\bitb}{})}(1-2\arval{\mk{\bitb}}))(\mk{\vl{v}} - \pd{\vl{v}}{})\\
	&= \arval{\mk{\bitb}}\mk{\vl{v}} - \arval{\mk{\bitb}}\pd{\vl{v}}{} + (2\arval{\mk{\bitb}} - 1)(\arval{(\pad{\bitb}{})}\pd{\vl{v}}{} - \mk{\vl{v}}\arval{(\pad{\bitb}{})})
\end{align*} 

During preprocessing, $P_0$ generates $\sgr{\cdot}$-sharing of $\arval{\pd{\bitb}}$, followed by verifying its correctness, similar to $\prot{\bitA}$. $\sgr{\cdot}$-shares of $\arval{(\pad{\bitb}{})}\pd{\vl{v}}{}$ are generated by multiplying $\sgr{\arval{(\pad{\bitb}{})}}$ and $\sgr{\pd{\vl{v}}{}}$ using $\piMultsgr$ (\boxref{fig:piMultsgr}). 
In the online phase, each pair of parties $(P_1, P_3)$, $(P_2, P_3)$ and $(P_1, P_2)$ locally compute an additive sharing of $\arval{(\bitb\vl{v})}$, generate its $\shr{\cdot}$-sharing using $\prot{\JSh}$, and locally add these shares to generate $\shr{\arval{(\bitb\vl{v})}}$.

\paragraph{Oblivious Selection}
\label{pa:4pcObvselect}
Given $\shr{\cdot}$-shares of $\vl{x}_0, \vl{x}_1 \in \Z{\ell}$ and $\shrB{\bitb}$ where $\bitb \in \bitset$, oblivious selection~($\piobv$) enables parties to generate re-randomized $\shr{\cdot}$-shares of $\vl{z} = \vl{x}_{\bitb}$. The protocol is similar in spirit to Oblivious Transfer primitive. Note that $\vl{z}$ can be written as $\vl{z} = \bitb (\vl{x}_1 - \vl{x}_0) + \vl{x}_0$. Parties invoke $\prot{\bitinj}$ to compute $\shr{\bitb (\vl{x}_1 - \vl{x}_0)}$, and sum it with $\shr{\vl{x}_0}$ to generate $\shr{\vl{z}}$.

\paragraph{Piece-wise Polynomials}
\label{p:4pcPiecewise}
Piece-wise polynomial functions  are constructed as a series of constant public polynomials $f_1, \ldots, f_m$ and $ c_1 < \ldots < c_m$ such that, 

\begin{footnotesize}
\begin{align*}
	f(y) = 
	\begin{cases}
		0, & y< c_1 \\
		f_1, & c_1 \leq y < c_2 \\
		\ldots & \\
		f_m, & c_m \leq y
	\end{cases}
\end{align*}
\end{footnotesize}

$f$ can be computed as, $f(y) = \sum_{i=1}^{m} \bitb_i \cdot (f_{i} - f_{i-1})$, where $f_0 = 0$, $f_m = 1$, and $\bitb_i = 1$ if $y  \geq c_i$ and $0$ otherwise, for $i \in \{1, \ldots, m\}$.
Given the $\shr{\cdot}$-shares of $y$, one can obtain the $\shrB{\cdot}$-shares of the bits $\bitb_1, \ldots, \bitb_m$ using secure comparison. Shares of the product terms, $\bitb_i \cdot (f_{i} - f_{i-1})$, can thus be generated by invoking $m$ $\prot{\bitinj}$, followed by a local addition. 
A naive application of $\prot{\bitinj}$ involves sharing (via $\prot{\JSh}$) additive shares of $\bitb_i \cdot (f_{i} - f_{i-1})$, thereby requiring $m$ $\prot{\JSh}$ in the online phase. Instead, it can be made independent of $m$ by first computing additive shares of $f(y)$, and then invoking one $\prot{\JSh}$.

Non-linear activation functions, such as Rectified Linear Unit and Sigmoid, can be viewed as instantiations of piece-wise polynomial functions as shown in ABY3~\cite{CCS:MohRin18}.

\paragraph{ArgMin/ ArgMax}
\label{p:4pcArgMinMaz}
Protocol $\piargmin$ (\boxref{fig:piargmin}) allows parties to compute the index of the smallest element in a vector $\vct{x} = (\vl{x}_1, \ldots, \vl{x}_m)$ of $m$ elements, where $\vct{x}$ is $\shr{\cdot}$-shared, i.e. each element $\vl{x}_i \in \Z{\ell}$ of $\vct{x}$ is $\shr{\cdot}$-shared. The protocol outputs a $\shrB{\cdot}$-shared bit vector $\vct{b}$ of size $m$ which has a $1$ at the index associated with the minimum value in $\vct{x}$, and $0$ elsewhere. 
We follow the standard tree-based approach~\cite{SP:DEFKSV19} to recursively find the minimum value in $\vct{x}$ while also updating $\vct{b}$ to reflect the index of this smallest element. Each bit of $\vct{b}$ is initialized to 1. The elements of $\vct{x}$ are grouped into pairs and securely compared to find their pairwise minimum. Using this information, $\vct{b}$ is updated such that $\bitb_j$'s are reset to $0$ for $\vl{x}_j$'s $\in \vct{x}$ which do not form the minimum in their respective pair; the other bits in $\vct{b}$ still equal $1$. The protocol recurses on the remaining elements $\vl{x}_j \in \vct{x}$, which were the pairwise minimums. Eventually, only one $\bitb_j \in \vct{b}$ equals $1$, indicating that $\vl{x}_j$ is the minimum, with index $j$. Computing $\piargmax$ can be done similarly.
\section{Implementation and Benchmarking}
\label{sec:4pcImplementation}
We benchmark training and inference phases for deep NNs with varying parameter sizes and the inference phase for Support Vector Machines (SVM) using MNIST~\cite{MNIST10} and CIFAR-10~\cite{CIFAR10} dataset. Training phase of SVM requires additional tools and primitives, and is out of scope of this work. Benchmarks of the protocols are against the state-of-the-art 4PC of Trident \cite{NDSS:ChaRacSur20} and SWIFT \cite{USENIX:KPPS21} 4PC (supports only inference).

\paragraph{Benchmarking Environment Details} 
The protocols are benchmarked over a Wide Area Network (WAN), instantiated using n1-standard-64 instances of Google Cloud\footnote{\url{https://cloud.google.com/}}, with machines located in East Australia ($P_0$), South Asia ($P_1$), South East Asia ($P_2$), and West Europe ($P_3$). The machines are equipped with 2.0 GHz Intel (R) Xeon (R) (Skylake) processors supporting hyper-threading, with 64 vCPUs, and 240 GB of RAM Memory. Parties are connected by pairwise authenticated bidirectional synchronous channels (e.g., instantiated via TLS over TCP/IP). We use a bandwidth of $40$ MBps between every pair of parties and the average round-trip time ($\rtt$)\footnote{Time for communicating 1 KB of data between a pair of parties} values among $P_0$-$P_1$, $P_0$-$P_2$, $P_0$-$P_3$, $P_1$-$P_2$, $P_1$-$P_3$, and $P_2$-$P_3$ are $153.74 ms$, $93.39 ms$, $274.84 ms$, $62.01 ms$, $174.15 ms$, and $219.46 ms$ respectively.

For a fair comparison, we implemented and benchmarked all the protocols, including the protocols of Trident and SWIFT, building on the ENCRYPTO library~\cite{ENCRYPTO} in C++17. Primitives such as maxpool, which Trident and SWIFT do not support, have been run using our building blocks. We would like to clarify that our code is developed for benchmarking, is not optimized for industry-grade use, and optimizations like GPU support can further enhance performance. Our protocols are instantiated over a $64$-bit ring ($\Z{64}$), and the collision-resistant hash function is instantiated using SHA-256. We use multi-threading, and our machines are capable of handling a total of 64 threads. Each experiment is run 10 times, and the average values are reported. We use $1$ KB = $8192$ bits and use a batch size of $B = 128$ for training.

\paragraph{Benchmarking Parameters} 
We evaluate the protocols across a variety of parameters as given in \tabref{notations}. In addition to parameters such as runtime, communication, and {\em online throughput} ($\TP$)~\cite{CCS:AFLNO16,SP:ABFLLN17,CCS:MohRin18,NDSS:ChaRacSur20}, the cumulative runtime (sum of the up-time of all the hired servers) is also reported. This is because when deployed over third-party cloud servers, one pays for them by the communication and the uptime of the hired servers. To analyze the cost of deployment of the framework, {\em monetary cost} ($\sf Cost$)~\cite{C:MPRSY20} is reported. This is done using the pricing of Google Cloud Platform\footnote{See \url{https://cloud.google.com/vpc/network-pricing} for network cost and  \url{https://cloud.google.com/compute/vm-instance-pricing} for computation cost.}, where for $1$ GB and $1$ hour of usage, the costs are USD $0.08$ and USD $3.04$, respectively. For protocols with an asymmetric communication graph, communication load is unevenly distributed among all the servers, leaving several communication channels underutilized. Load balancing improves the performance by running several execution threads in parallel, each with the roles of the servers changed. Load balancing has been performed in all the protocols benchmarked.

\vspace{-1mm}
\begin{table}[htb!]
	\centering
	\resizebox{0.48\textwidth}{!}{
		\begin{NiceTabular}{p{1.3cm} l }
			\toprule
			Notation & Description\\
			\midrule 
			${\sf T}_{\sf on,i}$        & Online runtime of party $P_i$.\\
			${\sf T}_{\sf tot,i}$        & Total runtime of party $P_i$.\\
			${\sf PT}_{\sf on}$        & Protocol online runtime; ${\sf max_i} \{ {\sf T}_{\sf on,i} \}$ .\\
			${\sf PT}_{\sf tot}$        & Protocol total runtime;  ${\sf max_i} \{ {\sf T}_{\sf tot,i} \}$ .\\
			${\sf CT}_{\sf on}$         & Cumulative online runtime; $\Sigma_i {\sf T}_{\sf on,i}$ .\\
			${\sf CT}_{\sf tot}$         & Cumulative total runtime; $\Sigma_i {\sf T}_{\sf tot,i}$ .\\
			${\sf Comm}_{\sf on}$    & Online communication.\\
			${\sf Comm}_{\sf tot}$    & Total communication.\\
			${\sf Cost}$                     & Total monetary cost.\\
			$\TP$                              & \makecell[l]{Online throughput (higher = better)\\(\#iterations /  \#queries per minute in online)}\\
			\bottomrule
		\end{NiceTabular}
	}
	\vspace{-1mm}
	\caption{\small Benchmarking parameters (lower is better, except for $\mathsf{TP}$)\label{tab:notations}}
	\vspace{-3mm}
\end{table}

\paragraph{Network Architectures}
We consider the following networks for benchmarking. These were chosen based on the different range of model parameters and types of layers used in the networks. We refer readers to \cite{SP:MohZha17,PoPETS:WTBKMR21} for the architecture and a detailed description of the training and inference steps for the ML algorithms.
\begin{myitemize}
	\item[--] {\em SVM:} Consists of 10 categories for classification~\cite{SP:DEFKSV19}.
	\item[--] {\em NN-1:} Fully connected network with 3 layers and around 118K parameters~\cite{CCS:MohRin18,NDSS:PatSur20}. 
	\item[--] {\em NN-2:} Convolutional neural network comprising of 2 hidden layers, with 100 and 10 nodes~\cite{ASIACCS:RWTSSK18,CCS:MohRin18,NDSS:ChaRacSur20}.
	\item[--] {\em NN-3:} LeNet~\cite{lenet}, comprises of 2 convolutional and fully connected layers, followed by maxpool for convolutional layers. This has approximately 431K parameters.
	\item[--] {\em NN-4:} VGG16~\cite{vgg16} has 16 layers in total and contains fully-connected, convolutional, ReLU activation and maxpool layers. This has $\approx$138 million parameters. 
\end{myitemize}

\paragraph{Datasets} 
We use the following datasets:
\begin{myitemize}
	\item[--] MNIST~\cite{MNIST10} is a collection of 28$\hspace{0.25em}\times\hspace{0.25em}$28 pixel, handwritten digit images with a label between 0 and 9 for each. It has 60,000 and respectively, 10,000 images in training and test set. We evaluate NN-1, NN-3, SVM on this dataset.
	\item[--] CIFAR-10~\cite{CIFAR10} has 32$\hspace{0.25em}\times\hspace{0.25em}$32 pixel images of 10 different classes such as dogs, horses, etc. It has 50,000 images for training and 10,000 for testing, with 6,000 images in each class. NN-2, NN-4 are evaluated on this dataset.
\end{myitemize}

\paragraph{Discussion}
Broadly speaking, we consider two deployment scenarios -- optimized for time~({\sf T}), and for cost~({\sf C}). In the first one, participants want the result of the output as soon as possible while maximizing the online throughput. In the second one, they want the overall monetary cost of the system to be minimal and are willing to tolerate an overhead in the execution time. 
Using multi-input multiplication gates and the 2 GC variant of the garbled makes the online phase faster but incur an increase in monetary cost. This is because they cause an overhead in communication in the preprocessing phase, and communication affects monetary cost more than uptime (in our setting).

$\thisT$ makes use of multi-input multiplication gates and the 2 GC variant of the garbled world and is the fastest variants of the framework. On the other hand, $\thisC$ is the variant with minimal monetary cost. We only report the numbers for the fair variant of $\this$ and not the robust variant. The overhead for the robust variant over the fair one is minimal, and is primarily due to (i) the use of \emph{robust} joint-send primitive and (ii) the augmented one-time verification check at the end of the preprocessing phase. The overhead amortises for deep networks, like the ones considered in this work.

\begin{figure*}[t!]
	\centering
	\begin{subfigure}{.36\textwidth}
		\centering
		\resizebox{.95\textwidth}{!}{
			\begin{tikzpicture}[
				every axis/.style={ 
					ybar stacked,
					ymin=0,ymax=75,
					xtick={1,2,3,4}, xticklabels={NN-1,NN-2,NN-3,NN-4},
					enlarge x limits=0.2,
					cycle list name=exotic, 
					every axis plot/.append style={fill,draw=none,no markers},
					legend style = {anchor = south, legend columns = -1, draw=none, area legend},
					bar width=10pt},]
				
				\begin{axis}[bar shift=-12pt,hide axis, legend style = {at={(0.2, 0.825)}}]
					\addplot+ coordinates
					{(1,8.06) (2,8.13) (3,21.79) (4,72.01)}; 	
					\addlegendentry{{\footnotesize Trident~~~~~~}}
				\end{axis}
				
				\begin{axis}[hide axis, legend style = {at={(0.2, 0.75)}}]
					\addplot+[fill=UniOrange] coordinates
					{(1,1.93) (2,2.05) (3,5.67) (4,25.90)}; 	
					\addlegendentry{{\footnotesize \thisT~~~~~}}
				\end{axis}
				
				\begin{axis}[bar shift=12pt, legend style = {at={(0.2, 0.675)}}]
					\addplot+[fill=UniGruen] coordinates
					{(1,2.55) (2,2.67) (3,8.40) (4,38.35)};  
					\addlegendentry{\footnotesize \thisC~~~~~}
				\end{axis}

			\end{tikzpicture}
		}
		\vspace{-1mm}
		\caption{\footnotesize Online Execution Time~(${\sf PT}_{\sf on}$)}\label{fig:TrainA}
	\end{subfigure}
    \hspace{2mm}
	\begin{subfigure}{.36\textwidth}
		\centering
		\resizebox{.95\textwidth}{!}{
			\begin{tikzpicture}[
				every axis/.style={ 
					ybar stacked,
					ymin=0,ymax=650,
					xtick={1,2,3,4}, xticklabels={NN-1,NN-2,NN-3,NN-4\footnote{scaled down by a factor of $10$ for better visibility}},
					enlarge x limits=0.2,
					cycle list name=exotic, 
					every axis plot/.append style={fill,draw=none,no markers},
					legend style = {anchor = south, legend columns = -1, draw=none, area legend},
					bar width=10pt},]
				
				\begin{axis}[bar shift=-12pt,hide axis, legend style = {at={(0.2, 0.825)}}]
					\addplot+ coordinates
					{(1,49.33) (2,70) (3,331.01) (4,577.927)}; 	
					\addlegendentry{{\footnotesize Trident~~~~~~}}
				\end{axis}
				
				\begin{axis}[hide axis, legend style = {at={(0.2, 0.75)}}]
					\addplot+[fill=UniOrange] coordinates
					{(1,58.51) (2,75.67) (3,343.73) (4,514.610)}; 	
					\addlegendentry{{\footnotesize \thisT~~~~~}}
				\end{axis}
				
				\begin{axis}[bar shift=12pt, legend style = {at={(0.2, 0.675)}}]
					\addplot+[fill=UniGruen] coordinates
					{(1,34.29) (2,49.16) (3,240.41) (4,399.930)};  
					\addlegendentry{\footnotesize \thisC~~~~~}
				\end{axis}

			\end{tikzpicture}
		}
		\vspace{-1mm}
		\caption{\footnotesize Monetary Cost~(${\sf Cost}$)}\label{fig:TrainC}
	\end{subfigure}
	\caption{\small Training of Neural Networks: in terms of ${\sf PT}_{\sf on}$  and ${\sf Cost}$ (lower is better)~(cf.~\tabref{notations})}\label{fig:MLTrain}
	\vspace{-4mm}
\end{figure*}

\subsection{ML Training}
\label{subsec:bench_train}
For training we consider NN-1, NN-2, NN-3 and NN-4 networks. We report values corresponding to one iteration, that comprises of a forward propagation followed by a backward propagation. More details are provided in \S\ref{app:implementation}. 

	\begin{table}[htb!]
		\centering
		\resizebox{0.43\textwidth}{!}{
			\begin{NiceTabular}{r r r r r}
				\toprule
				Algorithm & Parameter & Trident & \thisT & \thisC\\
				\midrule 
				\multirow{4}{*}{NN-1} 
				& ${\sf PT}_{\sf on}$            & 8.06       & 1.93      & 2.55    \\
				& ${\sf PT}_{\sf tot}$           & 10.76      & 5.05      & 5.27    \\
				& ${\sf CT}_{\sf tot}$           & 27.90      & 12.69     & 11.22.  \\
				& ${\sf Comm}_{\sf tot}$         & 0.16       & 0.30      & 0.16    \\
				& ${\sf Cost}$                   & 49.33      & 58.51     & 34.29   \\
				& $\TP$                          & 1904.79    & 3792.64   & 3725.49 \\
				\midrule
				\multirow{4}{*}{NN-2} 
				& ${\sf PT}_{\sf on}$            & 8.13       & 2.05      & 2.67    \\
				& ${\sf PT}_{\sf tot}$           &11.47	      &5.79	      &6.14	    \\
				& ${\sf CT}_{\sf tot}$           &30.88	      &14.82	  &13.40	\\
				& ${\sf Comm}_{\sf tot}$         & 0.28       & 0.39      & 0.24    \\
				& ${\sf Cost}$                   & 70.00	  &75.67	  &49.16	\\
				& $\TP$                          & 428.16     & 652.75    & 644.69  \\
				\midrule
				\multirow{4}{*}{NN-3} 
				& ${\sf PT}_{\sf on}$            & 21.79      & 5.67      & 8.40    \\
				& ${\sf PT}_{\sf tot}$           &30.66	      &15.14	  &17.87	\\
				& ${\sf CT}_{\sf tot}$           &91.68	      &40.01	  &42.76	\\
				& ${\sf Comm}_{\sf tot}$         & 1.59       & 1.94      & 1.28    \\
				& ${\sf Cost}$                   &331.01	  &343.73	  &240.41	\\
				& $\TP$                          & 53.62      & 55.71     & 54.13   \\
				\midrule
				\multirow{4}{*}{NN-4} 
				& ${\sf PT}_{\sf on}$            & 72.01      & 25.90     & 38.35    \\
				& ${\sf PT}_{\sf tot}$           &283.89	  &182.13	  &194.58	 \\
				& ${\sf CT}_{\sf tot}$           &859.09	  &500.13	  &522.32	 \\
				& ${\sf Comm}_{\sf tot}$         & 31.59      & 29.52     & 22.24    \\
				& ${\sf Cost}$                   &5779.27	  &5146.10	  &3999.30	 \\
				& $\TP$                          & 2.55       & 2.61      & 2.56     \\
				\bottomrule
			\end{NiceTabular}
		}
		\vspace{-1mm}
		\caption{\small Benchmarking of the training phase of ML algorithms. Time~(in seconds) and communication~(in GB) are reported for $1$ iteration. Monetary cost~(USD) is reported for $1000$ iterations.\label{tab:mltrain}}
		\vspace{-8mm}
	\end{table}


Starting with the time-optimized variant, $\thisT$ is $3 - 4\times$ faster than Trident in online runtime. 
The primary factor is the reduction in online rounds of our protocol due to multi-input gates. More precisely, we use the depth-optimized bit extraction circuit while instantiating the ReLU activation function using multi-input AND gates~(cf.~\S\ref{sec:4pcTools}). Looking at the total communication~(${\sf Comm}_{\sf tot}$) in \tabref{mltrain}, we observe that the gap in ${\sf Comm}_{\sf tot}$ between $\thisT$ vs. Trident decreases as the networks get deeper. This is justified as the improvement in communication of our dot product with truncation outpaces the overhead in communication caused by multi-input gates. The impact of this is more pronounced with NN-4, as observed by the lower monetary cost of $\thisT$ over Trident.
Another reason is that there are two active parties ($P_1, P_2$) in our framework, whereas Trident has three. Given the allocation of servers, the best $\rtt$ Trident can get with three parties~$(P_0,P_1,P_2)$ is $153.74ms$, as compared to $62.01ms$ of Tetrad, contributing to Tetrad being faster. However, if the $\rtt$ among all the parties were similar, this gap would be closed. Concretely, the online runtime (${\sf PT_{on}}$) of Trident will be similar to that of $\thisC$.  

The cost-optimized variant $\thisC$ on the other hand, is $1.5\times$ slower in the online phase compared to $\thisT$. However, it is still faster than Trident owing to the $\rtt$ setup, as discussed above. When it comes to monetary cost, this variant is up to $20-40\%$ cheaper than it's time-optimized counterpart and cheaper by around $30\%$ over Trident.

These trends can be better captured with a pictorial representation as given in \figref{MLTrain}.

\paragraph{Varying batch sizes and feature sizes}
\tabref{nn1} shows the online throughput~($\TP$) of neural network~(NN-1) training over varying batch sizes and feature sizes using synthetic datasets. 

	\begin{table}[htb!]
		\centering
		\resizebox{0.43\textwidth}{!}{
			\begin{NiceTabular}{r r r r r}
				\toprule
				Batch Size & Features & Trident & \thisT & \thisC \\
				\midrule 
				\Block{3-1}{128} 
				&10	&1905.58	&5407.35	&5271.88	
				\\
				&100	&1905.58	&5152.29	&5029.14	
				\\
				&1000	&1904.4	&3500.89	&3443.6	
				\\
				\midrule
				\Block{3-1}{256} 
				&10	&1905.58	&2818.4	&2744.87	
				\\
				&100	&1905.58	&2747.5	&2677.58	
				\\
				&1000	&1849.78	&2195.3	&2150.43	
				\\
				\bottomrule
			\end{NiceTabular}
		}
		\vspace{-1mm}
		\caption{\small Online throughput~($\TP$) of NN-1 training~(iterations per minute) over various batch sizes and features.\label{tab:nn1}}
		\vspace{-5mm}
	\end{table}

We find that both $\thisT, \thisC$ are up to $1.8 \times$ higher in $\TP$. However, as the batch size and feature size increase, both Trident and $\this$ experience a bandwidth bottleneck. The effect of the bandwidth limitation is higher for $\this$; hence the gain in $\TP$ over Trident decreases a bit.

\subsection{ML Inference}
\label{subsec:bench_inf}

We benchmark the inference phase of SVM and the aforementioned NNs. In addition to Trident~\cite{NDSS:ChaRacSur20}, we also benchmark against the 4PC robust protocol of SWIFT~\cite{USENIX:KPPS21} since it supports NN inference. Note that the best case performance of Fantastic Four~\cite{USENIX:DalEscKel20} when cast in the preprocessing model resembles that of SWIFT, while their worst case execution (3PC malicious) is an order of magnitude slower (cf. \S\ref{pa:fantasticfour}), as demonstrated in their paper (cf. Table 2 of~\cite{USENIX:DalEscKel20}). 

	\begin{table}[htb!]
		\centering
		\resizebox{0.475\textwidth}{!}{
			\begin{NiceTabular}{r r r r r r}
				\toprule
				Algorithm & Parameter & Trident & \thisT & \thisC & SWIFT\\
				\midrule 
				\multirow{4}{*}{SVM} 
				& ${\sf PT}_{\sf on}$            & 17.09     & 2.91      & 4.77         & 5.21\\
				& ${\sf PT}_{\sf tot}$           & 17.37     & 3.19      & 5.05         & 6.04\\
				& ${\sf CT}_{\sf tot}$           & 47.02     & 6.99      & 10.70        & 14.47\\
				& ${\sf Comm}_{\sf tot}$         & 1.36      & 2.34      & 1.25         & 1.36\\
				& ${\sf Cost}$                   & 39.92     & 6.26      & 9.23         & 12.43\\
				& $\TP$                          & 898.80    & 5271.74   & 3221.29      & 2949.76\\
				\midrule
				\multirow{4}{*}{NN-1} 
				& ${\sf PT}_{\sf on}$            & 5.87       & 1.31         & 1.87     & 2.31\\
				& ${\sf PT}_{\sf tot}$           & 6.15       & 1.58         & 2.14     & 3.13\\
				& ${\sf CT}_{\sf tot}$           & 16.75      & 3.76         & 4.88     & 8.65\\
				& ${\sf Comm}_{\sf tot}$         & 0.06       & 0.09         & 0.05     & 0.06\\
				& ${\sf Cost}$                   & 14.15      & 3.19         & 4.13     & 7.32\\
				& $\TP$                          & 2615.35    & 11734.60     & 8226.93  & 6661.00\\
				\midrule
				\multirow{4}{*}{NN-2} 
				& ${\sf PT}_{\sf on}$            & 5.87       & 1.31         & 1.87     & 2.31\\
				& ${\sf PT}_{\sf tot}$           & 6.15       & 1.58         & 2.14     & 3.13\\
				& ${\sf CT}_{\sf tot}$           & 16.75      & 3.77         & 4.88     & 8.66\\
				& ${\sf Comm}_{\sf tot}$         & 0.26       & 0.37         & 0.22     & 0.25\\
				& ${\sf Cost}$                   & 14.19      & 3.24         & 4.16     & 7.35\\
				& $\TP$                          & 2615.35    & 11734.60     & 8226.93  & 6661.00\\
				\midrule
				\multirow{4}{*}{NN-3} 
				& ${\sf PT}_{\sf on}$            & 14.42      & 2.61         & 4.10     & 4.54\\
				& ${\sf PT}_{\sf tot}$           & 14.71      & 2.91         & 4.39     & 5.39\\
				& ${\sf CT}_{\sf tot}$           & 39.92      & 6.43         & 9.40     & 13.18\\
				& ${\sf Comm}_{\sf tot}$         & 5.62       & 8.42         & 4.76     & 5.39\\
				& ${\sf Cost}$                   & 34.59      & 6.74         & 8.68     & 11.97\\
				& $\TP$                          & 1065.35    & 5882.44      & 3746.89  & 3384.51\\
				\midrule
				\multirow{4}{*}{NN-4} 
				& ${\sf PT}_{\sf on}$            & 47.05      & 7.85      & 12.69       & 13.13\\
				& ${\sf PT}_{\sf tot}$           & 47.61      & 8.44      & 13.28       & 14.33\\
				& ${\sf CT}_{\sf tot}$           & 129.41     & 17.77     & 27.46       & 31.35\\
				& ${\sf Comm}_{\sf tot}$         & 85.69      & 124.09    & 71.27       & 81.33\\
				& ${\sf Cost}$                   & 122.66     & 34.40     & 34.32       & 39.18\\
				& $\TP$                          & 326.46     & 934.34    & 891.19      & 891.19\\
				\bottomrule
			\end{NiceTabular}
		}
		\vspace{-1mm}
		\caption{\small Benchmarking of the inference phase of ML algorithms. Time~(in seconds) and communication~(in MB) are reported for $1$ query. Monetary cost~(USD) is reported for $1000$ queries.\label{tab:mlinf}}
		\vspace{-4mm}
	\end{table}

Similar to training, the time-optimized variant for inference is faster when it comes to ${\sf PT}_{\sf on}$, by $4 - 6\times$ over Trident. This is also reflected in the $\TP$, where the improvement is about $2.8 - 5.5\times$, as evident from \figref{MLTP}. In inference, the communication is in the order of megabytes, while run time is in the order of a few seconds. The key observation is that communication is well suited for the bandwidth used~(40 MBps). So unlike training, the monetary cost in inference depends more on run time rather than on communication. This is evident from \tabref{mlinf} which shows that $\thisT$ saves on monetary cost up to a factor of $6$ over Trident.

Note that the cost-optimized variant under performs in terms of monetary cost compared to $\thisT$. This is because, as mentioned earlier, run time plays a bigger role in monetary cost than communication. Hence for inference, the time-optimized variant becomes the optimal choice.

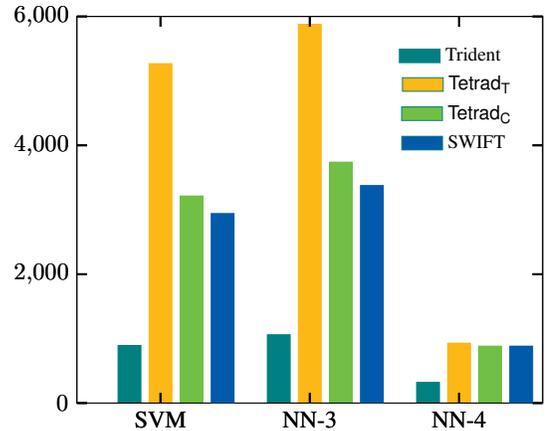
\begin{figure}[htb!]
	\centering
		\resizebox{.4\textwidth}{!}{
			\begin{tikzpicture}[
				every axis/.style={ 
					ybar stacked,
					ymin=0,ymax=6000,
					xtick={1,2,3}, xticklabels={SVM,NN-3,NN-4},
					enlarge x limits=0.28,
					cycle list name=exotic, 
					every axis plot/.append style={fill,draw=none,no markers},
					legend style = {anchor = south, legend columns = -1, draw=none, area legend},
					bar width=10pt},]
				
				\begin{axis}[bar shift=-13pt,hide axis, legend style = {at={(0.82, 0.85)}}]
					\addplot+ coordinates
					{(1,898.8) (2,1065.35) (3,326.46)}; 	
					\addlegendentry{{\footnotesize Trident~~~}}
				\end{axis}
				
				\begin{axis}[hide axis, legend style = {at={(0.82, 0.775)}}]
					\addplot+[fill=UniOrange] coordinates
					{(1,5271.74) (2,5882.44) (3,934.34)}; 	
					\addlegendentry{{\footnotesize \thisT~}}
				\end{axis}
				
				\begin{axis}[bar shift=13pt, legend style = {at={(0.82, 0.7)}}]
					\addplot+[fill=UniGruen] coordinates
					{(1,3221.29) (2,3746.89) (3,891.19)};  
					\addlegendentry{\footnotesize \thisC~}
				\end{axis}
				
			    \begin{axis}[bar shift=26pt, legend style = {at={(0.82, 0.625)}}]
			    	\addplot+[fill=UniBlau] coordinates
			    	{(1,2949.76) (2,3384.51) (3,891.19)}; 		
			    	\addlegendentry{\footnotesize SWIFT~~}
			    \end{axis}
			
			\end{tikzpicture}
		}
	\caption{\small Inference of SVM, NN-3 and NN-4: in terms of $\TP$ (higher is better)}\label{fig:MLTP}
\end{figure}

\vspace{-3mm}
\subsection{Comparison operations}
\tabref{compb} compares the performance of the frameworks for circuits of varying depth. At each layer of the circuits, we perform 128 comparisons where the comparison results are generated in arithmetic shared form. The idea is that each layer emulates a comparison layer in an NN with a batch size of 128. 

	\begin{table}[htb!]
		\centering
		\resizebox{0.41\textwidth}{!}{
			\begin{NiceTabular}{r r r r r}
				\toprule
				Depth & Parameter & Trident & \thisT & \thisC\\
				\midrule 
				\Block{3-1}{128} 
				& ${\sf PT}_{\sf on}$           &3.55	&0.53	&0.93	\\
				& ${\sf CT}_{\sf tot}$          &9.6	&1.06	&1.85	\\
				& ${\sf Cost}$                  &0.49	&0.05	&0.09	\\
				\midrule
				\Block{3-1}{1024} 
				& ${\sf PT}_{\sf on}$           &28.42	&4.23	&7.41	\\
				& ${\sf CT}_{\sf tot}$          &76.79	&8.47	&14.82	\\
				& ${\sf Cost}$                  &3.89	&0.43	&0.75	\\
				\midrule
				\Block{3-1}{8192} 
				& ${\sf PT}_{\sf on}$           &227.34	&33.87	&59.27	\\
				& ${\sf CT}_{\sf tot}$          &614.3	&67.76	&118.56	\\
				& ${\sf Cost}$                  &31.27	&3.48	&6.03	\\
				\bottomrule
			\end{NiceTabular}
		}
		\vspace{-1mm}
		\caption{\small Benchmarking of comparisons over various depths. Each of the layer has 128 comparisons. Time is reported in minutes, and monetary cost in USD.\label{tab:compb}}
		\vspace{-4mm}
	\end{table}

Interestingly, beyond a depth of roughly 100, the time-optimized variant~($\thisT$) starts outperforming in every metric, especially monetary cost, over the cost-optimized one~($\thisC$). This is because as the depth increases, runtime~({\sf CT}) grows at a much higher rate than the total communication. What we can infer from \tabref{compb} is that if one were to use a DNN with a depth of over 100, $\thisT$ becomes the optimal choice. 

\section*{Future Work}
\label{sec:4pcConclusion}
$\this$ requires the preprocessing to be function-dependent. Decoupling the preprocessing from the function to be computed in the online phase will make the framework more generic and is left as an interesting direction to pursue. 
Even though fixed-point arithmetic is efficient for the applications considered, in some cases, other representations such as floating-point and posit arithmetic might be desirable. Supporting alternative representations may require rethinking parts of the framework; hence it is left as an open problem.

The following are some of the challenges to be addressed while extending $\this$ to support training of other ML algorithms such as SVM, ResNet and LSTMs. 
In SVM training, the choice of kernel function plays an important role in determining the efficiency, especially for the non-linear classifiers. Some of the most widely used non-linear kernels include i) Polynomial: $(\vct{x} \band \vct{y})^d$, ii) Gaussian: $\exp(-\gamma\| \vct{x} - \vct{y}\|^2)$ for $\gamma > 0$, and iii) Hyperbolic: $\tanh(\mu \vct{x} \band \vct{y} + c)$ for some $\mu > 0$ and $c < 0$, where $\vct{x}, \vct{y}$ denote the input vectors. These kernels are expensive to compute (computation and communication) using standard MPC approaches such as circuit garbling, and hence, demand new MPC-friendly protocols which guarantee efficiency without losing out on accuracy (e.g., Sigmoid approximation of~\cite{SP:MohZha17}). 
Further, note that using the naive MPC protocols for training would demand a non-linear increase in bit-size of fixed-point arithmetic to accommodate for an increased dataset size~\cite{CVE21PPMLCRYPTO}.  Concretely, for a dataset with only 212 entries and 14 features, the ring size should be at least 246 bits. Thus, it is necessary to redesign the protocols to enable computation within the standard ring sizes.
For deep networks such as ResNet and LSTMs, they require performing batch normalization multiple times, each of which involves division and square-root operations~\cite{PoPETS:WTBKMR21}. Since the latter is expensive to perform over rings, designing efficient protocols for these operations is an interesting question.

Finally, although it is known how to instantiate the required primitives securely using standard MPC techniques, they are far from being practically efficient. Moreover, since the secure variant is known to have an overhead over the plaintext computation, sophisticated techniques are required to handle the large amount of intermediate data generated while training very deep networks. Existing PPML frameworks lack support for training the above ML algorithms to the best of our knowledge. We believe that accounting for the points above can bring the existing PPML frameworks, including $\this$, one step closer to the efficient realization of these algorithms.

\section*{Acknowledgements}
The authors would like to acknowledge support from Google PhD Fellowship 2019, Centre for Networked
Intelligence (a Cisco CSR initiative) 2021, SERB MATRICS (Theoretical Sciences) Grant 2020 and Google India AI/ML Research Award 2020. The authors would also like to acknowledge the financial support from Google Cloud to perform the benchmarking.

This project has received funding from the European Research Council (ERC) under the European Union’s Horizon 2020 research and innovation program (grant agreements No. 850990~(PSOTI) and No. 803096~(SPEC)) and from the Digital Research Centre Denmark (DIREC).  This work was co-funded by the Deutsche Forschungsgemeinschaft~(DFG) – SFB~1119 CROSSING/236615297.

\bibliographystyle{IEEEtran}
\bibliography{others,cryptobib/abbrev3,cryptobib/crypto}

\appendices
\section{Preliminaries}
\label{app:Prelims}

\subsection{Related Work}
\label{app:related}
Related work covers MPC protocols with an honest majority for high-throughput and constant-round setting and mixed-protocol frameworks for the case of PPML. 

ABY3~\cite{CCS:MohRin18} was the first framework for the case of 3 parties, supporting both training and inference. It had variants for both passive and active security, with the former being based on~\cite{CCS:AFLNO16} and the latter on~\cite{EC:FLNW17,SP:ABFLLN17}.
ASTRA~\cite{CCSW:CCPS19} improved upon the 3PC of~\cite{CCS:AFLNO16,EC:FLNW17,SP:ABFLLN17} by proposing faster protocols for the online phase with active security. As a result, secure inference of ASTRA is faster than ABY3. Building on~\cite{C:BBCGI19}, BLAZE~\cite{NDSS:PatSur20} proposed an actively secure framework that supports inference of neural networks. BLAZE pushes the expensive zero-knowledge part of the computation to the preprocessing phase, making its online phase faster than that of~\cite{C:BBCGI19}. SWIFT~(3PC) improved upon BLAZE by using the distributed zero-knowledge protocol of \cite{CCS:BGIN19}, thereby achieving GOD.
In an orthogonal line of work, FALCON~\cite{PoPETS:WTBKMR21} focused on enhancing the efficiency of actively secure protocols for large convolutional neural networks, supporting training and inference. 

In the high-throughput setting for 4PC, ~\cite{AC:GorRanWan18} explores protocols for the security notions of abort. Inspired by the theoretical GOD construction in~\cite{AC:GorRanWan18}, FLASH proposed practical protocols with GOD for secure inference.  Trident~\cite{NDSS:ChaRacSur20} improved protocols (in terms of communication) compared to~\cite{AC:GorRanWan18} with a focus on security with fairness. In addition, it was the first work to propose a mixed-protocol framework for the case of 4 parties. More recently,~\cite{USENIX:MLRG20} improved over~\cite{AC:GorRanWan18} to provide support for fixed-point arithmetic with applications to graph parallel computation, albeit with abort security.

Improving the security of Trident to GOD, SWIFT~\cite{USENIX:KPPS21} presented an efficient, robust PPML framework with protocols as fast as Trident. SWIFT only supports the secure inference of neural networks and lacks conversions similar to the ones from Trident and the garbled world. Fantastic Four~\cite{USENIX:DalEscKel20} also provides robust 4PC protocols which are on par with SWIFT. While they claim to provide a better security model called {\em private robustness} compared to SWIFT, it has been shown in SWIFT that the two security models are theoretically equivalent. Our security model is also similar to SWIFT, and we elaborate on its equivalence to private robustness in \S\ref{app:secmodel}.

In the regime of constant-round protocols,~\cite{CCS:MohRosZha15} presents 3PC protocols in the honest majority setting satisfying security with abort, which require communicating one garbled circuit and three rounds of interaction. The work of~\cite{C:IKKP15} presents a robust 4-party computation protocol (4PC) with GOD in $2$-rounds (which is optimal) at the expense of 12 garbled circuits. Further,~\cite{CCS:BJPR18} presents efficient 3PC and 4PC constructions providing security notions of fairness and GOD. 

A mixed-protocol framework for MPC was first shown to be practical, in the 2-party dishonest majority setting, by TASTY~\cite{CCS:HKSSW10}. TASTY was a passively secure compiler supporting generation of protocols based on homomorphic encryption and garbled circuits. This was followed by ABY~\cite{NDSS:DemSchZoh15}, which proposed a mixed protocol framework, also with passive security, combining the arithmetic, boolean and garbled worlds. The recent work of ABY2~\cite{USENIX:PSSY21} improves upon the ABY framework, providing a faster online phase with applications to PPML.
The work of~\cite{INDOCRYPT:RotWoo19,C:EGKRS20} proposed efficient mixed world conversions for the case of $n$ parties with a dishonest majority. Both works have active security, with \cite{INDOCRYPT:RotWoo19} supporting the inference of SVMs, and~\cite{C:EGKRS20} supporting neural network inference.

In the honest majority setting, ABY3~\cite{CCS:MohRin18} extended the idea to 3 parties and provided specialised protocols for the case of PPML. ABY3 was the first work to support secure training in the case of 3 parties, while Trident~\cite{NDSS:ChaRacSur20} extended it to the 4-party setting.

\subsection{Basic Primitives}
\label{app:basicprimitives}

\paragraph{Shared Key Setup}
\label{app:keysetup}
Let $F : \{0, 1\}^{\csec} \times \{0, 1\}^{\csec} \rightarrow X$ be a secure  pseudo-random function (PRF), with co-domain $X$ being $\Z{\ell}$. The following set of keys are established between the parties.
\begin{enumerate}[itemsep=0mm]
	\item[--] One key between every pair -- $\Key{ij}$ for $P_i, P_j$.
	\item[--] One key between every set of three parties -- $\Key{ijk}$ for $P_i, P_j, P_k$.
	\item[--] One shared keys $\Key{\Partyset}$ known to all parties in $\Partyset$.
\end{enumerate}

Suppose $P_0,P_1$ wish to sample a random value $r \in \Z{\ell}$ non-interactively. To do so they invoke $F_{k_{01}}(id_{01})$ and obtain $r$. Here, $id_{01}$ denotes a counter maintained by the parties, and is updated after every PRF invocation. The appropriate keys used to sample is implicit from the context, from the identities of the pair that sample or from the fact that it is sampled by all, and, hence, is omitted.
\vspace{-3mm}
\begin{systembox}{$\Func[Setup]$}{Ideal functionality for shared-key setup}{fig:FSETUP}
	\justify
	$\Func[Setup]$ interacts with the parties in $\Partyset$ and the adversary $\Sim$. $\Func[Setup]$ picks random keys $\Key{ij}$ and $\Key{ijk}$ for $i,j,k \in \{0,1,2,3\}$ and $\Key{\Partyset}$. Let $\vl{y}_s$ denote the keys corresponding to party $P_s$. Then
	\begin{myitemize}
		\item[--] $\vl{y}_s = (\Key{01}, \Key{02}, \Key{03}, \Key{012}, \Key{013}, \Key{023}$ and $\Key{\Partyset})$ when $P_s = P_0$.
		\item[--] $\vl{y}_s = (\Key{01}, \Key{12}, \Key{13}, \Key{012},  \Key{013}, \Key{123}$ and $\Key{\Partyset})$ when $P_s = P_1$.
		\item[--] $\vl{y}_s = (\Key{02}, \Key{12}, \Key{23}, \Key{012},  \Key{023}, \Key{123}$ and $\Key{\Partyset})$ when $P_s = P_2$.
		\item[--] $\vl{y}_s = (\Key{03}, \Key{13}, \Key{23}, \Key{013}, \Key{023}, \Key{123}$ and $\Key{\Partyset})$ when $P_s = P_3$.
	\end{myitemize}
	\begin{description}
		\item {\bf Output: } Send $(\OUTPUT, \vl{y}_s)$ to every $P_s \in \Partyset$.
	\end{description}
\end{systembox}

The key setup is modelled via a functionality $\Func[Setup]$ (\boxref{fig:FSETUP}) that  can be realised using any secure  MPC protocol. A simple instantiation of such an MPC protocol is as follows. $P_i$ samples key $\Key{ij}$ and sends to $P_j$. $P_i$ samples $\Key{ijk}$ and sens to $P_j$. $P_i, P_j$ $\jsend$ $\Key{ijk}$ to $P_k$. Similarly, $P_0$ samples $\Key{\Partyset}$ and sends to $P_3$. $P_0, P_3$ $\jsend$ $\Key{\Partyset}$ to $P_1$ and $P_2$.

\paragraph{Collision-Resistant Hash Function~\cite{FSE:RogShr04}}
\label{app:hash}. 
A family of hash functions $\{\Hash: \mathcal{K} \times \MS \rightarrow \mathcal{Y} \}$ is said to be collision resistant if for all PPT adversaries $\Adv$, given the hash function $\Hash_k$ for $k \in_R \mathcal{K}$, the following holds: $\Prob[(x, x^{\prime}) \leftarrow \Adv(k) : (x \neq x^{\prime}) \wedge \Hash_k(x) = \Hash_k(x^{\prime})] = \negl(\kappa)$, where $x, x^{\prime} \in \{0,1\}^{m}$ and $m = \poly(\kappa)$.

\subsection{Security Model}
\label{app:secmodel}
We prove security using the real-world/ ideal-word simulation paradigm~\cite{Goldreich04, EPRINT:Lindell16}. The security is analyzed by comparing what an adversary can do in the real world's execution of the protocol with what it can do in an ideal world execution where there is a trusted third party and is considered secure by definition. In the ideal world, the parties send their inputs to the trusted third party over perfectly secure channels that carries out the computation and sends the output to the parties. Informally, a protocol is secure if whatever an adversary can do in the real world can also be done in the ideal world. 

Let $\Adv$ denote the probabilistic polynomial time ($\ppt$) real-world adversary corrupting at most one party in $\Partyset$, $\Sim$ denote the corresponding ideal world adversary, and $\Func{ }$ denote the ideal functionality. Let $\Ideal_{\Func{ }, \Sim}(\onesec, z)$ denote the joint output of the honest parties and $\Sim$ from the ideal execution with respect to the security parameter $\csec$ and auxiliary input $z$. Similarly, let $\Real_{\Pi, \Adv}(\onesec, z)$ denote the joint output of the honest parties and $\Adv$ from the real world execution. We say that the protocol $\Pi$ securely realizes $\Func{ }$ if for every $\ppt$ adversary $\Adv$ there exists an ideal world adversary $\Sim$ corrupting the same parties such that $\Ideal_{\Func{ }, \Sim}(\onesec, z)$ and $\Real_{\Pi, \Adv}(\onesec, z)$ are computationally indistinguishable. The ideal functionality for computing a function $f$ with fairness and robustness appears in \boxref{ideal:fair} and \boxref{ideal:god}, respectively. 

\begin{systembox}{$\Func[Fair]$}{Fair functionality for computing function $f$}{ideal:fair}
	Every honest party $P_i \in \Partyset$ sends its input $x_i$ to the functionality. Corrupted parties may send arbitrary inputs as instructed by the adversary. While sending the inputs, the adversary is also allowed to send a special $\abort$ command.
	
	\noindent \textbf{Input:} On message $(\INPUT, x_i)$ from $P_i$, do the following: if $(\INPUT, \ast)$ already received from $P_i$, then ignore the current message. Otherwise, record $x_i^{\prime} = x_i$ internally. If $x_i$ is outside $P_i$'s domain, consider $x_i^{\prime} = \abort$.
	
	\noindent \textbf{Output:} If there exists an $i \in \{0, 1, 2, 3\}$ such that $x_i^{\prime} = \abort$, send $(\OUTPUT, \bot)$ to all the parties. Else, compute $y = f(x_0^{\prime}, x_1^{\prime}, x_2^{\prime}, x_3^{\prime})$ and send $(\OUTPUT, y)$ to all parties.
\end{systembox}

\vspace{-2mm}
\begin{systembox}{$\Func[Robust]$}{Robust functionality for computing function $f$}{ideal:god}
	Every honest party $P_i \in \Partyset$ sends its input $x_i$ to the functionality. Corrupted parties may send arbitrary inputs as instructed by the adversary.
	
	\noindent \textbf{Input:} On message $(\INPUT, x_i)$ from $P_i$, do the following: if $(\INPUT, \ast)$ already received from $P_i$, then ignore the current message. Otherwise, record $x_i^{\prime} = x_i$ internally. If $x_i$ is outside $P_i$'s domain, consider $x_i^{\prime}$ to be some predetermined default value.
	
	\noindent \textbf{Output:} Compute $y = f(x_0^{\prime}, x_1^{\prime}, x_2^{\prime}, x_3^{\prime})$ and send $(\OUTPUT, y)$ to all parties.
\end{systembox}

\paragraph{On the security of robust $\this$} 
\label{sec:privaterobust}
We emphasize that we follow the standard traditional (real-world / ideal-world based) security definition of MPC, according to which, in the 4-party setting with one corruption, exactly one party is assumed to be corrupt, and the rest are {\em honest}. As per this definition, disclosing the honest parties' inputs to a selected {\em honest} party is {\em not} a breach of security. Indeed in $\this$, the data sharing and the computation on the shared data are done so that any malicious behaviour leads to establishing a trusted party $\TTP$ who is enabled to receive all the inputs and compute the output on the clear. 
There has been a recent study on the additional requirement of hiding the inputs from a quorum of honest parties (treating them as semi-honest), termed as  Friends-and-Foes (FaF) security notion~\cite{C:AloOmrPas20}. This is a stronger security goal than the standard one. 
Informally, designing secure 4PC FaF protocols requires security against two independent corruptions. Our sharing semantics, designed to handle only one corruption, does not suffice. Hence, we leave FaF-secure 4PC for future exploration. 

Another security notion, called {\em private robustness}, was recently proposed in the work of Dalskov et al.~\cite{USENIX:DalEscKel20}, where the protocol does not demand the inputs be sent to a $\TTP$. Their work, however, considers a more restricted security model, where it is assumed that parties will discard messages which are {\em non-intended} and are not a part of the protocol. This involves assuming a {\em secure erasure}. Under this assumption, our model is equivalent to that of private robustness since the trusted party $\TTP$ will erase the input of the honest parties after computing the function output.

\subsection{Comparison with  Fantastic Four~\cite{USENIX:DalEscKel20}} 
\label{pa:fantasticfour}
We analyse the performance of Fantastic Four~\cite{USENIX:DalEscKel20} where execution proceeds in segments~(cf. \S6.4,~\cite{USENIX:DalEscKel20}). Elaborately, computation is carried out optimistically for each segment, followed by a verification phase before proceeding to the next segment. If verification fails, the current segment is recomputed via an active 3PC protocol. Subsequent segments also proceed with a 3PC execution until the verification fails again. In this case, a semi-honest 2PC with a helper is carried out for the current and rest of the segments. For analysis, we consider their best and worst-case execution cost. 

\begin{table}[htb!]
	\centering
	\resizebox{0.48\textwidth}{!}{
		\begin{NiceTabular}{r r r c}
			\toprule
			\Block{2-1}{Protocol} & \Block[c]{1-2}{Dot Product w/ Truncation} &
			& \Block[c]{2-1}{\#Active\\Parties}\\ \cmidrule{2-3}
			& Preprocessing & Online & \\
			\midrule
			Fantastic Four: Case I    & $\ell$  & $9\ell$  & 4 \\
			Fantastic Four: Case II   & $76(\ell+\kappa)+54x + 12$  & $9\ell + 6\kappa$  & 3 \\
			\thisA (on-demand)       & -        & $5\ell$  & 3 \\                     
			\bottomrule
		\end{NiceTabular}
	}
	\vspace{-1mm}
	\caption{\small Comparison with Fantastic Four~\cite{USENIX:DalEscKel20}}\label{tab:compFour}
\end{table}

Observe that the best case happens when the verification is always successful, which we call as {\em Case I}. In this case, the communication cost is that of the 4PC execution. Note that an adversary can {\em always} make the verification fail in the first segment itself. This results in executing the entire protocol (all segments) with their active 3PC, which accounts for their worst-case cost. We denote this as {\em Case II}. Their 3PC protocols are designed to work over the extended ring of size $\ell + \kappa$ bits. As evident from Tables 2, 3 of their paper, their 3PC is at least $10 \times$ more expensive than their 4PC in terms of both runtime and communication. Thus, the higher cost of 3PC defeats the purpose of having an additional honest party in the system. 

Observe that their protocols are designed to work with a function-independent preprocessing. Thus, for a fair comparison, we compare both cases against the on-demand variant of our robust protocols~(\thisA). The results are summarised in~\tabref{compFour}. We remark that the values for their cases are obtained from Table 1 of their paper~\cite{USENIX:DalEscKel20}. 
\section{4PC Protocol}
\label{app:4pc}
Here we detail the additional information regarding the 4PC protocols.

\paragraph{Joint-send for robust protocols}
\label{app:jsendrobust}
The formal protocol for $\prot{\jsend}$ in the robust setting~\cite{USENIX:KPPS21} is given in \boxref{fig:4pcjsendrobust}.

\begin{lemma}[Communication]
	\label{appl:pijsend}
	Protocol $\prot{\jsend}$~(\boxref{fig:4pcjsendrobust}) requires an amortized communication of $\ell$ bits and $1$ round.
\end{lemma}
\begin{proof}
	In the protocol $\prot{\jsend}(P_i, P_j,\vl{v},P_k)$ for the fair variant, $P_i$ communicates $\vl{v}$ to $P_k$ requiring communication of $\ell$ bits and one round. The hash value communication from $P_j$ to $P_k$ can be clubbed for multiple instances with the same set of parties and hence the cost gets amortized. The analysis is similar for the robust case as well. Here, though the verification consists of multiple steps, the cost gets amortized over multiple instances.
\end{proof}

\medskip
\begin{protocolbox}{$\prot{\jsend}(P_i, P_j,\vl{v},P_k)$}{Joint-Send for robust protocols}{fig:4pcjsendrobust}
	\justify
	$P_s \in \Partyset$ initializes an inconsistency bit $\bitb_s = 0$. If $P_s$ remains silent instead of sending $\bitb_s$ in any of the following rounds, the recipient sets $\bitb_s$ to $1$.
	
	\begin{myitemize} 
	\item[--] {\em Send: } $P_i$ sends $\vl{v}$ to $P_k$.
	
	\item[--] {\em Verify: } $P_j$ sends $\Hash(\vl{v})$ to $P_k$. 
	\begin{itemize}
		\item[-]  $P_k$ sets $\bitb_k = 1$ if the received values are inconsistent or if the value is not received. 
		\item[-] $P_k$ sends $\bitb_k$ to all parties. $P_s$ for $s \in \{i, j, l\}$ sets $\bitb_s = \bitb_k$.
		\item[-] $P_s$ for $s \in \{i, j, l\}$ mutually exchange their bits. $P_s$ resets $\bitb_s = \bitb^{\prime}$ where $\bitb^{\prime}$ denotes the bit which appears in majority among $\bitb_i, \bitb_j, \bitb_l$.
		\item[-] All parties set $\TTP = P_l$ if $\bitb^{\prime} = 1$, terminate otherwise.
	\end{itemize}
	\end{myitemize} 
\end{protocolbox}

\paragraph{Sharing Protocol}
\label{app:share}

\begin{lemma}[Communication]
	\label{appl:pish}
	Protocol $\prot{\Sh}$~(\boxref{fig:piSh}) requires an amortized communication of at most $3\ell$ bits and $1$ round in the online phase.
\end{lemma}
\begin{proof}
	The preprocessing of $\prot{\Sh}$ is non-interactive as the parties sample non interactively using key setup $\Func[Setup]$~(\S\ref{app:basicprimitives}). in the online phase, $P_i$ sends $\mk{\vl{v}}$ to $P_1, P_2, P_3$ resulting in 1 round and communication of at most $3\ell$ bits~($P_i = P_0$). The next round of hash exchange can be clubbed for several instances and the cost gets amortized over multiple instances.
\end{proof}

\paragraph{Reconstruction Protocol}
\label{app:rec}

\begin{lemma}[Communication]
	\label{appl:pirec}
	Protocol $\prot{\Rec}$~(\boxref{fig:piRec}) requires an amortized communication of $4\ell$ bits and $1$ round in the online phase.
\end{lemma}
\begin{proof}
	The protocol involves 4 invocations of $\prot{\jsend}$ protocol and the communication follows from Lemma~\ref{appl:pijsend}.
\end{proof}

\medskip
\begin{protocolbox}{$\prot{\Rec}(\Partyset, \shr{\vl{v}})$}{Reconstruction (with abort) of $\vl{v}$ among $\Partyset$.}{fig:piRec}
	\justify
	\detail{
		{\bf Input(s):} $\shr{\vl{v}}$.\\
		{\bf Output:} $\vl{v}$.
	}
	\begin{enumerate}
		\item $P_1, P_0$ $\jsend$ $\pad{\vl{v}}{1}$ to $P_2$;~~~$P_2, P_0$ $\jsend$ $\pad{\vl{v}}{3}$ to $P_3$;\newline $P_3, P_0$ $\jsend$ $\pad{\vl{v}}{2}$ to $P_1$;~~~$P_1, P_2$ $\jsend$ $\mk{\vl{v}}$ to $P_0$.
		\item Compute $\vl{v} = \mk{\vl{v}} - \pad{\vl{v}}{1} - \pad{\vl{v}}{2} - \pad{\vl{v}}{3}$. 
	\end{enumerate}
\end{protocolbox}

\paragraph{Multiplication Protocol}
\label{app:mult}

\begin{lemma}[Communication]
	\label{appl:piMultF}
	Protocol $\piMult$~(\boxref{fig:piMultiplication})~(in $\this$) requires $2 \ell$ bits of communication in the preprocessing phase, and $1$ round and $3 \ell$ bits of communication in the online phase.
\end{lemma}
\begin{proof}
	During preprocessing, sampling of values ${\vl{u}}^1, {\vl{u}}^2$ are performed non-interactively using $\Func[Setup]$. A communication of $\ell$ bits is required for the joint sharing of $\vl{q}$ by $P_0, P_3$ as explained in \S\ref{p:jsh}. In addition, $P_0$ communicates $\vl{w}$ to $P_3$ requiring additional $\ell$ bits. 
	During online, two instances of $\prot{\jsend}$ are executed in parallel resulting in a communication of $2\ell$ bits and 1 round. This is followed by a joint sharing by $P_1,P_2$ to $P_3$ for which an additional communication of $\ell$ bits are required. However, in joint sharing, the communication is from $P_1$ to $P_3$ and the same can be deferred till the verification stage. Thus the online round is retained as $1$ in an amortized sense. 
\end{proof}

\begin{lemma}[Communication]
	\label{appl:piMultR}
	Protocol $\piMult$~(\boxref{fig:piMultiplication})~(in $\thisA$) requires $2 \ell$ bits of communication in the preprocessing phase, and $1$ round and $3 \ell$ bits of communication in the online phase.
\end{lemma}

\subsection{Function-independent preprocessing}
\label{app:nopre}
We provide the fair multiplication, $\piMultO$, for {\em function-independent} preprocessing in \boxref{fig:piMultNoPre}. 
The protocol incurs no overhead over the fair multiplication~($\piMult$) in $\this$. This is due to the design of $\piMult$ where values ${\vl{u}}^1, {\vl{u}}^2$ are sampled non-interactively in the preprocessing. Thus the joint-sharing by $P_0, P_3$~(Step 5 (a) in \boxref{fig:piMultNoPre}) can be performed along with the communication among $P_1, P_2$~(Step 4 in \boxref{fig:piMultNoPre}) in the online. Moreover, the rest of the communication can be deferred till the verification stage and thus, the online round complexity is retained. The protocol for robust setting is similar.

\smallskip
\begin{protocolbox}{$\piMultO(\vl{a}, \vl{b}, \isTr)$}{Fair multiplication without preprocessing.}{fig:piMultNoPre}
	Let $\isTr$ be a bit that denotes whether truncation is required ($\isTr =1$) or not~($\isTr=0$). \\
	\detail{
		{\bf Input(s):} $\shr{\vl{a}}, \shr{\vl{b}}$.\\
		{\bf Output:} $\shr{\vl{o}}$ where $\vl{o} = \vl{z}^{\vl{t}}$ if $\isTr = 1$ and $\vl{o} = \vl{z}$ if $\isTr = 0$ and $\vl{z} = \vl{ab}$.
	}
	\justify 
	\vspace{-2mm}
	\algoHead{Online:} 
	\begin{enumerate}[itemsep=0mm]
		\item Locally compute the following:
		\begin{align*}
			P_0, P_1: \gm{\vl{a}\vl{b}}{1} &= \pad{\vl{a}}{1} \pad{\vl{b}}{3} + \pad{\vl{a}}{3} \pad{\vl{b}}{1} + \pad{\vl{a}}{3} \pad{\vl{b}}{3} \\
			P_0, P_2: \gm{\vl{a}\vl{b}}{2} &= \pad{\vl{a}}{2} \pad{\vl{b}}{3} + \pad{\vl{a}}{3} \pad{\vl{b}}{2} + \pad{\vl{a}}{2} \pad{\vl{b}}{2} \\
			P_0, P_3: \gm{\vl{a}\vl{b}}{3} &= \pad{\vl{a}}{1} \pad{\vl{b}}{2} + \pad{\vl{a}}{2} \pad{\vl{b}}{1} + \pad{\vl{a}}{1} \pad{\vl{b}}{1}
		\end{align*}
		\item $P_0, P_3$ and $P_j$ sample random ${\vl{u}}^j \in_R \Z{\ell}$ for $j \in \{1,2\}$. Let ${\vl{u}^1} + \vl{u}^2 = \gm{\vl{a}\vl{b}}{3} - \vl{r}$ for a random $\vl{r} \in_R \Z{\ell}$.  
		\item Let $\vl{y} = (\vl{z} - \vl{r}) - \mk{\vl{a}} \mk{\vl{b}}$. Locally compute the following:
		\begin{align*}
			P_1: \vl{y}_1 &= - \pad{\vl{a}}{1} \mk{\vl{b}} - \pad{\vl{b}}{1} \mk{\vl{a}} + \gm{\vl{a}\vl{b}}{1} + {\vl{u}}^1 \\
			P_2: \vl{y}_2 &= - \pad{\vl{a}}{2} \mk{\vl{b}} - \pad{\vl{b}}{2} \mk{\vl{a}} + \gm{\vl{a}\vl{b}}{2} + {\vl{u}}^2 \\
			P_1, P_2: \vl{y}_3 &= - \pad{\vl{a}}{3} \mk{\vl{b}} - \pad{\vl{b}}{3} \mk{\vl{a}}
		\end{align*}
		\item $P_1$ sends $\vl{y}_1$ to $P_2$, while $P_2$ sends $\vl{y}_2$ to $P_1$.
		\item Parties proceed as follows:
		\begin{enumerate}[itemsep=0mm]
			\item $P_0, P_3$: $\vl{r} = \gm{\vl{a}\vl{b}}{3} - {\vl{u}^1} - \vl{u}^2$; $\vl{q} = \vl{r}^{\vl{t}}$  if $\isTr = 1$, else $\vl{q} = \vl{r}$; Execute $\piJSh(P_0, P_3, \vl{q})$.
			\item $P_1, P_2$: $\vl{z} - \vl{r} = (\vl{y}_1 + \vl{y}_2 + \vl{y}_3) + \mk{\vl{a}} \mk{\vl{b}}$; $\vl{p} = (\vl{z} - \vl{r})^{\vl{t}}$ if $\isTr = 1$, else $\vl{p} = \vl{z} - \vl{r}$; Execute $\piJSh(P_1, P_2, \vl{p})$.
		\end{enumerate}
		\item Locally compute $\shr{\vl{o}} = \shr{\vl{p}} + \shr{\vl{q}}$. Here $\vl{o} = \vl{z}^{\vl{t}}$ if $\isTr = 1$ and $\vl{z}$ otherwise.
	\end{enumerate}
	\justify
	\vspace{-2mm}
	\algoHead{Verification:}
	\begin{enumerate}[itemsep=0mm]
		\item  $P_0, P_1, P_2$ sample random ${\vl{s}_1}, \vl{s}_2 \in_R \Z{\ell}$ and set $\vl{s} = \vl{s}_1 + \vl{s}_2$. $P_0$ sends $\vl{w} = \gm{\vl{a}\vl{b}}{1} + \gm{\vl{a}\vl{b}}{2} + {\vl{s}}$ to $P_3$.
		\item $P_3$ computes $\vl{v} = - (\pad{\vl{a}}{1} + \pad{\vl{a}}{2}) \mk{\vl{b}} - (\pad{\vl{b}}{1} + \pad{\vl{b}}{2}) \mk{\vl{a}} + {\vl{u}^1} + \vl{u}^2 + \vl{w}$ and sends $\Hash(\vl{v})$ to $P_1$ and $P_2$. Parties $P_1, P_2$ $\abort$ iff $\Hash(\vl{v}) \ne \Hash(\vl{y}_1 + \vl{y}_2 + {\vl{s}})$.
	\end{enumerate}     
\end{protocolbox}

\section{Building Blocks}
\label{app:tools}

\paragraph{Dot Product~(Scalar Product)}
\label{pa:4pcDotp}

\begin{lemma}[Communication]
	\label{appl:pidotpf}
	Protocol $\prot{\dotp}$~(\boxref{fig:piDotP})~(in $\this$) requires $2 \ell$ bits of communication in preprocessing, and $1$ round and $3 \ell$ bits of communication in the online phase.
\end{lemma}
\begin{proof}
	Here, the parties add up the locally computed shares corresponding to each partial product of the form $\vl{a}_i \vl{b}_i$ and then performs the communication of the sum. The communication pattern is similar to that of the fair multiplication protocol~(\boxref{fig:piMultiplication}) and the costs follow from Lemma~\ref{appl:piMultF}.
\end{proof}

\paragraph{Multi-input Multiplication}
\label{pa:4pcMultiinputfair}

\begin{lemma}[Communication]
	\label{appl:pimultTf}
	Protocol $\piMultT$~(\boxref{fig:piMultT})~(in $\this$) requires $9 \ell$ bits of communication in preprocessing, and $1$ round and $3 \ell$ bits of communication in the online phase.
\end{lemma}
\begin{proof}
	In the preprocessing, computation of $\gm{\vl{ab}}{}$ involves three instances of $\jsend$. Each of the computation of $\gm{\vl{ac}}{}, \gm{\vl{bc}}{}$ involves one instance of $\jsend$ and a communication from $P_0$ to $P_3$. The computation of $\gm{\vl{abc}}{}$ is similar to the preprocessing of fair multiplication protocol~(\boxref{fig:piMultiplication}). The communication pattern of the online phase is similar to that of the fair multiplication protocol. The costs follow from Lemma~\ref{appl:piMultF} and Lemma~\ref{appl:pijsend}.
\end{proof}

For the robust 3-input multiplication, correctness of three messages, ${\vl{w}}_{\vl{ac}}, {\vl{w}}_{\vl{bc}}, {\vl{w}}_{\vl{abc}}$, sent by $P_0$ have to be verified by invoking $\piVrfyP$.

\begin{protocolbox}{$\piMultT(\vl{a}, \vl{b}, \vl{c}, \isTr)$}{3-input fair multiplication in $\this$.}{fig:piMultT}
	Let $\isTr$ be a bit that denotes whether truncation is required ($\isTr =1$) or not~($\isTr=0$). \\
	\detail{
		{\bf Input(s):} $\shr{\vl{a}}, \shr{\vl{b}}, \shr{\vl{c}}$.\\
		{\bf Output:} $\shr{\vl{o}}$ where $\vl{o} = \vl{z}^{\vl{t}}$ if $\isTr = 1$ and $\vl{o} = \vl{z}$ if $\isTr = 0$ and $\vl{z} = \vl{abc}$.
	}
	\justify 
	\vspace{-2mm}
	\algoHead{Preprocessing:} 
	\begin{enumerate}[itemsep=0mm]
		\item Computation for $\gm{\vl{ab}}{}$: Invoke $\piMultsgr(\pd{\vl{a}}, \pd{\vl{b}})$~(\boxref{fig:piMultsgr}).
		\item Computation for $\gm{\vl{ac}}{}$: 
		\begin{myitemize}
			\item[--] Locally compute the following:
			\begin{align*}
				P_0, P_1: \gm{\vl{ac}}{1}  &= \pad{\vl{a}}{1} \pad{\vl{c}}{3} + \pad{\vl{a}}{3} \pad{\vl{c}}{1} + \pad{\vl{a}}{3} \pad{\vl{c}}{3}\\
				P_0, P_2: \gm{\vl{ac}}{2} &= \pad{\vl{a}}{2} \pad{\vl{c}}{3} + \pad{\vl{a}}{3} \pad{\vl{c}}{2} + \pad{\vl{a}}{2} \pad{\vl{c}}{2}\\
				P_0, P_3: \gm{\vl{ac}}{3} &= \pad{\vl{a}}{1} \pad{\vl{c}}{2} + \pad{\vl{a}}{2} \pad{\vl{c}}{1} + \pad{\vl{a}}{1} \pad{\vl{c}}{1}
			\end{align*}
			\item[--] $P_0, P_3$ and $P_1$ sample random ${\vl{u}}_{\vl{ac}}^1 \in_R \Z{\ell}$. $P_0, P_3$ compute and $\jsend$ ${\vl{u}}_{\vl{ac}}^2 = \gm{\vl{ac}}{3} - {\vl{u}}_{\vl{ac}}^1$ to $P_2$.  
			\item[--]  $P_0, P_1, P_2$ sample random ${\vl{s}}_{\vl{ac_1}}, {\vl{s}}_{\vl{ac_2}} \in_R \Z{\ell}$ and set ${\vl{s}}_{\vl{ac}} = {\vl{s}}_{\vl{ac_1}} + {\vl{s}}_{\vl{ac_2}}$. $P_0$ sends ${\vl{w}}_{\vl{ac}} = \gm{\vl{ac}}{1} + \gm{\vl{ac}}{2} + {\vl{s}}_{\vl{ac}}$ to $P_3$.
		\end{myitemize}
		\smallskip
		\item Computation for $\gm{\vl{bc}}{}$: Similar to Step 2 (for $\gm{\vl{ac}}{}$). $P_1, P_2$ obtain ${\vl{u}}_{\vl{bc}}^1, {\vl{u}}_{\vl{bc}}^2$ respectively such that ${\vl{u}}_{\vl{bc}}^1 + {\vl{u}}_{\vl{bc}}^2 = \gm{\vl{bc}}{3}$ . $P_3$ obtains ${\vl{w}}_{\vl{bc}}  = \gm{\vl{bc}}{1} + \gm{\vl{bc}}{2} + {\vl{s}}_{\vl{bc}}$. 
		\item Computation for $\gm{\vl{abc}}{}$: 
		\begin{myitemize}
			\item[--] Using $\gm{\vl{ab}}{}$~(Step 1), $\pad{\vl{c}}{}$, compute the following:
			\begin{align*}
				P_0, P_1: \gm{\vl{abc}}{1}  &= \gm{\vl{ab}}{1} \pad{\vl{c}}{3} + \gm{\vl{ab}}{3} \pad{\vl{c}}{1} + \gm{\vl{ab}}{3} \pad{\vl{c}}{3}\\
				P_0, P_2: \gm{\vl{abc}}{2} &= \gm{\vl{ab}}{2} \pad{\vl{c}}{3} + \gm{\vl{ab}}{3} \pad{\vl{c}}{2} + \gm{\vl{ab}}{2} \pad{\vl{c}}{2}\\
				P_0, P_3: \gm{\vl{abc}}{3} &= \gm{\vl{ab}}{1} \pad{\vl{c}}{2} + \gm{\vl{ab}}{2} \pad{\vl{c}}{1} + \gm{\vl{ab}}{1} \pad{\vl{c}}{1}
			\end{align*}
			\item[--] $P_0, P_3$ and $P_j$ sample random ${\vl{u}}_{\vl{abc}}^j \in_R \Z{\ell}$ for $j \in \{1,2\}$. Let ${\vl{u}}_{\vl{abc}}^1 + {\vl{u}}_{\vl{abc}}^2 = \gm{\vl{abc}}{3} + \vl{r}$ for $\vl{r} \in_R \Z{\ell}$.   
			\item[--]  $P_0, P_1, P_2$ sample random ${\vl{s}}_1, {\vl{s}}_2 \in_R \Z{\ell}$ and set ${\vl{s}} = {\vl{s}}_1 + {\vl{s}} _2$\footnote{For the fair protocol, it is enough for $P_0, P_1, P_2$ to sample ${\vl{s}}$ directly.}. $P_0$ sends ${\vl{w}}_{\vl{abc}} = \gm{\vl{abc}}{1} + \gm{\vl{abc}}{2} + {\vl{s}}$ to $P_3$.
		\end{myitemize}
		\smallskip
		\item $P_0, P_3$ compute $\vl{r} = {\vl{u}}_{\vl{abc}}^1 + {\vl{u}}_{\vl{abc}}^2 - \gm{\vl{abc}}{3}$ and set $\vl{q} = \vl{r}^{\vl{t}}$  if $\isTr = 1$, else set $\vl{q} = \vl{r}$. Execute $\piJSh(P_0, P_3, \vl{q})$ to generate $\shr{\vl{q}}$.
	\end{enumerate}
	\justify
	\vspace{-2mm}
	\algoHead{Online:} Let $\vl{y} = (\vl{z} - \vl{r}) - \mk{\vl{abc}}$.
	\begin{enumerate}[itemsep=0mm]
		\item Locally compute the following:
		\begin{align*}
			P_1: \vl{y}_1 &= - \pad{\vl{a}}{1} \mk{\vl{bc}} - \pad{\vl{b}}{1} \mk{\vl{ac}} - \pad{\vl{c}}{1} \mk{\vl{ab}} 
			+ \gm{\vl{ab}}{1} \mk{\vl{c}} \\&+ (\gm{\vl{ac}}{1} + {\vl{u}}_{\vl{ac}}^1) \mk{\vl{b}} + (\gm{\vl{bc}}{1} + {\vl{u}}_{\vl{bc}}^1) \mk{\vl{a}} - (\gm{\vl{a}\vl{bc}}{1} + {\vl{u}}_{\vl{a}\vl{bc}}^1)\\
			P_2: \vl{y}_2 &= - \pad{\vl{a}}{2} \mk{\vl{bc}} - \pad{\vl{b}}{2} \mk{\vl{ac}} - \pad{\vl{c}}{2} \mk{\vl{ab}} 
			+ \gm{\vl{ab}}{2} \mk{\vl{c}} \\&+ (\gm{\vl{ac}}{2} + {\vl{u}}_{\vl{ac}}^2) \mk{\vl{b}} + (\gm{\vl{bc}}{2} + {\vl{u}}_{\vl{bc}}^2) \mk{\vl{a}} - (\gm{\vl{a}\vl{bc}}{2} + {\vl{u}}_{\vl{a}\vl{bc}}^2)\\
			P_1, P_2: \vl{y}_3 &= - \pad{\vl{a}}{3} \mk{\vl{bc}} - \pad{\vl{b}}{3} \mk{\vl{ac}} - \pad{\vl{c}}{3} \mk{\vl{ab}} + \gm{\vl{ab}}{3} \mk{\vl{c}}
		\end{align*}
		\item $P_1$ sends $\vl{y}_2$ to $P_2$, while $P_2$ sends $\vl{y}_1$ to $P_1$, and they locally compute $\vl{z} - \vl{r} = (\vl{y}_1 + \vl{y}_2 + \vl{y}_3) + \mk{\vl{abc}}$.
		\item If $\isTr = 1$, $P_1, P_2$ locally set $\vl{p} = (\vl{z} - \vl{r})^{\vl{t}}$, else $\vl{p} = \vl{z} - \vl{r}$. \newline Execute $\piJSh(P_1, P_2, \vl{p})$ to generate $\shr{\vl{p}}$. 
		\item Locally compute $\shr{\vl{o}} = \shr{\vl{p}} + \shr{\vl{q}}$. Here $\vl{o} = \vl{z}^{\vl{t}}$ if $\isTr = 1$ and $\vl{z}$ otherwise.
		\item {\em Verification:} 
		\begin{myitemize}
			\item[--] Locally compute the following:
			\begin{align*}
				P_3: \vl{v} &= - (\pad{\vl{a}}{1} + \pad{\vl{a}}{2})\mk{\vl{bc}} - (\pad{\vl{b}}{1} + \pad{\vl{b}}{2} )\mk{\vl{ac}} - (\pad{\vl{c}}{1} + \pad{\vl{c}}{2}) \mk{\vl{ab}} \\&+(\gm{\vl{ab}}{1} + \gm{\vl{ab}}{2})\mk{\vl{c}} + ({\vl{w}}_{\vl{ac}} + \gm{\vl{ac}}{3}) \mk{\vl{b}} + ({\vl{w}}_{\vl{bc}} + \gm{\vl{bc}}{3}) \mk{\vl{a}} \\&- ({\vl{u}}_{\vl{abc}}^1 + {\vl{u}}_{\vl{abc}}^2 + {\vl{w}}_{\vl{abc}})\\
				P_1, P_2: \vl{v}' &=  \vl{y}_1 + \vl{y}_2 + {\vl{s}}_{\vl{ac}} \mk{\vl{b}} + {\vl{s}}_{\vl{bc}} \mk{\vl{a}} - {\vl{s}}
			\end{align*}
			\item[--] $P_3$ sends $\Hash(\vl{v})$ to $P_1, P_2$, who $\abort$ iff $\Hash(\vl{v}) \ne \Hash(\vl{v}')$.
		\end{myitemize}
	\end{enumerate}     
\end{protocolbox}

{\em 4-input multiplication: } To obtain $\shr{\cdot}$-sharing of $\vl{z}= \vl{a} \vl{b} \vl{c} \vl{d}$ given the $\shr{\cdot}$-sharing of $\vl{a}, \vl{b}, \vl{c}, \vl{d}$, we can write $\vl{z} + \vl{r}$ as 
\begin{align*}
	\vl{z} - \vl{r} &=(\mk{\vl{a}} - \pd{\vl{a}})(\mk{\vl{b}} - \pd{\vl{b}})(\mk{\vl{c}} - \pd{\vl{c}})(\mk{\vl{d}} - \pd{\vl{d}}) - \vl{r} \\
	&= \mk{\vl{a}\vl{b}\vl{c}\vl{d}} - \mk{\vl{b}\vl{c}\vl{d}} \pd{\vl{a}} - \mk{\vl{a}\vl{c}\vl{d}}  \pd{\vl{b}} - \mk{\vl{a}\vl{b}\vl{d}} \pd{\vl{c}} - \mk{\vl{a}\vl{b}\vl{c}} \pd{\vl{d}} \\
	&~~~+ \mk{\vl{a}\vl{b}} \gm{\vl{c} \vl{d}}{} 
	+ \mk{\vl{a}\vl{c}} \gm{\vl{b} \vl{d}}{} + \mk{\vl{a}\vl{d}} \gm{\vl{b} \vl{c}}{} + \mk{\vl{b}\vl{c}} \gm{\vl{a} \vl{d}}{} + \mk{\vl{b}\vl{d}} \gm{\vl{a} \vl{c}}{} \\
	&~~~+ \mk{\vl{c}\vl{d}} \gm{\vl{a} \vl{b}}{} - \mk{\vl{a}}\gm{\vl{b}\vl{c}\vl{d}}{} - \mk{\vl{b}}\gm{\vl{a}\vl{c}\vl{d}}{} - \mk{\vl{c}}\gm{\vl{a}\vl{b}\vl{d}}{} - \mk{\vl{d}}\gm{\vl{a}\vl{b}\vl{c}}{} \\
	&~~~+\gm{\vl{a} \vl{b} \vl{c} \vl{d}}{} - \vl{r}
\end{align*}
While the online phase proceeds similarly to the 2 and 3-input multiplication, in the preprocessing phase, the parties need to generate the additive shares of $\gm{\vl{a} \vl{b}}{}, \gm{\vl{a} \vl{c}}{}, \gm{\vl{a} \vl{d}}{}, \gm{\vl{b} \vl{c}}{}, \allowbreak \gm{\vl{b} \vl{d}}{}, \gm{\vl{c} \vl{d}}{}, \gm{\vl{a} \vl{b} \vl{c}}{}, \gm{\vl{a} \vl{b} \vl{d}}{}, \gm{\vl{a} \vl{c} \vl{d}}{}, \gm{\vl{b} \vl{c} \vl{d}}{}, \allowbreak \gm{\vl{a} \vl{b} \vl{c} \vl{d}}{}$. This is computed similarly as in the case of 3-input multiplication as follows. 
Parties generate shares of $\gm{\vl{a} \vl{c}}{}, \gm{\vl{a} \vl{d}}{}, \gm{\vl{b} \vl{c}}{}, \gm{\vl{b} \vl{d}}{}$ similar to the generation of shares of $\gm{\vl{a} \vl{c}}{}$ in the 3-input multiplication. For $\gm{\vl{a} \vl{b}}{}, \gm{\vl{c} \vl{d}}{}$, parties proceed similar to generation of shares of $\gm{\vl{a} \vl{b}}{}$ in the 3-input multiplication, where the respective $\sgr{\cdot}$-shares are generated. This is followed by generation of shares of $\gm{\vl{a} \vl{b} \vl{c}}{}, \gm{\vl{a} \vl{b} \vl{d}}{}, \gm{\vl{a} \vl{c} \vl{d}}{}, \gm{\vl{b} \vl{c} \vl{d}}{}, \gm{\vl{a} \vl{b} \vl{c} \vl{d}}{}$ following steps similar to the ones involved in generating $\gm{\vl{a} \vl{b} \vl{c}}{}$ in the 3-input multiplication. Since the protocol is very similar to the 3-input protocol, we omit the formal details. 

\paragraph{Bit to Arithmetic}
\label{pa:4pcBit2A}
For verifying the $\sgr{\cdot}$-sharing of $\vl{u}$ by $P_0$, we let $P_3$ obtain the bit $(\pad{\bitb}{} \xor \vl{r}_{b})$ as well as its arithmetic equivalent $\arval{(\pad{\bitb}{} \xor \vl{r}_{b})}$ in clear. Here $\vl{r}_{b}$ denotes a random bit known to $P_0, P_1, P_2$. $P_3$ checks if both the received values are equivalent and raise a complaint if they are inconsistent. To catch a corrupt $P_0$ from sharing a wrong $\vl{u}$ value, parties use the $\sgr{\cdot}$-shares of $\vl{u}$ to compute $\arval{(\pad{\bitb}{} \xor \vl{r}_{b})}$. Moreover, the verification steps are designed in such a way that every value communicated can be locally computed by at least two parties. This enables to use $\jsend$ for communication and hence the desired security guarantee is achieved.

\begin{protocolbox}{$\prot{\bitA}(\shrB{\bitb})$}{Bit to Arithmetic conversion}{fig:piBitA}
	Let $\vl{u} = \arval{(\pad{\vl{\bitb}}{})}$ and $\vl{v} = \arval{\mk{\vl{\bitb}}}$.\\
	\detail{
		{\bf Input(s):} $\shrB{\bitb}$.\\
		{\bf Output:} $\shr{\bitb}$.
	}
	\justify 
	\vspace{-2mm}
	\algoHead{Preprocessing:} 
	\begin{enumerate}[itemsep=0mm]
		\item Generation of $\sgr{\vl{u}}$: $P_0, P_3, P_i$ for $i \in \{1,2\}$ sample $\vl{u}^i$. $P_0$ sends $\vl{u}^3 = \vl{u} - \vl{u}^1 - \vl{u}^2$ to $P_1, P_2$.
		\item $P_0, P_1, P_2$ sample random $\vl{r}_{\bitb} \in \bitset$ and $\vl{r} \in \Z{\ell}$.
		\item $P_1, P_2$ $\jsend$ $\pad{\bitb}{3} \xor \vl{r}_{\bitb}$ to $P_3$. $P_3$ locally sets $\pad{\bitb}{} \xor \vl{r}_{\bitb} = (\pad{\bitb}{1} \xor \pad{\bitb}{2}) \xor (\pad{\bitb}{3} \xor \vl{r}_{\bitb})$.
		\item Parties compute: $P_1, P_0: \vl{w}_1 = \arval{\vl{r}_{\bitb}} + (\vl{u}^1 + \vl{u}^3) (1 - 2 \arval{\vl{r}_{\bitb}}) + \vl{r},~~P_2, P_0: \vl{w}_2 = (\vl{u}^2) (1 - 2 \arval{\vl{r}_{\bitb}}) - \vl{r}$. 
		\item $P_1, P_0$ $\jsend$ $\vl{w}_1$ to $P_3$, while $P_2, P_0$ $\jsend$ $\Hash(\vl{w}_2)$ to $P_3$.
		\item $P_3$ sets $\flag = \continue$ if $\Hash(\arval{(\pad{\bitb}{} \xor \vl{r}_b)} - \vl{w}_1) = \Hash(\vl{w}_2)$, else $\flag = \abort$. $P_3$ sends $\flag$ to $P_0, P_1, P_2$. Parties mutually exchange the flag and accept the value that forms the majority.
		\item For robust setting, if $\flag = \abort$, then $\TTP = P_1$ (or $P_2$).
	\end{enumerate}
	\justify
	\vspace{-2mm}
	\algoHead{Online:} Let $\vl{y} = \arval{\bitb}$.
	\begin{enumerate}[itemsep=0mm]
		\item Parties locally compute the following:
		\begin{align*}
			P_1, P_3: \vl{y}_1  &=  \vl{v} + \vl{u}^1 (1 - 2\vl{v}) \\
			P_2, P_3: \vl{y}_2 &=  \vl{u}^2 (1 - 2\vl{v}) \\
			P_1, P_2: \vl{y}_3  &=  \vl{u}^3 (1 - 2\vl{v})
		\end{align*}
		\item $(P_1, P_3), (P_2, P_3), (P_1, P_2)$ execute $\prot{\JSh}$ on $\vl{y}_1, \vl{y}_2, \vl{y}_3$ to generate the respective $\shr{\cdot}$-shares.
		\item Compute $\shr{\vl{y}} = \shr{\vl{y}_1} + \shr{\vl{y}_2} + \shr{\vl{y}_3}$.
	\end{enumerate}     
\end{protocolbox}

\begin{lemma}[Communication]
	\label{appl:pibitA}
	Protocol $\prot{\bitA}$~(\boxref{fig:piBitA}) requires $3 \ell + 1$ bits of communication in preprocessing, and $1$ round and $3 \ell$ bits of communication in the online phase.
\end{lemma}
\begin{proof}
	During preprocessing, generation of $\sgr{\vl{u}}$ involves communication of $\ell$ bits from $P_0$ to each of $P_1, P_2$. As part of verification, two instances of $\jsend$ are executed, one on $1$ bit and other on $\ell$ bits. The communication for hash gets amortized over multiple instances. The online phase involves three instances of joint sharing protocol resulting in $1$ rounds and a communication of $3\ell$ bits. The costs follow from Lemma~\ref{appl:pijsend}.
\end{proof}

\paragraph{Bit Injection}
\label{pa:4pcbitInj}

\begin{lemma}[Communication]
	\label{appl:pibitinj}
	Protocol $\prot{\bitinj}$ requires $6\ell + 1$ bits of communication in preprocessing, and $1$ round and $3 \ell$ bits of communication in the online phase.
\end{lemma}
\begin{proof}
	During preprocessing, generation of $\sgr{\vl{u}_i}$ for $i \in [m]$ and its verification is similar to $\prot{\bitA}$. The cost of generating $\sgr{\mu_i}$ follows from $\piMultsgr$. The communication in the online phase is similar to that of the $\prot{\bitA}$ protocol. The cost follows from Lemma~\ref{appl:pibitA}. 
\end{proof}

\paragraph{Piecewise Polynomials}
\label{pa:4pcPiecewise}

\begin{lemma}[Communication]
	\label{appl:piecewise}
	Protocol $\prot{\piecewise}$~(\boxref{fig:piecewise}) requires $m(6\ell + 1)$ bits of communication in preprocessing, and $1$ round and $3 \ell$ bits of communication in the online phase.
\end{lemma}
\begin{proof}
	During preprocessing, generation of $\sgr{\vl{u}_i}, \sgr{\mu_i}$ for $i \in [m]$ and its verification is similar to $\prot{\bitinj}$. The communication in the online phase is similar to that of the $\prot{\bitinj}$ protocol except that parties locally add the values before executing $\piJSh$. The cost follows from Lemma~\ref{appl:pibitinj}. 
\end{proof}

\medskip
\begin{protocolbox}{$\prot{\piecewise} \left( \{ \shrB{\bitb_i}, \shr{\vl{v}_i} \}_{i=1}^{m} \right)$}{Piecewise polynomial evaluation protocol}{fig:piecewise}
	\smallskip
	Let $\vl{u}_i = \arval{\pd{\bitb_i}}$ and $\mu_i = \arval{\pd{\bitb_i}} \pd{\vl{v}_i} $.\\
	\detail{
		{\bf Input(s):} $ \{ \shrB{\bitb_i}, \shr{\vl{v}_i} \}_{i=1}^{m} $.\\
		{\bf Output:} $\shr{\vl{z}} = \shr{\sum_{i=1}^{m} \bitb_i \cdot \vl{v}_i}$.
	}
	\justify 
	\vspace{-2mm}
	\algoHead{Preprocessing:} For $i \in [m]$, perform the following:
	\begin{enumerate}[itemsep=0mm]
		\item Parties proceed similar to $\prot{\bitA}$ to generate $\sgr{\vl{u}_i}$~(\boxref{fig:piBitA}).
		\item Generation of $\sgr{\mu_i}$: Invoke $\piMultsgr(\vl{u}_i, \pd{\vl{v}_i})$.  
	\end{enumerate}
	\justify
	\vspace{-2mm}
	\algoHead{Online:} 
	\begin{enumerate}[itemsep=0mm]
		\item Parties locally compute the following:
		\begin{align*}
			P_1, P_3 : \vl{z}_i^1 &= \arval{\mk{\bitb_i}} \mk{\vl{v}_i} - \arval{\mk{\bitb_i}} \pad{\vl{v}_i}{1}  + (2 \arval{\mk{\bitb_i}} - 1) (\mu_i^{1} -  \mk{\vl{v}_i} \vl{u}_i^{1})\\
			P_2, P_3 : \vl{z}_i^2 &= ~~~~~~~~~~~~~- \arval{\mk{\bitb_i}} \pad{\vl{v}_i}{2}  + (2 \arval{\mk{\bitb_i}} - 1) (\mu_i^{2} -  \mk{\vl{v}_i} \vl{u}_i^{2})\\
			P_1, P_2 : \vl{z}_i^3 &= ~~~~~~~~~~~~~- \arval{\mk{\bitb_i}} \pad{\vl{v}_i}{3}  + (2 \arval{\mk{\bitb_i}} - 1) (\mu_i^{3} -  \mk{\vl{v}_i} \vl{u}_i^{3})
		\end{align*}
		\item Set $\vl{z}^1 = \sum_{i=1}^{m} \vl{z}_i^1$,~~$\vl{z}^2 = \sum_{i=1}^{m} \vl{z}_i^2$,~~$\vl{z}^3 = \sum_{i=1}^{m} \vl{z}_i^3$
		\item $(P_1, P_3), (P_2, P_3), (P_1, P_2)$ execute $\prot{\JSh}$ on $\vl{z}^1, \vl{z}^2, \vl{z}^3$ to generate the respective $\shr{\cdot}$-shares.
		\item Compute $\shr{\vl{z}} = \shr{\vl{z}^1} + \shr{\vl{z}^2} + \shr{\vl{z}^3}$.
	\end{enumerate}     
\end{protocolbox}

\paragraph{Non-Linear Activation functions}
\label{pa:4pcAct}
We discuss two widely used activation functions, (i) Rectified Linear Unit ($\relu$) and (ii) Sigmoid (Sig). These functions can be viewed as piece-wise polynomial functions and can thus be evaluated using the protocol mentioned above~($\prot{\piecewise}$,~\boxref{fig:piecewise}). 

{\em (i) ReLU:}
The $\relu$ function, $\relu(\vl{v}) = \max(0, \vl{v})$, can be written as a piece-wise polynomial function as follows.

\begin{small}
	\begin{align*}
		\relu(\vl{v}) = 
		\begin{cases}
			0, & \vl{v} <0 \\
			\vl{v} & 0 \leq \vl{v}
		\end{cases}
	\end{align*}
\end{small}

{\em (ii) Sig:}
We use the MPC-friendly variant of the Sigmoid function~\cite{SP:MohZha17, CCS:MohRin18, CCSW:CCPS19} which is given below:

\begin{small}
	\begin{align*}
		\sig(\vl{v}) = \left\{
		\begin{array}{lll}
			0                  & \quad \vl{v} < -\frac{1}{2} \\
			\vl{v} + \frac{1}{2} & \quad - \frac{1}{2} \leq \vl{v} \leq \frac{1}{2} \\
			1                  & \quad \frac{1}{2} < \vl{v}
		\end{array}
		\right.
	\end{align*}
\end{small}

\paragraph{ArgMin/ ArgMax}
\label{pa:4pcArgMinMaz}
The formal protocol appears in \boxref{fig:piargmin}. Here, $\pibitextf(\shr{\vl{x}_1}, \shr{\vl{x}_2})$ computes the boolean sharing corresponding to the $\msb$ of $\vl{x}_1 - \vl{x}_2$.

\begin{protocolbox}{$\piargmin (\shr{\vct{x}})$}{Protocol to find index of smallest element in $\vct{x}$}{fig:piargmin}
	\justify
	Let $\vct{b}$ be the bit vector of size $m$, where $m$ equals the size of $\vct{x}$. Parties execute the following steps in the respective preprocessing and online phases. 
	\begin{enumerate}[itemsep=0mm]
		\item If $m = 2$, do the following.
		\begin{enumerate}
			\item[--] $\shrB{\vl{d}_1} = \pibitextf(\shr{\vl{x}_1}, \shr{\vl{x}_2})$ and $\shrB{\vl{d}_2} = 1 \xor \shrB{\vl{d}_1}$.
			\item[--] $\shr{\vl{y}} = \piobv(\shr{\vl{x}_2}, \shr{\vl{x}_1}, \shrB{\vl{d}_1})$.
			\item[--] Return $(\shrB{\vl{d}_1}, \shrB{\vl{d}_2}, \shr{\vl{y}})$.
		\end{enumerate} 
		\item Else, if $m = 3$, do the following
		\begin{enumerate}
			\item[--] $\shrB{\vl{d}_1^{\prime}} = \pibitextf(\shr{\vl{x}_1}, \shr{\vl{x}_2})$.
			\item[--] $\shr{\vl{y}^{\prime}} = \piobv(\shr{\vl{x}_2}, \shr{\vl{x}_1}, \shrB{\vl{d}_1^{\prime}})$. 
			\item[--] $\shrB{\vl{d}_2^{\prime}} = \pibitextf(\shr{\vl{y}^{\prime}}, \shr{\vl{x}_3})$.
			\item[--] $\shr{\vl{y}} = \piobv(\shr{\vl{x}_3}, \shr{\vl{y}^{\prime}},\shrB{\vl{d}_2^{\prime}})$.
			\item[--] $\shrB{\vl{d}_1} = \piMult(\shrB{\vl{d}_1^{\prime}}, \shrB{\vl{d}_2^{\prime}})$, $\shrB{\vl{d}_2} = \shrB{\vl{d}_2^{\prime}} \xor \shrB{\vl{d}_1}$.
			\item[--] $\shrB{\vl{d}_3} = 1 \xor \shrB{\vl{d}_1^{\prime}} \xor \shrB{\vl{d}_2^{\prime}}$. 
			\item[--] Return $(\shrB{\vl{d}_1}, \shrB{\vl{d}_2}, \shrB{\vl{d}_3}, \shr{\vl{y}})$.
		\end{enumerate}
	    \item Else, let $\vct{x_1} = (\vl{x}_1, \ldots, \vl{x}_{\lfloor m/2 \rfloor})$ and $\vct{x_2} = (\vl{x}_{\lfloor m/2 \rfloor + 1}, \ldots, \vl{x}_m)$.
	    \begin{enumerate}
	    	\item[--]  $\big(\shrB{\vl{d}_1}, \ldots, \shrB{\vl{d}_{\lfloor m/2 \rfloor}}, \shr{\vl{y}_1} \big) = \piargmin(\shr{\vct{x_1}})$.
	    	\item[--]  $\big(\shrB{\vl{d}_{\lfloor m/2 \rfloor + 1}}, \ldots, \shrB{\vl{d}_m}, \shr{\vl{y}_2} \big) = \piargmin(\shr{\vct{x_2}})$.
	    	\item[--]  $\shrB{\vl{d}} = \pibitextf(\shr{\vl{y}_1}, \shr{\vl{y}_2})$.
	    	\item[--]  $\shr{\vl{y}} = \piobv(\shr{\vl{y}_2}, \shr{\vl{y}_1}, \shrB{\vl{d}})$.
	    	\item[--] $\shrB{\bitb_j} = \piMult(\shrB{\vl{d}}, \shrB{\vl{d}_j})$ ;~$j \in \{1, \ldots, \lfloor m/2 \rfloor \}$.
	    	\item[--] $\shrB{\bitb_j} = \piMult(1 \xor \shrB{\vl{d}}, \shrB{\vl{d}_j})$ ;~$j \in \{ \lfloor m/2 \rfloor + 1, \ldots, m \}$.
	    	\item[--] Return $\big( \shrB{\bitb_1},  \ldots, \shrB{\bitb_m}, \shr{\vl{y}} \big)$.
	    \end{enumerate}
	\end{enumerate}		
\end{protocolbox}

To begin with, parties initialize $\bitb_j = 1$ for $\bitb_j \in \vct{b}$ by locally setting $\mk{{\bitb}_j} = 1$ and $\pad{{\bitb}_j}{1} = \pad{{\bitb}_j}{2} = \pad{{\bitb}_j}{3} = 0$. The minimum, $\vl{y}_{ij}$, of two elements, $\vl{x}_i,\vl{x}_j$ can be computed as: one invocation of bit extraction protocol to obtain $\shrB{\cdot}$-sharing of $\bitb_{ij}$, where $\bitb_{ij} = 1$ if $\vl{x}_i < \vl{x}_j$, and $\bitb_{ij} = 0$ otherwise; one invocation of oblivious selection protocol $\piobv(\vl{x}_j, \vl{x}_i, \bitb_{ij})$, which outputs $\shr{\cdot}$-shares of $\vl{y}_{ij} = \vl{x}_j$ if $\bitb_{ij} = 0$, and $\vl{y}_{ij} = \vl{x}_i$, otherwise. 
To update $\vct{b}$ to reflect the pairwise minimums, we view the elements $\vl{x}_j \in \vct{x}$ as the leaves of a binary tree, in a bottom-up manner. For two elements in a pair, say $(\vl{x}_i, \vl{x}_j)$, whose pairwise minimum is $\vl{y}_{ij}$, we let $\vl{y}_{ij}$ be the root node with $\vl{x}_i$ as its left child and $\vl{x}_j$ as its right child. Now, to update $\vct{b}$, parties multiply $\bitb_{ij}$ with the bits in $\vct{b}$ associated with the {\em left-reachable leaf nodes}, which comprise of all the leaf nodes (elements of $\vct{x}$) that are reachable through the left child of the root. 
Similarly, parties multiply $1 \xor \bitb_{ij}$ with the bits in $\vct{b}$ associated with the {\em right-reachable leaf nodes}, which comprise of all the leaf nodes (elements of $\vct{x}$) that are reachable through the right child of the root. Thus, if $\bitb_{ij} = 1$ indicating that $\vl{x}_i < \vl{x}_j$, $\bitb_i$ remains $1$ as it gets multiplied by $\bitb_{ij} = 1$ while $\bitb_j$ gets reset to $0$ as it gets multiplied by $1 \xor \bitb_{ij} = 0$. The case for $\bitb_{ij} = 0$ holds for similar reasons.
Given the values $\vl{y}_{ij}$ for the next level, and the updated $\vct{b}$, the steps are applied recursively until the minimum element is obtained. 

The protocol $\piargmax$ which allows the parties to compute the index of the largest element in a $\shr{\cdot}$-shared vector $\vct{x} = (\vl{x}_1, \ldots, \vl{x}_m)$, is similar to $\piargmin$ with the following difference. To find the maximum among two elements $(\shr{\vl{x}_i}, \shr{\vl{x}_j})$, parties run the bit extraction protocol to obtain $\shrB{\bitb_{ij}}$ as before, followed by $\piobv(\vl{x}_i, \vl{x}_j, \bitb_{ij})$, which outputs $\shr{\cdot}$-shares of $\vl{y}_{ij} = \vl{x}_i$ if $\bitb_{ij} = 0$, and $\vl{y}_{ij} = \vl{x}_j$, otherwise. Now, $\vct{b}$ is updated in each level by multiplying $1 \xor \bitb_{ij}$ with the bits in $\vct{b}$ associated with the {\em left-reachable leaf nodes} (as described before) and multiplying $\bitb_{ij}$ with the bits in $\vct{b}$ associated with the {\em right-reachable leaf nodes}.

\section{Garbled World}
\label{app:garbled}

\subsection{Garbling scheme and properties}
\label{app:gcproperties}
As per Yao's garbling circuit paradigm~\cite{FOCS:Yao82b}, every wire in the circuit is assigned two $\csec$-bit strings, called ``keys'', one each for bit value $0$ and $1$ on that wire. Let $(\key{\vl{x}}{0},\key{\vl{x}}{1})$ denote the zero-key and one-key, respectively, on wire $\vl{x}$ in the circuit.
For simplicity, the same notation is used for wire identity as well as the value on the wire. For instance,  the key-pair for wire $\vl{x}$ is denoted as $(\key{\vl{x}}{0},\key{\vl{x}}{1})$, while the key corresponding to bit  $\vl{x}$ on the wire is denoted as $\key{\vl{x}}{\vl{x}}$.
Then, each gate is constructed by encrypting the output-wire key with the appropriate input-wire keys. For example, for an AND gate with input wires $\vl{x}, \vl{y}$ and output wire $\vl{z}$, $\key{\vl{z}}{0}$ is double encrypted with keys $\key{\vl{x}}{0}, \key{\vl{y}}{0}$, with $\key{\vl{x}}{0}, \key{\vl{y}}{1}$, and with $\key{\vl{x}}{1}, \key{\vl{y}}{0}$, while $\key{\vl{z}}{1}$ is double encrypted with $\key{\vl{x}}{1}, \key{\vl{y}}{1}$. Given one key on each input wire, the output wire key can be obtained by decrypting the ciphertext which was encrypted using the corresponding input wire keys. These ciphertexts are provided in a permuted order so that the evaluating party does not learn which key, $\key{\vl{z}}{0}$ or $\key{\vl{z}}{1}$, it obtains after decryption.

Formally, a garbling scheme $\GS$, consists of four algorithms $(\Gb, \En, \allowbreak \Ev, \De)$ defined as follows:
\begin{enumerate}
	\item $\Gb(\onesec,\Ckt) \rightarrow (\GC,e,d)$: $\Gb$ takes as input the security parameter $\kappa$ and the circuit $\Ckt$ to be garbled, and outputs a garbled circuit $\GC$, encoding information $e$ and decoding information $d$.
	\item $\En(x,e) \rightarrow \X$: $\En$ encodes input $x$ using $e$ to output encoded input $\X$. $\X$ is referred to as encoded input or encoded keys interchangeably.
	\item $\Ev(\GC,\X) \rightarrow \Y$:  $\Ev$ evaluates the garbled circuit $\GC$ on the encoded input $\X$ and produces the encoded output $\Y$.
	\item $\De(\Y,d) \rightarrow y$: The encoded output $\Y$ is decoded into the clear output $y$ by running the $\De$ algorithm on $\Y$ and $d$.
\end{enumerate}

We rely on the following properties of garbling scheme~\cite{CCS:BelHoaRog12} in our constructions.
\begin{enumerate}[itemsep=0mm]
	\item A garbling scheme $\GS = (\Gb, \En, \Ev, \De)$ is {\em correct} if for all input lengths $n \leq \poly(\csec)$, circuits $C: \{0,1\}^{n} \rightarrow \{0,1\}^{m}$ and inputs $x \in \{0,1\}^{n}$, the following holds.
		\begin{align*}
			\Prob[\De(\Ev(\GC, \En(x, e)), d) \neq C(x) : \\~~~(\GC, e, d) \leftarrow \Gb(\onesec, C)] < \negl(\csec)
		\end{align*}
	\item A garbling scheme $\GS$ is said to be \textit{private} if for all $n \leq \poly(\csec)$, circuit $C: \{0,1\}^{n} \rightarrow \{0,1\}^{m}$, there exists a $\ppt$ simulator $\Sim_{\priv}$ such that for all $x \in \{0,1\}^{n}$, for all $\ppt$ adversary $\Adv$ the following distributions are computationally indistinguishable.
	\begin{myitemize}
		\item[-] $\Real(C, x)$: run $(\GC, e, d) \leftarrow \Gb(\onesec, C)$ and output $(\GC, \En(x, e), d)$.
		\item[-] $\Ideal(C, C(x))$: run $(\GC^{\prime}, \textbf{X}, d^{\prime}) \leftarrow \Sim_{\priv}(\onesec, C, C(x))$ and output $(\GC^{\prime}, \textbf{X}, d^{\prime})$.
	\end{myitemize}
	\item A garbling scheme $\GS$ is \textit{authentic} if for all $n \leq \poly(\csec)$, circuit $C: \{0,1\}^{n} \rightarrow \{0,1\}^{m}$, input $x \in \{0,1\}^{n}$ and for all $\ppt$ adversary $\Adv$, the following probability is $\negl(\csec)$.
	
	\begin{footnotesize}
		\begin{align*}
			\Prob \Bigg(
			\begin{aligned}
				&\hat{\textbf{Y}} \neq \Ev(\GC, \textbf{X}) \\
				&\wedge \De(\hat{\textbf{Y}}, d) \neq \bot
			\end{aligned}
			:
			\!
			\begin{aligned}
				\textbf{X} = \En(x, e), &(\GC, e, d) \leftarrow \Gb(\csec, \Ckt), \\
				&\hat{\textbf{Y}} \leftarrow \Adv(\GC, \textbf{X})
			\end{aligned}
			\Bigg)
		\end{align*}
	\end{footnotesize}
	
\end{enumerate}

\subsection{2GC Variant}
\label{subsec:gcworld2}
We begin with the details of the evaluation and output phases. 

\paragraph{Evaluation} 
\label{subsec:gbeval} 
Let $f(\vl{x})$ be the function to be evaluated. At this point, the function input is $\shrC{\cdot}$-shared. This renders $\shrG{\cdot}$-sharing for the input of the GC that corresponds to the function $f'\big({\mk{\vl{x}}}, {\av{\vl{x}}}, {\pad{\vl{x}}{3}} \big)$ which first combines the given boolean-shares to compute the actual input and then applies $f$ on it. Let $\GC_j$ denote the garbled circuit to be sent to $P_j \in \{P_1, P_2\}$ by garblers in $\PlSet{j}$. Sending of $\GC_j$ is overlapped  with the key transfer (during generation of $\shrC{\vl{x}}$), to save rounds, where garblers in $\{P_0, P_3\}$ $\jsend$ $\GC_j$ to $P_j$. On receiving the $\GC$, evaluators evaluate their respective GCs and obtain the key corresponding to the output, say $\vl{z}$. This generates $\shrG{\vl{z}}$.

\paragraph{Output phase} 
\label{subsec:gbop} 
The goal of output computation is to compute the output $\vl{z}$ from $\shrG{\vl{z}}$.
To reconstruct $\vl{z}$ towards $P_j \in \{P_1, P_2\}$, two garblers in $\PlSet{j}$ send the least significant bit $\vl{p}^j$ of $\key{\vl{z}}{0, j}$, referred to as the decoding information, to $P_j$. If the received values are consistent, $P_j$ uses the received $\vl{p}^j$ to reconstruct $\vl{z}$ as $\vl{z} = \vl{p}^j \xor \vl{q}^j$, where $\vl{q}^j$ denotes the least significant bit of $\key{\vl{z}}{\vl{z}, j}$; else $P_j$ aborts. 
To reconstruct $\vl{z}$ towards the garblers $P_g \in \{P_0, P_3\}$, one evaluator, say $P_1$ sends the least significant bit, $\vl{q}^1$, of $\key{{\vl{z}}}{\vl{z}, 1}$ along with $\h = \Hash(\key{{\vl{z}}}{\vl{z}, 1})$ to $P_g$, where $\Hash$ is a collision-resistant hash function. If a garbler received a consistent $(\vl{q}^1, \h)$ pair from $P_1$ such that there exists a $K \in \{\key{{\vl{z}}}{{0, 1}}, \key{{\vl{z}}}{{1, 1}}\}$ whose least significant bit is $\vl{q}^1$ and $\Hash(K) = \h$, then it uses $\vl{q}^1$ for reconstructing $\vl{z}$; else the garbler aborts the computation.
Note that a corrupt evaluator $P_1$ cannot create confusion among garblers in $\{P_0, P_3\}$ by sending the key that was not output by the GC owing to the authenticity of the garbling scheme. Reconstruction is lightweight and requires a single round for garblers while reconstruction towards evaluators can be overlapped with key transfer and does not incur extra rounds.
The protocol appears in \boxref{fig:pirec}.

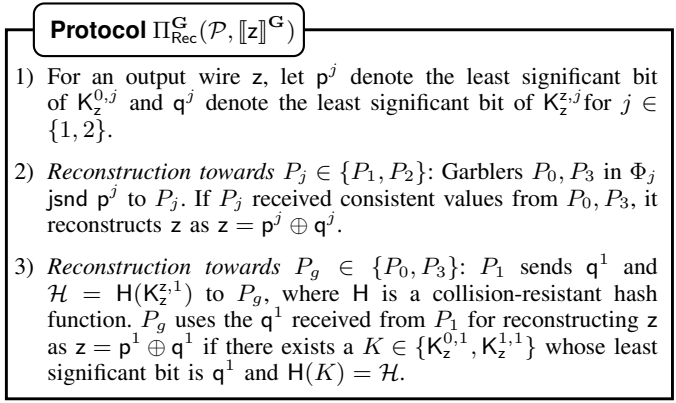
\begin{protocolbox}{$\pigrec(\Partyset, \shrG{\vl{z}})$}{Output computation: reconstruction of $\vl{z}$}{fig:pirec}
	\justify
	\begin{enumerate}
		\item For an output wire $\vl{z}$, let $\vl{p}^j$ denote the least significant bit of $\key{{\vl{z}}}{0,j}$ and $\vl{q}^j$ denote the least significant bit of $\key{{\vl{z}}}{\vl{z},j}$for $j \in \{1, 2\}$.
		\item {\em Reconstruction towards $P_j \in \{P_1, P_2\}$}: Garblers $P_0, P_3$ in $\PlSet{j}$ $\jsend$ $\vl{p}^j$ to $P_j$. If $P_j$ received consistent values from $P_0, P_3$, it reconstructs $\vl{z}$ as $\vl{z} = \vl{p}^j \xor \vl{q}^j$.
		\item {\em Reconstruction towards $P_g \in \{P_0, P_3\}$}: $P_1$ sends $\vl{q}^1$ and $\h = \Hash(\key{{\vl{z}}}{\vl{z},1})$ to $P_g$, where $\Hash$ is a collision-resistant hash function. 	
		$P_g$ uses the $\vl{q}^1$ received from $P_1$ for reconstructing $\vl{z}$ as $\vl{z} = \vl{p}^1 \xor \vl{q}^1$ if there exists a $K \in \{\key{{\vl{z}}}{{0,1}}, \key{{\vl{z}}}{{1,1}}\}$ whose least significant bit is $\vl{q}^1$ and $\Hash(K) = \h$. 
	\end{enumerate}
\end{protocolbox}

\paragraph{Optimizations when deployed in mixed framework}
\label{subsec:gbopt}
Working in the preprocessing model enables transfer of the (communication-intensive) GC and generating $\shrG{\cdot}$-shares of the input-independent shares of $\vl{x}$ (i.e. $ {\av{\vl{x}}}, {\pad{\vl{x}}{3}}$) in the preprocessing phase. Thus, the online phase is very light and only requires one round to generate $\shrG{\cdot}$-shares  for the input-dependent data (i.e. ${\mk{\vl{x}}}$). Since evaluation is local, evaluators obtain $\shrG{\cdot}$-sharing of the GC output at the end of $1$ round. 

\paragraph{Achieving fairness and robustness}
To ensure fairness, we require a fair reconstruction protocol which proceeds as follows. As described in \S\ref{sec:fair4pc}, parties first ensure that all parties are alive. 
If so, they proceed similar to the protocol in \boxref{fig:pirec}, except with the following differences. For reconstruction towards evaluators, {\em all} three respective garblers send it the decoding information. The evaluator selects the value appearing in majority for reconstruction. For reconstruction towards garblers $P_0, P_3$, {\em both} the evaluators send the least significant bit of the output key together with its hash to the garbler. The presence of at least one honest evaluator guarantees that both garblers will be on the same page. 

To achieve robustness, the main difference from its fair counterpart is use of a robust $\jsend$ primitive. This guarantees that in the event that a misbehaviour is detected, a $\TTP$ is identified which can take the computation to completion and deliver the output to all. 

\subsection{1 GC Variant}
\label{subsubsec:gcworld1}
The input $\vl{x} = \vl{x}_1 \xor \vl{x}_2 \xor \vl{x}_3$ for this variant consists of the shares, $\vl{x}_1 = \mk{\vl{x}} \xor \pad{\vl{x}}{2}$ and $\vl{x}_2 = \pad{\vl{x}}{3}, \vl{x}_3 = \pad{\vl{x}}{1}$, where $\mk{\vl{x}}, \pad{\vl{x}}{1}, \pad{\vl{x}}{2}, \pad{\vl{x}}{3}$ are as defined in $\shrB{\vl{x}}$. While keys for the GC are sampled by all three garblers $P_0, P_2, P_3$, it suffices for only $P_0, P_3$ to generate and $\jsend$ the GC to evaluator $P_1$, and $P_2$ assists only in the key transfer. Elaborately, the common input $\vl{x}_3$ held by $P_0, P_3$ is hard-coded in the circuit before being garbled by them. This necessitates a key transfer only for inputs $\vl{x}_1$ and $\vl{x}_2$. Garblers $P_0, P_2, P_3$ generate keys for the inputs following a similar procedure as in the 2GC variant. Then, $P_2, P_3$ $\jsend$ the key for $\vl{x}_1$ to $P_1$ while garblers $P_0, P_2$ $\jsend$ the key for $\vl{x}_2$. 

The evaluation and output phases are similar to the 2GC variant except that now there exists only a single garbling instance. Looking ahead, in the mixed protocol framework, the output has to be reconstructed towards $P_1, P_2$. Reconstruction towards $P_1$ does not incur additional rounds since sending of decoding information can be overlapped with key transfer. However, unlike in the 2GC variant where reconstruction towards $P_2$ can be done similar to reconstruction towards $P_1$, in the 1GC variant an additional round is required as $P_2$ is no longer an evaluator. This incurs one extra round as opposed to the 2GC variant.

\paragraph{Achieving fairness}
To ensure fair reconstruction~(\S\ref{sec:fair4pc}), parties first perform an aliveness check. Following this, they proceed towards fair reconstruction of $\vl{z}$ from $\shrG{\vl{z}}$ as follows. First, reconstruction of $\vl{z}$ is carried out towards the garblers $P_g \in \PlSet{1}$. For this, $P_1$ sends $\vl{q}$ (least significant bit of \key{{\vl{z}}}{\vl{z}}) and $\h = \Hash(\key{{\vl{z}}}{\vl{z}})$ to $P_g$ as before. Now, if a garbler received a consistent $(\vl{q}, \h)$ pair from $P_1$ such that there exists a $K \in \{\key{{\vl{z}}}{{0}}, \key{{\vl{z}}}{{1}}\}$ whose least significant bit is $\vl{q}$ and $\Hash(K) = \h$, then it uses $\vl{q}$ for reconstructing $\vl{z}$, and sends $\vl{z}$ to its co-garblers. Else, a garbler accepts a $\vl{z}$ received from a co-garbler as the output. Thus, further dissemination of the output by garblers ensures that all parties are on the same page. 
If garblers receive the output, reconstruction of $\vl{z}$ is carried out towards $P_1$. For this, all garblers (who received the output) send the decoding information to $P_1$ who selects the majority value to reconstruct $\vl{z}$. 

\paragraph{Achieving robustness}
To attain robustness, we list below the differences from the fair protocol that have to be carried out. The first difference is use of a robust variant of $\jsend$. Second, in input sharing protocol, where $\vl{x}_1$ is held by only garbler $P_0$, a corrupt $P_0$ may refrain from providing $P_1$ with the correct key (sent as the opening information for the commitment). To ensure robustness, in the event that $P_1$ fails to receive the correct key from $P_0$, we let $P_1$ complain to all parties about this inconsistency by sending an inconsistency bit. All parties exchange this inconsistency bit among themselves, and agree on the majority value. If all parties agree on the presence of an inconsistency, then $P_0, P_1$ are identified to be in conflict and $\TTP = P_2$ is set to carry out the rest of the computation. 
Finally, to ensure a robust reconstruction, the following approach is taken. Observe that the fair reconstruction provides robustness as long as evaluator $P_1$ is honest. In the event when none of the garblers obtain the output in the fair protocol, it is guaranteed that evaluator $P_1$ is corrupt. Thus, in such a scenario, all parties take $P_1$ to be corrupt, and proceed with $P_0$ as the $\TTP$.

\section{Mixed Framework}
\label{app:mixed}

\paragraph{Arithmetic to Boolean Conversion}
The protocol for arithmetic to boolean conversion appears in \boxref{fig:piab}.

\begin{protocolbox}{$\piab$}{Arithmetic to Boolean Conversion}{fig:piab}
	\justify
	\algoHead{Preprocessing:} $P_0, P_3$ execute joint boolean sharing to generate $\shrB{\vl{v}_2}$,  where $\vl{v}_2 =-(\pad{\vl{v}}{1} + \pad{\vl{v}}{2})$.
	\justify
	\vspace{-2mm}
	\algoHead{Online:}
	\begin{enumerate}[itemsep=0mm]
		\item $P_1, P_2$ execute joint boolean sharing to generate $\shrB{\vl{v}_1}$, where $\vl{v}_1 = \mk{\vl{v}} - \pad{\vl{v}}{3}$.
		\item Parties obtain $\shrB{\vl{v}} = \shrB{\vl{v}_1} + \shrB{\vl{v}_2}$ using addition circuit.
	\end{enumerate}
\end{protocolbox}

\paragraph{Boolean to Arithmetic Conversion}
The protocol for arithmetic to boolean conversion appears in \boxref{fig:piba}. We remark that the protocol $\piba$ can be used to efficiently generate edaBits~\cite{C:EGKRS20} in our setting. For this, the parties non-interactively generate the boolean sharing for $\ell$-bits and perform the $\piba$ conversion to obtain the equivalent arithmetic value.

\begin{protocolbox}{$\piba(\Partyset, \shrB{\vl{v}})$}{Boolean to Arithmetic Conversion}{fig:piba}
	Let $\vl{v}_i$ denote the $i$th bit of $\vl{v}$. Let ${\pd{\vl{v}}}_i = \pad{\vl{v}_i}{1} \xor \pad{\vl{v}_i}{2} \xor \pad{\vl{v}_i}{3}$, $\vl{p}_i = \arval{({{\mk{\vl{v}}}}_i)}$, and $\vl{q} = \arval{({{\pd{\vl{v}}}}_i)}$ \\
	\justify 
	\vspace{-2mm}
	\algoHead{Preprocessing:} 
	\begin{enumerate}[itemsep=0mm]
		\item For $i \in \{0, 1, \ldots, \ell-1 \} $, parties execute the preprocessing of $\prot{\bitA}$ (\boxref{fig:piBitA}) for each bit $\vl{v}_i$ of $\vl{v}$, to generate $\sgr{{\vl{q}_i}} = (\vl{q}_i^1, \vl{q}_i^2, \vl{q}_i^3)$.
	\end{enumerate}
	\justify
	\vspace{-2mm}
	\algoHead{Online:} Let $\vl{y}_i = \arval{\vl{v}_i}$ and $\vl{y}$ denotes the arithmetic equivalent of $\vl{v}$.
	\begin{enumerate}
		\item Parties locally compute the following:
		\begin{align*}
			P_1, P_3: \vl{y}^1 = \sum_{i=0}^{\ell-1} 2^i \vl{y}_i^1  &=  \sum_{i=0}^{\ell-1} 2^i (\vl{p}_i + \vl{q}_i^1 (1 - 2\vl{p}_i)) \\
			P_2, P_3: \vl{y}^2 = \sum_{i=0}^{\ell-1} 2^i \vl{y}_i^2 &=  \sum_{i=0}^{\ell-1} 2^i (\vl{q}_i^2 (1 - 2\vl{p}_i))  \\
			P_1, P_2: \vl{y}^3 = \sum_{i=0}^{\ell-1} 2^i \vl{y}_i^3  &=  \sum_{i=0}^{\ell-1} 2^i (\vl{q}_i^3 (1 - 2\vl{p}_i)) 
		\end{align*}
		\item $(P_1, P_3), (P_2, P_3), (P_1, P_2)$ execute $\prot{\JSh}$ on $\vl{y}^1, \vl{y}^2, \vl{y}^3$ to generate the respective $\shr{\cdot}$-shares.
		\item Parties locally compute $\shr{\vl{y}} = \shr{\vl{y}^1} + \shr{\vl{y}^2} + \shr{\vl{y}^3}$.
	\end{enumerate}     
\end{protocolbox}

\paragraph{End-to-end Conversions}

\tabref{4PCConvM2}, \tabref{4PCConvM1} compare our sharing conversions with Trident~\cite{NDSS:ChaRacSur20}. 
The cost for the 2GC variant of Trident is computed by incorporating a parallel execution, where $P_3$ is additionally made an evaluator together with $P_0$.
For uniformity, we consider a function, {\sf F}, to be computed on an $\ell$-bit inputs $\vl{x}, \vl{y}$ using a garbled circuit (GC) in the mixed framework, which gives an $\ell$-bit output $\vl{z} = \mathsf{F(\vl{x}, \vl{y})}$, where $\ell$ denotes the ring size in bits. Let  $\Grb{F}$  denote the corresponding GC. In the table, $\Grb{S2}$ denotes a 2-input garbled subtraction circuit; $\GrbD{}$ denotes the garbled circuit with decoding information; $\Grb{n_1\times1,\ldots,n_m \times m}$ denotes ${\sf n_i}$ instances of GC $\Grb{i}$ for $i \in \{1,\dots,{\sf m}\}$ and $\Size{\Grb{n_1\times1,\ldots,n_m \times m}}$ denotes the collective size.

\begin{table}[htb!]
	\centering
	\resizebox{0.48\textwidth}{!}{
		\begin{NiceTabular}{r|r|r|r|r}
			\toprule 
			Protocol\tabularnote{A: arithmetic, B: boolean, G: garbled} 
			& Reference 
			& \Block[r]{}{Communication\tabularnote{Notations: $\ell$ - size of ring in bits, $\kappa$ - computational security parameter.} \\ (Preprocessing)}
			& \Block[r]{}{Rounds \\ (Online)}
			& \Block[r]{}{Communication \\ (Online)} \\
			\midrule
			\Block[r]{2-1}{A-G-A}
			& Trident          & \Block[r]{2-1}{$2\Size{\GrbD{2 \times S2,F}} + 6\ell \kappa + \ell$}
			& $2$ & $4\ell \kappa + 2\ell$\\ 
			& \this & 
			& $1$ & $4\ell \kappa + \ell$\\ 
			\midrule
			\Block[r]{2-1}{A-G-B}
			& Trident          & \Block[r]{2-1}{$2\Size{\Grb{S2,F}} + 6\ell \kappa + \ell$}
			& $2$ & $4\ell \kappa + 2\ell$\\
			& \this & 
			& $1$ & $4\ell \kappa + \ell$\\
			\midrule
			\Block[r]{2-1}{B-G-A}
			& Trident          & \Block[r]{2-1}{$2\Size{\GrbD{S2,F}} + 6\ell\kappa + \ell$}
			& $2$ & $4\ell \kappa + 2\ell$\\
			& \this  & 
			& $1$ & $4\ell \kappa + \ell$\\
			\midrule
			\Block[r]{2-1}{B-G-B}
			& Trident          & \Block[r]{2-1}{$2\Size{\Grb{F}} + 6\ell\kappa + \ell$}
			& $2$ & $4\ell \kappa + 2\ell$\\ 
			& \this & 
			& $1$ & $4\ell \kappa + \ell$\\ 
			\bottomrule
		\end{NiceTabular}
	}
	\caption{\small Conversions (2GC variant):  Trident~\cite{NDSS:ChaRacSur20} and $\this$.}\label{tab:4PCConvM2}
\end{table}

\begin{table}[htb!]
	\centering
	\resizebox{0.48\textwidth}{!}{
		\begin{NiceTabular}{r|r|r|r|r}
			\toprule 
			Protocol\tabularnote{A: arithmetic, B: boolean, G: garbled} 
			& Reference 
			& \Block[r]{}{Communication\tabularnote{Notations: $\ell$ - size of ring in bits, $\kappa$ - computational security parameter.} \\ (Preprocessing)}
			& \Block[r]{}{Rounds \\ (Online)}
			& \Block[r]{}{Communication \\ (Online)} \\
			\midrule
			\Block[r]{2-1}{A-G-A}
			& Trident          & \Block[r]{2-1}{$\Size{\GrbD{2 \times S2,F}} + 3\ell \kappa + \ell$}
			& $2$ & $2\ell \kappa + 3\ell$\\ 
			& \this & 
			& $2$ & $2\ell \kappa + 2\ell$\\ 
			\midrule
			\Block[r]{2-1}{A-G-B}
			& Trident          & \Block[r]{2-1}{$\Size{\Grb{S2,F}} + 3\ell \kappa + \ell$}
			& $2$ & $2\ell \kappa + 3\ell$\\
			& \this & 
			& $2$ & $2\ell \kappa + 2\ell$\\
			\midrule
			\Block[r]{2-1}{B-G-A}
			& Trident          & \Block[r]{2-1}{$\Size{\GrbD{S2,F}} + 3\ell\kappa + \ell$}
			& $2$ & $2\ell \kappa + 3\ell$\\
			& \this  & 
			& $2$ & $2\ell \kappa + 2\ell$\\
			\midrule
			\Block[r]{2-1}{B-G-B}
			& Trident          & \Block[r]{2-1}{$\Size{\Grb{F}} + 3\ell\kappa + \ell$}
			& $2$ & $2\ell \kappa + 3\ell$\\ 
			& \this & 
			& $2$ & $2\ell \kappa + 2\ell$\\ 
			\bottomrule
		\end{NiceTabular}
	}
	\caption{\small Conversions (1GC variant):  Trident~\cite{NDSS:ChaRacSur20} and $\this$.}\label{tab:4PCConvM1}
	\vspace{-4mm}
\end{table}

\newpage

\section{ML Algorithms}
\label{app:implementation}

\paragraph{Training and Inference of NN}
An NN can be divided into various layers, where each layer contains a predefined number of nodes. These nodes are a linear function composed of a non-linear "activation" function. 
The nodes at the input layer are evaluated on the input features to evaluate a neural network. The outputs from these nodes are fed as inputs to the nodes in the next layer. This process is repeated for all the layers to obtain the output. The underlying operation involved is a computation of activation matrices for all the layers. This constitutes the forward propagation phase. 
The backward propagation involves adjusting model parameters according to the difference in the computed output and the actual output and comprises computing error matrices. 

Concretely, each layer comprises matrix multiplications followed by an application of the ReLU function. The maxpool layer additionally follows convolutional layers after the ReLU layer. After evaluating the layers in a sequential manner, at the output layer, we use the MPC friendly variant of the softmax activation function, $\sftmx(u_i) = \frac{\relu(u_i)}{\sum_{j = 1}^{\vl{n}} \relu(u_j)}$, proposed by SecureML~\cite{SP:MohZha17}. To perform the division, we switch from arithmetic to garbled world and then use a division garbled circuit~\cite{FCW:PulSii15} followed by a switch back to the arithmetic world. 
For training, we use Gradient Descent, where the forward propagation comprises computing activation matrices for all the layers in the network.  The backward propagation comprises computing error matrices involving matrix multiplications with derivative of maxpool and derivative of $\relu$, depending on the network architecture. We refer readers to~\cite{SP:MohZha17,CCS:MohRin18,NDSS:PatSur20,NDSS:ChaRacSur20,PoPETS:WTBKMR21} for formal details.

\paragraph{Inference of SVM}
SVM is a function which takes as input an $n$-dimensional {\em feature vector}, $\vct{x}$, and outputs the {\em category} to which the feature vector belongs. SVM is implemented as a matrix $\Mat{F}$, of dimension $q \times n$ where each row of $\Mat{F}$ is called the support vector and a vector $\vct{b} = (b_1, \ldots, b_q)$, is called the {\em bias}. Each element of $\Mat{F}$ and $\vct{b}$ lies in $\Z{\ell}$. Each support vector along with a scalar from the bias can classify the input $\vct{x}$ into a specific category. More precisely, let $\Mat{F}_i$ denote the $i^{\text{th}}$ row of matrix $\Mat{F}$. Then, the value $\Mat{F}_i \cdot \vct{x} + b_i$ specifies how likely $\vct{x}$ is to be in category $i$. To find the most likely category, we compute argmax over these values, i.e. $\text{category}(\vct{x}) = \text{argmax}_{i \in \{1, \ldots, q\}} \Mat{F}_i \cdot \vct{x} + b_i$. We refer the readers to~\cite{SP:DEFKSV19} for more details. 

\section{Security proofs}
\label{app:security}
Without loss of generality, we prove the security of our robust framework. The case for fairness follows similarly, and we omit its details. We provide proofs in the $\FSETUP, \Func[\jsend]$-hybrid model, where $\FSETUP$~(\boxref{fig:FSETUP}), $\Func[\jsend]$~(\boxref{fig:JsendFunc}) denote the ideal functionality for the shared-key setup and $\jsend$, respectively. 

The strategy for simulating the computation of function $f$ (represented by a circuit $\Ckt$) is as follows: Simulation begins with the simulator emulating the shared-key setup~($\FSETUP$) functionality and giving the respective keys to the adversary. This is followed by the input sharing phase in which $\Sim$ computes the input of $\Adv$, using the known keys, and sets the inputs of the honest parties, to be used in the simulation, to $0$. $\Sim$ invokes the ideal functionality $\Func[Robust]$ on behalf of $\Adv$ using the extracted input and obtains the output $\vl{y}$. $\Sim$ now knows the inputs of $\Adv$ and can compute all the intermediate values for each of the building blocks. $\Sim$ proceeds with simulating each of the building blocks in the topological order.

For modularity, we provide the simulation steps for each building block (arithmetic/garbled) separately. Carrying out these blocks in the topological order yields the simulation for the entire computation. If a $\TTP$ is identified during the simulation, the simulator stops and returns the function output to the adversary on behalf of the $\TTP$ as per $\Func[\jsend]$.

\paragraph{Ideal $\jsend$ Functionality}
The ideal $\jsend$ functionality for fairness security appears in \boxref{fig:JsendFFunc} and that for the robust setting appears in \boxref{fig:JsendFunc}.

\vspace{-1mm}
\begin{systembox}{$\Func[\jsend]$~(for fair security)}{Ideal functionality for $\jsend$ in $\this$}{fig:JsendFFunc}
	\justify
	$\Func[\jsend]$ interacts with the parties in $\Partyset$ and the adversary $\Sim$. 
	\begin{myitemize}
		\item[\bf Step 1:] $\Func[\jsend]$ receives $(\INPUT,\vl{v}_s)$ from senders $P_s$ for $s \in \{i,j\}$, $(\INPUT,\bot)$ from receiver $P_k$ and fourth party $P_l$. While sending the inputs, the adversary is also allowed to send a special $\abort$ command.
		\item[\bf Step 2:] Set $\msg_i = \msg_j = \msg_l = \bot$.
		\item[\bf Step 3:] If $\vl{v}_i = \vl{v}_j$, set $\msg_k = \vl{v}_i$. Else, set $\msg_k = \abort$.
		\item[\bf Step 4:] Send $(\OUTPUT, \msg_s)$ to $P_s$ for $s \in \{0,1,2, 3\}$.
	\end{myitemize}
\end{systembox}

\vspace{-2mm}
\begin{systembox}{$\Func[\jsend]$~(for robust security)}{Ideal functionality for robust $\jsend$~\cite{USENIX:KPPS21}}{fig:JsendFunc}
	\justify
	$\Func[\jsend]$ interacts with the parties in $\Partyset$ and the adversary $\Sim$. 
	\begin{myitemize}
		\item[\bf Step 1:] $\Func[\jsend]$ receives $(\INPUT,\vl{v}_s)$ from senders $P_s$ for $s \in \{i,j\}$, $(\INPUT,\bot)$ from receiver $P_k$ and fourth party $P_l$, while it receives $(\SELECT,\ttp)$ from $\Sim$. Here $\ttp$ is a boolean value, with a $1$ indicating that $\TTP = P_l$ should be established. 
		\item[\bf Step 2:] If $\vl{v}_i =\vl{v}_j$ and $\ttp = 0$, or  if $\Sim$ has corrupted $P_l$\footnote{This condition is used to capture the fact that a corrupt $P_l$ cannot create an inconsistency in $\Func[\jsend]$ since the parties actively involved in $\Func[\jsend]$ would be honest}, set $\msg_i = \msg_j = \msg_l = \bot, \msg_k = \vl{v}_i$ and go to {\bf Step 4}.
		\item[\bf Step 3:] Else, set $\msg_i = \msg_j = \msg_k = \msg_l = P_l$.
		\item[\bf Step 4:] Send $(\OUTPUT, \msg_s)$ to $P_s$ for $s \in \{0,1,2, 3\}$.
	\end{myitemize}
\end{systembox}
\vspace{-2mm}

\subsection{Arithmetic/Boolean World}
\label{app:absec}
We provide the simulation for the case for corrupt $P_0, P_1$ and $P_3$. The case for corrupt $P_2$ is similar to that of $P_1$. 

\paragraph{Sharing Protocol~($\prot{\Sh}$, \boxref{fig:piSh})}
During the preprocessing, $\Sim_{\prot{\Sh}}^{P_0}$ emulates $\FSETUP$ and gives the respective keys to $\Adv$. The values commonly held with $\Adv$ are sampled using the respective keys, while others are sampled randomly. The details for the online phase are provided next. We omit the simulation for corrupt $P_3$ as it is similar to that of $P_1,P_2$.
\begin{simulatorbox}{$\Sim_{\prot{\Sh}}^{P_0}$}{Simulator $\Sim_{\prot{\Sh}}^{P_0}$ for corrupt $P_0$ }{fig:ShFSim0}
	\justify 
	\algoHead{Online:} \vspace{-2mm}
	\begin{description}
		\item[--] If dealer is $\Adv$, $\Sim_{\prot{\Sh}}^{P_0}$ receives $\mk{\vl{v}}$ from $\Adv$ on behalf of $P_1, P_2, P_3$. If the received values are consistent, $\Sim_{\prot{\Sh}}^{P_0}$ computes $\Adv$'s input $\vl{v}$ as $\vl{v} = \mk{\vl{v}} - \sqr{\pd{\vl{v}}}_1 - \sqr{\pd{\vl{v}}}_2 - \sqr{\pd{\vl{v}}}_3$, else sets $\vl{v}$ as the default value. It invokes $\Func[Robust]$ on $(\INPUT, \vl{v})$ to obtain the function output $\vl{y}$. 
		\item[--] If dealer is $P_1, P_2$ or $P_3$, there is nothing to simulate as $P_0$ doesn't receive any value during the protocol. 
	\end{description}
\end{simulatorbox}
\vspace{-3mm}

\begin{simulatorbox}{$\Sim_{\prot{\Sh}}^{P_1}$}{Simulator $\Sim_{\prot{\Sh}}^{P_1}$ for corrupt $P_1$ }{fig:ShFSim1}
	\justify 
	\algoHead{Online:} \vspace{-2mm}
	\begin{description}
		\item[--] If dealer is $\Adv$, $\Sim_{\prot{\Sh}}^{P_1}$ receives $\mk{\vl{v}}$ from $\Adv$ on behalf of $P_2, P_3$. If the received values are consistent, $\Sim_{\prot{\Sh}}^{P_1}$ computes $\Adv$'s input $\vl{v}$ as $\vl{v} = \mk{\vl{v}} - \sqr{\pd{\vl{v}}}_1 - \sqr{\pd{\vl{v}}}_2 - \sqr{\pd{\vl{v}}}_3$, else sets $\vl{v}$ as the default value. It invokes $\Func[Robust]$ on $(\INPUT, \vl{v})$ to obtain the function output $\vl{y}$. 
		\item[--] If dealer is $P_0, P_2$ or $P_3$, $\Sim_{\prot{\Sh}}^{P_1}$ sets $\vl{v} = 0$ and performs the protocol steps honestly. 
	\end{description}
\end{simulatorbox}
\vspace{-2mm}

Shares unknown to $\Adv$ are sampled randomly in the simulation, whereas in the real protocol, they are sampled using the pseudorandom function (PRF). The indistinguishability of the simulation thus follows by a reduction to the security of the PRF. The same holds for the rest of the blocks.

The simulation for the joint sharing protocol~($\prot{\JSh}$) is similar to that of the sharing protocol. The protocol's design is such that the simulator will always know the value to be sent as part of the joint sharing protocol. The communication is constituted by $\jsend$ calls and is emulated according to the simulation of $\Func[\jsend]$.

\medskip
\paragraph{Multiplication Protocol~($\piMult$)}
~
\medskip
\begin{simulatorbox}{$\Sim_{\piMult}^{P_0}$}{Simulator $\Sim_{\piMult}^{P_0}$ for corrupt $P_0$ }{fig:MulFSim0}
	\justify 
	\algoHead{Preprocessing:} \vspace{-2mm}
	\begin{description}
		\item[--]  Computes $\gm{\vl{a}\vl{b}}{1}, \gm{\vl{a}\vl{b}}{2}$, and $\gm{\vl{a}\vl{b}}{3}$ on behalf of $P_1, P_2, P_3$.
		\item[--] Samples $\vl{u}^1, \vl{u}^2$ using the respective keys with $\Adv$ and computes $\vl{r}$. The joint sharing of $\vl{q}$ is simulated as discussed earlier.
		\item[--] Receives $\vl{w}$ from $\Adv$ on behalf of $P_3$.
		\item[--] Simulating $\piVrfyP$: Joint sharing of $\vl{e_1}, \vl{e_2}, \vl{e}$ is simulated as discussed earlier. The rest of the steps are simulated honestly. This is possible since $\Sim_{\piMult}^{P_0}$ knows the randomness and inputs that should be used by $\Adv$. 
	\end{description}
	\justify 
	\vspace{-2mm}
	\algoHead{Online:}
	$P_0$ has no communication in the online phase except the $\jsend$ instances which are emulated by $\Sim_{\piMult}^{P_0}$.
\end{simulatorbox}
 \vspace{-2mm}

\begin{simulatorbox}{$\Sim_{\piMult}^{P_1}$}{Simulator $\Sim_{\piMult}^{P_1}$ for corrupt $P_1$ }{fig:MulFSim1}
	\justify 
	\algoHead{Preprocessing:}  \vspace{-2mm}
	\begin{description}
		\item[--]  Computes $\gm{\vl{a}\vl{b}}{1}, \gm{\vl{a}\vl{b}}{2}$, and $\gm{\vl{a}\vl{b}}{3}$ on behalf of $P_0, P_2, P_3$.
		\item[--] Samples $\vl{u}^1$ using the respective keys with $\Adv$. Samples a random $\vl{u}^2$ and computes $\vl{r}$. The joint sharing of $\vl{q}$ is simulated as discussed earlier.
		%
		\item[--] Simulate the steps of $\piVrfyP$ honestly. 
	\end{description}
	\justify 
	\vspace{-2mm}
	\algoHead{Online:}  \vspace{-2mm}
	\begin{description}
		\item[--] Computes $\vl{y}_1 + \vl{s}_1, \vl{y}_2 + \vl{s}_2, \vl{y}_3$ honestly. 
		\item[--] Emulates two instances of $\Func[\jsend]$ -- i) $\Adv$ as sender to send $\vl{y}_1 + \vl{s}_1$ to $P_2$, and ii) $\Adv$ as receiver to obtain $\vl{y}_2 + \vl{s}_2$ from $P_2$. 
		\item[--] Simulates joint sharing as discussed earlier.
	\end{description}
\end{simulatorbox}
 \vspace{-2mm}

\begin{simulatorbox}{$\Sim_{\piMult}^{P_3}$}{Simulator $\Sim_{\piMult}^{P_3}$ for corrupt $P_3$ }{fig:MulFSim3}
	\justify 
	\algoHead{Preprocessing:} \vspace{-2mm}
	\begin{description}
		\item[--]  Computes $\gm{\vl{a}\vl{b}}{1}, \gm{\vl{a}\vl{b}}{2}$, and $\gm{\vl{a}\vl{b}}{3}$ on behalf of $P_0, P_1, P_2$.
		\item[--] Samples $\vl{u}^1, \vl{u}^2$ using the respective keys with $\Adv$ and computes $\vl{r}$. The joint sharing of $\vl{q}$ is simulated as discussed earlier.
		\item[--] Honestly computes and sends $\vl{w}$ to $\Adv$.
		\item[--] Simulate the steps of $\piVrfyP$ honestly. 
	\end{description}
	\justify 
	\vspace{-2mm}
	\algoHead{Online:}  \vspace{-2mm}
	\begin{description}
		\item[--] Computes $\vl{y}_1 + \vl{s}_1, \vl{y}_2 + \vl{s}_2, \vl{y}_3$ honestly. 
		\item[--] Emulates two instances of $\Func[\jsend]$ with $\Adv$ as sender to exchange $\vl{y}_1 + \vl{s}_1, \vl{y}_2 + \vl{s}_2$ among $P_1, P_2$. 
		\item[--] Simulates joint sharing as discussed earlier.
	\end{description}
\end{simulatorbox}
 \vspace{-2mm}

\paragraph{Reconstruction Protocol~($\prot{\Rec}$, \boxref{fig:piRec})}
Using the input of $\Adv$ obtained during simulation of sharing protocol, $\Sim_{\prot{\Rec}}$ invokes $\Func[Robust]$ on behalf of $\Adv$ and obtains the function output $\vl{y}$ in clear. $\Sim_{\prot{\Rec}}$ calculates the missing share of $\Adv$ using $\vl{y}$ and the other shares. The missing share is then communicated to $\Adv$ by emulating the $\Func[\jsend]$ functionality.

\subsection{Security Proof for Garbled World}
\label{app:2gcsec}
In this section, we present the proof of security for our robust GC protocol with 2GCs. The case for 1 GC is similar, and we omit the details. For completeness, we provide the simulation assuming function evaluation entirely through the GC. However, as in the previous section, simulation steps are provided for the different phases separately. Thus, the simulation for the appropriate phase can be used while simulating the entire protocol in the mixed framework.  

The simulation begins with the simulator emulating the shared-key setup~($\FSETUP$) functionality and giving the respective keys to the adversary. This is followed by the input sharing phase in which $\Sim$ computes the input of $\Adv$, using the known keys, and sets the inputs of the honest parties, to be used in the simulation, to $0$. $\Sim$ invokes the ideal functionality $\Func[Robust]$ on behalf of $\Adv$ using the extracted input and obtains the output $\vl{y}$.
$\Sim$ proceeds with simulating the GC computation phase using the output $\vl{y}$ by invoking the privacy simulator for the GC. The reconstruction phase follows this. 
We provide the simulation steps in the following order: 
\begin{description}
	\item[--] Generation of boolean shares for the input.
	\item[--] Transfer of keys and GC to the evaluator.
	\item[--] Output computation.
\end{description}
We give the proof with respect to a corrupt $P_0$ and a corrupt $P_1$. Proofs for corrupt $P_3$ and corrupt $P_2$ follow similar to proof for corrupt $P_0$ and $P_1$, respectively.

\paragraph{Generation of boolean shares for the input}
This simulation proceeds as per the simulation of the boolean world mentioned in \S\ref{app:absec}.

\paragraph{Key, GC transfer and  evaluation}
The simulation for $\pigsh$ coupled with the GC transfer for a corrupt $P_1$ and corrupt $P_0$ are provided here. Cases for corrupt $P_2, P_3$ follow.

\medskip
\begin{simulatorbox}{$\Sim_{\Ev}^{P_0}$}{Simulator $\Sim_{\Ev}^{P_0}$ for corrupt $P_0$ }{fig:gshFSim0}
	\justify 
	\begin{description}
		\item[--] With respect to the $j^{\text{th}}$ garbling instance for $j \in \{1, 2\}$, $\Sim_{\Ev}^{P_0}$ generates the keys $\{\key{{\mk{\vl{x}}}}{\bitb, j}, \key{{\av{\vl{x}}}}{\bitb, j}, \key{ {\lambda_{\vl{x}}^{3}} }{\bitb, j}\}_{\bitb \in \{0, 1\}}$ for each function input $\vl{x}$ and the GC as per the honest execution.
		\item[--] Sends the keys for $\key{{\mk{\vl{x}}}}{\mk{\vl{x}}, j}, \key{{\av{\vl{x}}}}{\av{\vl{x}}, j}$ and $\GC_j$ to $P_j$ for $j \in \{1, 2\}$ by emulating $\Func[\jsend]$ with $\Adv$ as the sender.
	\end{description}
\end{simulatorbox}

\begin{simulatorbox}{$\Sim_{\Ev}^{P_1}$}{Simulator $\Sim_{\Ev}^{P_1}$ for corrupt $P_1$ }{fig:gshFSim1}
	\justify 
	\begin{description}
		\item[--] With respect to the first garbling instance, $\Sim_{\Ev}^{P_1}$ runs $(\GC_1, \textbf{X}_1, d_1) \leftarrow \Sim_{\priv}(\onesec, \Ckt, \vl{y})$ where $\vl{y}$ is obtained via invoking $\Func[Robust]$ on $\Adv$'s input. 
		With respect to the second garbling instance, $\Sim_{\Ev}^{P_1}$ generates the keys $\{\key{{\mk{\vl{x}}}}{\bitb, 2}, \key{{\av{\vl{x}}}}{\bitb, 2}, \key{{\lambda_{\vl{x}}^{3}}}{\bitb, 2}\}_{\bitb \in \{0, 1\}}$ for each function input $\vl{x}$ and $\GC_2$ as per the honest execution.
		\item[--] $\Sim_{\Ev}^{P_1}$ sends the keys for each input $\vl{v}$ to the GC, and $\GC_1$ by emulating $\Func[\jsend]$ with $\Adv$ as the receiver.
		\item[--] $\Sim_{\Ev}^{P_1}$ emulates $\Func[\jsend]$ together with $\Adv$ as the sender to send $\key{{\mk{\vl{x}}}}{\mk{\vl{x}}, 2}, \key{{\lambda_{\vl{x}}^{3}}}{\lambda_{\vl{x}}^{3}, 2}$ to $P_2$. 
	\end{description}
\end{simulatorbox}

\paragraph{Output computation}

\begin{simulatorbox}{$\Sim_{\Rec}^{P_0}$}{Simulator $\Sim_{\Rec}^{P_0}$ for corrupt $P_0$}{fig:srec0}
	\medskip
	\justify
	\begin{description}
		\item[--] Let $\lsb(\vl{v})$ denote the least significant bit of $\vl{v}$.
		\item[--] $\Sim_{\Rec}^{P_0}$ sends $\vl{q}^J = \vl{y} \xor \lsb(\key{{\vl{y}}}{0, j})$ and $\h^j = \Hash(\key{}{})$ to $\Adv$ on behalf of honest $P_j \in \SetE$ such that $\key{}{} \in \{\key{\vl{y}}{0, j}, \key{\vl{y}}{1, j}\}$ and $\vl{q}^j = \lsb({\key{}{}})$, where $\vl{y}$ is obtained via invoking $\Func[Robust]$.
	\end{description}
\end{simulatorbox}

\begin{simulatorbox}{$\Sim_{\Rec}^{P_1}$}{Simulator $\Sim_{\Rec}^{P_1}$ for corrupt $P_1$}{fig:srec1}
	\medskip
	\justify
	\begin{description}
		\item[--] Let $\lsb(\vl{v})$ denote the least significant bit of $\vl{v}$.
		\item[--] $\Sim_{\Rec}^{P_1}$ sends $\vl{p}^1 = \lsb(\key{{\vl{y}}}{0, 1})$ to $\Adv$ on behalf of honest garblers in $\PlSet{1}$ where $\vl{y}$ is obtained via invoking $\Func[Robust]$.
	\end{description}
\end{simulatorbox}

\paragraph{Indistinguishability argument}
We argue that $\Ideal_{\calF, \Sim_{\mathrm{\Pi}}} \overset{c}{\approx} \Real_{\mathrm{\Pi}, \Adv}$ when $\Adv$ corrupts $P_1$ based on the following series of intermediate hybrids.

$\Hyb_0$: Same as $ \Real_{\mathrm{\Pi}, \Adv}$.

$\Hyb_1$: Same as $\Hyb_0$, except that $P_0$, $P_2, P_3$ use uniform randomness instead of pseudo-randomness to sample values not known to $P_1$.

$\Hyb_2$: Same as $\Hyb_1$ except that $\GC_1$ is created as $(\GC_1, \textbf{X}_1, d_1) \leftarrow \Sim_\prv(1^\kappa, \Ckt, \vl{y})$.

Since $\Hyb_2 := \Ideal_{\calF, \Sim_{\mathrm{\Pi}}}$, to conclude the proof we show that every two consecutive hybrids are indistinguishable.	

$\Hyb_0 \overset{c}{\approx} \Hyb_1$ :  The difference between the hybrids is that $P_0, P_2, P_3$ use uniform randomness in $\Hyb_1$ rather than pseudo-randomness as in $\Hyb_0$ (for sampling $\sqr{\av{}}_2$). The indistinguishability follows via reduction to the security of the PRF.

$\Hyb_1 \stackrel{c}{\approx} \Hyb_2$: The difference between the hybrids is in the way $(\GC_1, \textbf{X}_1, d_1)$ is generated. In $\Hyb_1$, $(\GC_1, e_1, d_1) \leftarrow \Gb(1^\kappa, \Ckt)$ is run. In $\Hyb_2$, it is generated as $(\GC_1, \textbf{X}_1, d_1) \leftarrow \Sim_\prv(1^\kappa, \Ckt, \vl{y})$. Indistinguishability follows via reduction to the privacy of the garbling scheme. 


We argue that $\Ideal_{\calF, \Sim_{\mathrm{\Pi}}} \overset{c}{\approx} \Real_{\mathrm{\Pi}, \Adv}$ when $\Adv$ corrupts $P_0$ based on the following series of intermediate hybrids.

$\Hyb_0$: Same as $ \Real_{\mathrm{\Pi}, \Adv}$.

$\Hyb_1$: Same as $\Hyb_0$, except that $P_1$, $P_2, P_3$ use uniform randomness instead of pseudo-randomness to sample values not known to $P_0$.

$\Hyb_2$: Same as $\Hyb_1$ except that hash of the key $\key{}{}$ where $\key{}{} \in \{\key{\vl{y}}{0, j}, \key{\vl{y}}{1, j}\}$ to be sent to $\Adv$ is computed such that $\lsb(\key{}{}) \xor \lsb(\key{{\vl{y}}}{0, j}) = \vl{y}$,  for $j \in \{1,2\}$ instead of obtaining it as output of GC evaluation.

Since $\Hyb_2 := \Ideal_{\calF, \Sim_{\mathrm{\Pi}}}$, to conclude the proof we show that every two consecutive hybrids are indistinguishable.	

$\Hyb_0 \overset{c}{\approx} \Hyb_1$ :  The difference between the hybrids is that $P_1, P_2, P_3$ use uniform randomness in $\Hyb_1$ rather than pseudo-randomness as in $\Hyb_0$ (for sampling $\pd{}{3}$). The indistinguishability follows via reduction to the security of the PRF.

$\Hyb_1 \stackrel{c}{\approx} \Hyb_2$: The difference between the hybrids is that in $\Hyb_1$, key $\key{}{}$ where $\key{}{} \in \{\key{\vl{y}}{0, j}, \key{\vl{y}}{1, j}\}$ for $j \in \{1,2\}$ is computed as output of the GC evaluation while in $\Hyb_2$, it is computed such that $\lsb(\key{}{}) \xor \lsb(\key{{\vl{y}}}{0, j}) = \vl{y}$. Due to the correctness of the garbling scheme, the equivalence of $\key{}{}$ computed in both the hybrids holds.

\end{document}